%% file: main.tex
\documentclass[a4paper,UKenglish,cleveref, autoref, thm-restate]{lipics-v2021}
\title{Deciding two-variable logic with global and local Presburger quantification}

 \author{Michael Benedikt}{University of Oxford, UK}{michael.benedikt@cs.ox.ac.uk}{https://orcid.org/0000-0003-2964-0880}{}
 \author{Chia-Hsuan Lu} {University of Oxford, UK}{chia-hsuan.lu@cs.ox.ac.uk}{https://orcid.org/0009-0008-5446-9150}{}
  \author{Alex Perry}{University of Oxford, UK}{}{}{}
 \author{Tony Tan} {University of Liverpool, UK}{tony.tan@liverpool.ac.uk}{https://orcid.org/0009-0005-8341-2004}{}

\input{macroswithoutcomments}
 \nolinenumbers

\authorrunning{Benedikt, Lu, Perry, and Tan} 

\Copyright{Michael Benedikt, Chia-Hsuan Lu, Alex Perry, and Tony Tan} 



\category{} 

\relatedversion{}

\begin{document}

\maketitle

\input{0-abstract}

\newpage  

\input{1-intro}
\input{2-resultstatement}

\input{2-gnn}
\input{3-proofoutline}

\input{4-behaviorchar}

\input{4-realization}

\input{4-gcc}

\input{7-conclusion}

\bibliography{paper}

\newpage
\appendix
\input{appendix}

\end{document}

%% file: macroswithoutcomments.tex
\newcommand{\omitforarxiv}[1]{}
\newcommand{\forarxiv}[1]{#1}
\newcommand{\tildeeta}{\eta\makebox[0pt][r]{\raisebox{-0.3ex}{$\widetilde{\phantom{\eta}}$}}}

\usepackage{amssymb}
\usepackage{amsthm}
\usepackage{stmaryrd}
\usepackage{mathtools}
\usepackage{enumerate}
\usepackage{paralist}
\usepackage{tikz,tkz-graph,tkz-berge}
\usetikzlibrary{petri,topaths}
\usetikzlibrary{arrows.meta}
\usepackage{algorithm}
\usepackage{algpseudocode}

\usepackage{thmtools}
\usepackage{thm-restate}

\usepackage{makecell}
\usepackage{tabularray}

\newcommand{\itemtitle}[1] {{\normalfont \textbf{(#1)}}}
\newcommand{\myparagraph}[1]{\medskip\noindent \textbf{#1.}}

\newcommand{\myeat}[1]{}

\tikzset{
    vtx/.style={
        draw,
        circle,
        line width=.5pt,
        minimum size=5pt,
    }
}

\tikzset{
    newvtx/.style={
        draw,
        circle,
        line width=.5pt,
        minimum size=3pt,
        inner sep = 2pt
    }
}

\tikzset{
    edge/.style={
        line width=1pt,
        black
    }
}
\tikzset{
    bdr/.style={
        rounded corners,
        line width=2pt,
        gray,
        opacity=0.1
    }
}

\newcommand{\transgptwo}     {\tau^\gptwo}
\newcommand{\transgcc}       {\tau^\gcc}
\newcommand{\vocabext}       {\vocab^\mathsf{ext}}
\newcommand{\transstruct}[1] {\tau^{\mathsf{str}}_{ #1 }}
\newcommand{\transgptwop}[1]  {\transgptwo\left[ #1 \right]}
\newcommand{\transgccp}[1]    {\transgcc\left[ #1 \right]}
\newcommand{\vocabextp}[1]    {\vocabext\left[ #1 \right]}
\newcommand{\transstructp}[2] {\transstruct{#1}\left[ #2 \right]}

\newcommand{\reals}{\mathbb{R}}

\newcommand{\asentenceinourlogic}{\Phi}
\newcommand{\univpartphi}{\varphi_\forall}
\newcommand{\quantifierfreepartphi}{\varphi_{\mathsf{qf}}}
\newcommand{\anfgptwosentence}{\varphi_{\mathsf{nf}}}
\newcommand{\aoptwosentence}{\varphi^O}
\newcommand{\anfoptwosentence}{\varphi^O_{\mathsf{nf}}}

\newcommand{\anourlogicsentence}{\varphi}
\newcommand{\cA}{\mathcal{A}}
\newcommand{\cB}{\mathcal{B}}
\newcommand{\cC}{\mathcal{C}}

\newcommand{\cG}{\mathcal{G}}
\newcommand{\cH}{\mathcal{H}}

\newcommand{\cK}{\mathcal{K}}
\newcommand{\cL}{\mathcal{L}}

\newcommand{\cN}{\mathcal{N}}

\newcommand{\cP}{\mathcal{P}}
\newcommand{\cQ}{\mathcal{Q}}
\newcommand{\cS}{\mathcal{S}}

\newcommand{\cV}{\mathcal{V}}

\newcommand{\bbN}{\mathbb{N}}

\newcommand{\bbQ}{\mathbb{Q}}

\newcommand{\bbZ}{\mathbb{Z}}

\newcommand{\bfA}{\mathbf{A}}
\newcommand{\bfC}{\mathbf{C}}
\newcommand{\bfP}{\mathbf{P}}
\newcommand{\bfR}{\mathbf{R}}

\newcommand{\bfb}{\mathbf{b}}
\newcommand{\bfbbar}{\overline{\mathbf{b}}}
\newcommand{\bfc}{\mathbf{c}}

\newcommand{\bfp}{\mathbf{p}}

\newcommand{\bfu}{\mathbf{u}}
\newcommand{\bfv}{\mathbf{v}}

\newcommand{\agcc}{\mathbf{gcc}}
\newcommand{\abehave}{\mathbf{beh}} 

\newcommand{\bfw}{\mathbf{w}}
\newcommand{\bfx}{\mathbf{x}}

\newcommand{\bfzero}{\mathbf{0}}

\newcommand{\frb}{\mathfrak{b}}

\newcommand{\frp}{\mathfrak{p}}

\newcommand{\frx}{\mathfrak{x}}
\newcommand{\frP}{\mathfrak{P}}

\newcommand{\abehav}{\mathsf{beh}}

\newcommand{\gbv}{\mathsf{GBV}}
\newcommand{\agbv}{\mathsf{gbv}}
\newcommand{\lbv}{\mathsf{LBV}}
\newcommand{\albv}{\mathsf{lbv}}
\newcommand{\vocab}{\mathsf{V}}

\newcommand{\oMC}{{\cG\text{-}\sf{MC}}}

\newcommand{\bh}[1]{\mathsf{bh}( #1 )}

\newcommand{\bld}[1]{\boldsymbol{#1}}

\newcommand{\obh}{\mathsf{OutBeh}}

\newcommand{\set}[1]          {\left\{ #1 \right\}}
\newcommand{\setc}[2]         {\left\{ #1 \middle|\  #2 \right\}}
\newcommand{\tuple}[1]        {\left\langle #1 \right\rangle}
\newcommand{\intsinterval}[1] {{\left[ #1 \right]}}
\newcommand{\intinterval}[2]  {{\left[ #1 \mathinner {,}\nobreak #2 \right]}}

\newcommand{\abs}[1]          {\left| #1 \right|}
\newcommand{\sol}[1]          {\left\llbracket #1 \right\rrbracket}

\newcommand{\np}     {\mathsf{NP}}
\newcommand{\nexp}   {\mathsf{NEXP}}

\newcommand{\astructure}{\frak M}
\newcommand{\alogicstructure}{\cA}
\newcommand{\FOtwo} {\mathsf{FO^2}}
\newcommand{\fotwo}{\FOtwo}
\newcommand{\Ctwo}  {\mathsf{C^2}}
\newcommand{\ctwo}  {\Ctwo}
\newcommand{\Ptwo}  {\mathsf{P^2}}
\newcommand{\ptwo}  {\Ptwo}

\newcommand{\GCtwo} {\mathsf{GC^2}}
\newcommand{\GPtwo} {\mathsf{GP^2}}
\newcommand{\gptwo} {\GPtwo}

\newcommand{\GCC}   {\mathsf{GCC}}
\newcommand{\gcc}{\GCC}

\newcommand{\ourlogic}{\mathsf{GandLP^2}}

\newcommand{\optwo}{\mathsf{OutP^2}}
\newcommand{\ounmixedgptwo}{\mathsf{OutGandLP^2}}
\newcommand{\outonlyintegersystemsat}{\mathsf{OutLinSysSat}}
\newcommand{\krlbv}{\mathsf{KRLBV}}
\newcommand{\akrlbv}{\mathsf{krlbv}}
\newcommand{\relu}{\mathsf{ReLU}}
\newcommand{\trrelu}{\mathsf{TrReLU}}
\newcommand{\presby}[1] {{\#_y\left[ #1 \right]}}
\newcommand{\presbx}[1] {{\#_x\left[ #1 \right]}}
\newcommand{\termg}[1]  {{\#\left[ #1 \right]}}

\newcommand{\onetypeof}      {\mathsf{1Tp}}
\newcommand{\twotypeof}      {\mathsf{2Tp}}
\newcommand{\onetypes}       {\mathsf{OneTps}}
\newcommand{\twotypes}       {\mathsf{TwoTps}}
\newcommand{\twotypespi}[2]  {\twotypescmp{}{#1}{#2}}
\newcommand{\onetypescmp}[1] {\mathsf{OneTps^{ #1 }}}
\newcommand{\twotypescmp}[3] {\mathsf{TwoTps^{ #1 }_{ #2, #3 }}}
\newcommand{\twotypescmpi}[2]  {\twotypescmp{}{#1}{#2}}
\newcommand{\dual}[1]        {\widetilde{ #1 }}

\newcommand{\onetypeking}[2]    {\mathsf{KingOneTps}(#1, #2)}
\newcommand{\onetypepeasant}[2] {\mathsf{PeasantOneTps}(#1, #2)}

\newcommand{\villagesystem}{\mathsf{VillSys}}

\newcommand{\homo}        {\mathsf{homog}}
\newcommand{\basetext}    {\mathrm{base}}
\newcommand{\periodtext}  {\mathrm{period}}
\newcommand{\kingtext}    {\mathrm{king}}
\newcommand{\prekingtext} {\mathrm{pre\text{-}king}}

\newcommand{\shadowtext}  {\mathrm{shadow}}

\newcommand{\bhset}[1]        {\cB^{ #1 }}
\newcommand{\bhsethomo}[1]    {\cB^{ #1 }_\homo}
\newcommand{\bhP}[1]          {\bfP^{ #1 }}
\newcommand{\bhbase}[2]       {\bfb^{ #1 }_{ #2 }}
\newcommand{\bhperiod}[2]     {\bfp^{ #1 }_{ #2 }}
\newcommand{\numbhbases}[1]   {{\ell^{ #1 }_\basetext}}
\newcommand{\numbhperiods}[1] {{\ell^{ #1 }_\periodtext}}
\newcommand{\numbhxs}[1]      {{\ell^{ #1 }_x}}

\newcommand{\eval}[1]         {\left\llbracket #1 \right\rrbracket}

\newcommand{\twotypetoking}[1]{{#1}^{th}\mathsf{\text{-}2\text{-}type}}
\newcommand{\rep}             {\mathsf{witness}}

\newcommand{\overapprox}[1]{ {\mathsf{OverApprox}}_{#1}}
\newcommand{\bhg}[2]{\mathsf{bh}_{ #1 }( #2 )}

\newcommand{\numouterperiods}{{\textbf{\#outerperiods}}}
\newcommand{\pibar}{\overline{\pi}}

\newcommand{\phan}   {\mathsf{phan}}
\newcommand{\match}  {\mathsf{MATCH}}
\newcommand{\structcon} {\mathsf{STRUCT}}
\newcommand{\sumcon}    {\mathsf{SUM}}
\newcommand{\liftgcc} {\mathsf{LIFTGCC}}

\newcommand{\projectionbase}[1] {\abs{ #1 }_\basetext}

\newcommand{\projectionbasepi}[2] {\abs{ #1 }_{\basetext, { #2 }}}
\newcommand{\transgptwoex}   {\tau^\gptwo}
\newcommand{\transgptwoexp}[1]{\transgptwoex\left[ #1 \right]}

\newcommand{\gnndim}[1] {{d^{\left( #1 \right)}}}

\newcommand{\coefA}[2]  {\bfA^{\left( #1 \right)}_{ #2 }}
\newcommand{\coefC}[1]  {\bfC^{\left( #1 \right)}}
\newcommand{\coefR}[1]  {\bfR^{\left( #1 \right)}}
\newcommand{\coefb}[1]  {\bfb^{\left( #1 \right)}}

\newcommand{\feat}[1]   {\xi^{\left( #1 \right)}}
\newcommand{\bfwl}[1]    {\bfw^{\left( #1 \right)}}

\newcommand{\inc}       {\text{in}}
\newcommand{\out}       {\text{out}}
\newcommand{\nbr}[1]    {\cN_{ #1 }}

\newcommand{\michael}[1]{  }
\newcommand{\ch}[1]     {  }
\newcommand{\tony}[1]   {  }

%% file: 0-abstract.tex
\begin{abstract}
We look at a two-variable logic over finite graphs that allows Presburger quantifiers: linear arithmetic inequalities
between the number of elements that satisfy subformulas. The prior literature shows that closing two-variable logic under Presburger quantification leads to undecidability, while closing under \emph{local quantification} -- the counts are always over neighbors of a given node -- preserves decidability. Here we show decidability  of satisfiability for an intermediate logic that allows both local counts and global counts, but where the linear constraints do not mix the two. 
The proof requires a significant extension of the prior machinery developed for two-variable logic with counting. 

As a corollary we obtain  decidability of satisfiability for an important family of Graph Neural Networks (GNNs). GNNs are a popular graph learning formalism. For example, in a variation of GNNs considered in the past literature,  local aggregation is used in each layer to create modified values at each vertex of a graph,   and then global aggregation is used in a final layer to obtain  a single numerical value for the graph as a whole. Our logic is motivated by the fact that mixing of local counting and global aggregation in GNNs has been shown to lead to undecidability. Our results here show that if global and local aggregation are separated, one regains decidability.  Decidability of other static analysis problems, like equivalence of GNNs, follows from the decidability of satisfiability. 

\omitforarxiv{Using the same machinery developed to obtain our decidability results, we show that the spectrum (set of cardinalities of finite models) of a formula in this logic is a semilinear set.}
\end{abstract}

%% file: 1-intro.tex

\section{Introduction}\label{sec:intro}

First-order logic is well known to have an undecidable satisfiability problem over finite structures, 
even when we restrict to vocabularies with unary relations and a single binary relation~\cite{classicaldecision}: 
we refer to this as a \emph{graph signature} below. This contrasts with the following two decidability results. 
\emph{Presburger arithmetic}, which uses only additive arithmetic over the integers, is decidable: 
there is a procedure taking a sentence using $+,<$ and constant symbols, and determining whether it holds in the integers. 
\emph{Two-variable logic} ($\fotwo$) is decidable over finite graphs: 
there is a procedure taking a sentence $\phi$ in a graph signature in which every formula has at most two variables, determining whether there is some finite structure for the signature where it holds. 
Various attempts have been made to combine these paradigms. 
For example, BAPA~\cite{qfbapa} is a decidable two-sorted logic extending logic over 
an uninterpreted unary signature with Presburger arithmetic. 
\emph{Two-variable logic with counting} ($\Ctwo$)~\cite{ctwodecide, lidiactwocomplexity, ctwocomplexity} 
is another fragment of first-order logic, decidable over finite structures. 
It extends $\fotwo$ with fixed counting quantification -- e.g.  ``there are at least five $y$ such that $\phi(x,y)$''. 
Guarded two-variable Presburger logic \cite{percentage,twovarpres}, $\gptwo$, 
is a language incomparable in expressiveness with $\ctwo$, 
allowing linear arithmetic comparison between unbounded counts, 
but only counts of the number of neighbors of a given node satisfying a subformula. 
Thus one can express, for example, a property of a node $x$ in a colored graph like 
``$x$ has twice as many green neighbors as it has red neighbors''.
It is easy to define a more general Presburger quantifier ($\ptwo$)
allowing formulas in a free variable $x$ via  linear inequalities between the counts of subformulas $\phi_i(x,y)$. But this language 
has an undecidable satisfiability problem \cite{bartosz-wollic22}. 

Our main contribution is a decidable logic over finite graphs (without self-loops, the default below) intermediate between $\gptwo$ and $\ptwo$, 
allowing both ``local Presburger quantification''  -- 
arithmetic inequalities over the number of neighbors of a given node satisfying a subformula -- 
as well as ``global Presburger quantification'' -- arithmetic inequalities over 
the number of nodes in the entire graph satisfying a formula. 
This logic $\ourlogic$ is the first extension of BAPA and $\gptwo$ known to be decidable over finite structures. 
While we do not have tight complexity bounds for our full logic, 
we consider a variant over directed graphs 
where aggregation is restricted to one direction, namely ``outgoing neighbors'';
for this logic, we show satisfiability is $\nexp$-complete.

As a consequence of our main result, we provide a new decidability result for  graph neural networks (GNNs), 
the dominant formalism for graph learning.
We consider a model originating from~\cite{logicalexpressiveness}: 
\emph{aggregate-combine GNNs with truncated $\relu$ activations and global readout}. 
GNNs can be used to define Boolean queries on nodes in graphs, just as a logical formula with one free variable would. 
They are formed in layers which alternate applying a traditional neural network 
on fixed-size vectors with summation to aggregate an unbounded number of vectors into one. 
``Truncated'' refers to the fact that the neural networks do not use the standard piecewise linear $\relu$ function as building blocks, but rather an analog that is constant going out to infinity in each direction. ``Global readout'' refers to the fact that there are two kinds of aggregation, one over neighbors of a given node, and another over all nodes: analogous to the two quantifications in our logic. 
Several prior works \cite{marcoalogicfor, marcoquantized, fundamental,ourtocl} studied satisfiability for variations of these models. \cite{ourtocl,marcoalogicfor} showed that satisfiability with \emph{only} local aggregation was decidable, while \cite{ourtocl} showed that aggregations that mix global and local in the same layer lead to undecidability.  \cite{quantizedglobal} showed decidability for ``quantized GNNs'', where every operator works on a finite domain. The case where we have local aggregation and also global readout, but not in the same layer, was left open -- even though global readout is generally used only standalone at the final layer. Decidability  of these GNNs with unmixed global readout easily follows from the main result on $\ourlogic$.

To prove our main result, we require new machinery for analyzing models of $\gptwo$, building on and substantially extending the infrastructure developed in prior work on logics over uninterpreted finite graphs with arithmetic \cite{percentage, lidiactwocomplexity, ianrevisited, gctwocomplexity, rudolphcardinality}. One idea from prior work is to describe the structures satisfying a formula using linear inequalities over variables representing various cardinalities of a structure. \cite{uscsl} presented an obstacle to extending this method to mix local and global counting. The \emph{$1$-type spectrum} is the collection of tuples of numbers associated with a $\gptwo$ sentence -- for each finite structure, it includes the vectors of sizes for each Boolean combination of unary atoms. \cite{uscsl} shows that this may not be semilinear. We work around this by not describing a structure directly, but rather an associated \emph{decomposed structure}, in which  the behavior of elements is decomposed into different vector components. 

\omitforarxiv{By reusing proof techniques for the main theorem, we also provide a new result on the \emph{spectrum} for $\gptwo$ sentences. 
The spectrum of a logical sentence is the set of cardinalities of finite models of the sentence.
It has been shown that the spectrum of $\ctwo$ is a semilinear set: definable in Presburger arithmetic~\cite{ianrevisited,spectrumeryktony}.
We show that the same holds for $\gptwo$. This result is surprising, given the negative result on $1$-type spectra from \cite{uscsl} mentioned above.}

%% file: 2-resultstatement.tex

\section{Two-variable Presburger Logic with local and global aggregation, GNNs,  and our main results}\label{sec:results}

\begin{definition}[Linear, semilinear, hybrid linear sets, $\cL(\bfb; \bfP)$, Presburger formulas]\label{def:semilinearpresburger}
A \emph{linear set} is a set of the form 
$\setc{\bfb + i_1 \bfp_1 + \cdots + i_k \bfp_k}{i_1, \ldots, i_k \in \bbN}$,
for vectors of non-negative integers $\bfb, \bfp_1, \ldots, \bfp_k$.
We denote such a set by $\cL(\bfb; \bfP)$ where $\bfP := \set{\bfp_1, \ldots, \bfp_k}$. The pair $\tuple{\bfb, \bfP}$ is called a representation.
The vector $\bfb$ is called the base vector and each vector in $\bfP$ is a period vector.
We let $\cL_\homo(\bfP)$ be $\cL(\bfzero; \bfP)$.
A \emph{semilinear set} is a finite union of linear sets, and a \emph{hybrid linear set} \cite{semilinear} is a semilinear set where each component in the union has the same period vector.
It is known \cite{presburger} that the class of semilinear sets is equivalent in expressiveness to sets defined by
\emph{Presburger formulas}: 
first-order logic with $+,<$, interpreted over the natural numbers.
The conversion between Presburger formulas and semilinear representations is  effective~\cite{ginsburg-spanier}. 
\end{definition}

\myparagraph{Logic with Presburger quantifiers}
Let $\vocab$ be a vocabulary consisting of unary and binary predicates.

\begin{definition}[Structure] \label{def:structure}
A structure for a vocabulary with unary and binary symbols in this work will consist of a nonempty finite domain, along with an assignment of each unary relation to a finite set of elements (sometimes called \emph{nodes} or \emph{vertices}),
and an assignment of each binary relation to a set of pairs of elements,  \emph{without self-loops}: that is, without a pair $\langle a, a \rangle$ being assigned.
\end{definition}

We review logics with Presburger quantifiers, originating with~\cite{percentage}.

\begin{definition}[Local and global Presburger quantifiers]\label{def:localpresburger}
A \emph{local Presburger quantifier} is of the form:
$\cP(x)\ :=\ 
\sum\limits_{t \in \intsinterval{k}} \lambda_t \cdot \presby{G_t(x,y) \wedge \varphi_t(x,y)}\ \circledast\ \delta$,
where $\delta \in \bbZ$;
each $\lambda_t \in \bbZ$;
each $\varphi_t(x,y)$ is a formula with free variables from $\{x,y\}$, 
$G_t(x,y)$ is an atom $R(x,y)$ or $R(y,x)$, and 
$\circledast$ is one of $=$, $\neq$, $\leq$, $\geq$, $<$, or $>$.
Note that $\cP(x)$ has one free variable $x$.
A \emph{global Presburger quantifier} is of the form:
$\cP(x)\ :=\ 
\sum\limits_{t \in \intsinterval{k}} \lambda_t \cdot \presby{\varphi_t(y)}\ \circledast\ \delta$,
where $\lambda_t$, $\circledast$, and $\delta$ are as in the local case, and 
$\varphi_t(x)$ is a formula with $x$ the only free variable.

The semantics of these quantifiers is given inductively,
assuming a semantics for $\varphi_t(x, y)$, and $\varphi_t(x)$.  The semantic function
takes as arguments  a structure and valuation for the free variables. 
{\bf Free variable convention for Global Presburger quantifiers} Although the truth value of $\cP(x)$ is independent of the assignment to $x$, we \emph{still treat it as a formula with one free variable $x$}: the semantics will either assign true to every element of the domain or to none.

For the local quantifier, given a structure $\cG$ and an element $v$ in $\cG$,
we say that \emph{$\cP(x)$ holds in $\cG, v$},
denoted by ${\cG \models \cP(v)}$,
if the following (in)equality holds in $\bbZ$:
\begin{equation*}
\sum\limits_{t \in \intsinterval{k}} \lambda_t \cdot
\abs{\setc{u \in V}{\cG \models G_t(v,u) \wedge \varphi_t(v, u)}}
\ \circledast\ \delta.
\end{equation*}
The semantics for the global quantifier is defined in a similar manner.
Each $\presbx{\varphi_t(x)}$ is interpreted as the number of elements in $\cG$
that satisfy $\varphi_t(x)$.
\end{definition}

Note that \emph{the semantics of all logics considered here is over finite structures}, 
since finiteness guarantees that 
the linear inequalities in Presburger quantifiers are applied to ordinary integers.
Since this will be our default, we say that a sentence $\phi$ is \emph{satisfiable} if 
there is a \emph{finite} structure where it holds, and
similarly for formulas with free variables.

\begin{definition}[$\ourlogic$]
The syntax of \emph{global and local two-variable Presburger logic ($\ourlogic$)} over vocabulary $\vocab$ 
is built up starting with arbitrary relational atoms 
over the vocabulary, with closure under the usual connectives and the following rules for quantifiers.
\begin{compactitem}
\item 
If $\varphi(x)$ is a formula,
then so are $\forall x\ \varphi(x)$ and $\exists x\ \varphi(x)$.
\item 
If $\varphi(x, y)$ is a formula,
then so are $\forall x\ \epsilon(x, y) \to \varphi(x, y)$ and $\exists x\ \epsilon(x, y) \land \varphi(x, y)$, 
where $\epsilon(x, y)$ is an atom involving both variables $x$ and $y$. 
\item 
\itemtitle{Pure local Presburger formula}
If $\set{\varphi_t(x, y)}_{t \in \intsinterval{k}}$ are formulas and
$\set{\epsilon_t(x, y)}_{t \in \intsinterval{k}}$ are atoms involving both $x$ and $y$,
then 
    $\sum\limits_{t \in \intsinterval{k}} \lambda_t \cdot \presby{\epsilon_t(x,y)\land\varphi_t(x, y)}
    \ \circledast\ \delta$
is a formula.
\item \itemtitle{Pure global Presburger formula}
If $\set{\varphi_t(y)}_{t \in \intsinterval{k}}$ is a set of formulas,
then 
    $\sum\limits_{t \in \intsinterval{k}} \lambda_t \cdot \presbx{\varphi_t(x)}
    \ \circledast\ \delta$
is also a formula.
\end{compactitem}
\end{definition}

The inductive definition of free variables of a formula and the semantic function are straightforward, keeping in mind our convention about global Presburger quantifiers having
an implicit free variable.

Note that the standard existential quantification can easily be expressed using Presburger quantifiers, and
thus $\ourlogic$ extends $\GCtwo$, the  guarded two-variable first-order logic with counting.
It also extends 
\emph{the guarded fragment of two-variable logic with Presburger quantifiers} 
($\GPtwo$)~\cite{percentage,uscsl,twovarpres},
the fragment of $\ourlogic$  obtained by excluding pure global Presburger formulas.
For example,
$\exists x \left(\presby{R(x,y)} - 3 \cdot \presby{R(y,x)} = 0\right)$
is a sentence of $\GPtwo$ that holds in a directed graph when 
there is an element that has 3 times as many outgoing edges as incoming edges.
$\ourlogic$ also allows us to express \emph{cardinality constraints on the structure as a whole}.
For example, $\presbx{A(x)}-2\cdot \presbx{B(x)}\ge 0$ is a formula that holds
of a node in a structure when the size of the predicate $A$ is at least twice the size of the predicate $B$.
In fact, up to equisatisfiability, $\GPtwo$ and 
cardinality constraints on the structure are essentially \emph{all} $\ourlogic$ can express: 
see Lemma~\ref{lem:normal}.

Our main result can now be stated:

\begin{theorem}[$\ourlogic$ decidability]\label{thm:ourlogic}
Satisfiability of $\ourlogic$ is decidable.
\end{theorem}

\myparagraph{The modal variant of $\ourlogic$}
We define the natural modal variant of $\ourlogic$: 
Local Presburger quantifiers that aggregate only over \emph{out-neighbors}
$\cP(x)\ :=\ 
\sum\limits_{t \in \intsinterval{k}} \lambda_t \cdot \presby{G_t(x,y) \wedge \varphi_t(y)}\ \circledast\ \delta$,
where $G_t$ now must be an atom $R(x,y)$.
Similarly, we restrict the use of universal and existential quantifiers to the modal case:
If $\varphi(y)$ is a formula,
then so are $\forall y\ R(x, y) \to \varphi(y)$ and $\exists y\ R(x, y) \land \varphi(y)$, 
for some binary relation $R$.

This ``modal'' restriction simplifies decidability considerably.
\begin{theorem}\label{thm:outgoingonlynexptimeinformal} 
Satisfiability of the modal variant of the logic is $\nexp$-complete.
\end{theorem}
In fact, \emph{we get decidability
even allowing mixing of global and local counting within Presburger quantifiers}.
Note that the outgoing-only result is quite similar to results of \cite{marcoalogicfor,nexptimemodalglobalcounting}, and the technique we use is a variation of one introduced for another logic in \cite{daumantas,daumantasgnn}.

%% file: 2-gnn.tex

\myparagraph{GNN analysis} We now apply the results to the analysis of graph neural networks (GNNs).

\begin{definition}[Truncated $\relu$] Truncated $\relu$ ($\trrelu$) is the function from $\reals$ to $\reals $ that is the identity on $[0,1]$, $1$ for $x\geq 1$, and zero otherwise. We extend this function to vectors by applying it pointwise.
\end{definition}

\begin{definition}[Colored graph]
Following prior work (e.g. \cite{ourtocl}),
we consider a GNN model that takes as input $n$-color graphs for some non-negative integer $n$.
These are standard directed  (finite) graphs without self-loops, extended with $n$ unary predicates  (colors) representing the colors.
Formally, it is a graph $\cG = \tuple{V, E, \set{U_c}_{c \in \intsinterval{n}}}$,
where each $U_c\subseteq V$.
We do not require the sets $U_1,\ldots,U_n$ to be pairwise disjoint.
Note that this restricts structures over $\vocab$ to the case
with only one binary relation $E$, thus our logics also apply to colored graphs.
\end{definition}


\begin{definition}[Syntax of GNNs]\label{def:gnn}
A \emph{graph neural network} (GNN) is a tuple\\
$\tuple{
\set{\gnndim{\ell}}_{\ell \in \intinterval{0}{L}},
\set{\coefC{\ell}}_{\ell \in \intsinterval{L}},
\set{\coefA{\ell}{x}}_{\substack{
        \ell \in \intsinterval{L} \\
x \in \set{\out, \inc}}},
\set{\coefR{\ell}}_{\ell \in \intsinterval{L}},
\set{\coefb{\ell}}_{\ell \in \intsinterval{L}}
}$,
where $L \in \bbN^+$ is the number of layers;
each ${\gnndim{\ell} \in \bbN}$, called the \emph{dimension} of the $\ell^{th}$ layer, 
requiring $\gnndim{0} = n$, the number of colors, 
each ${\coefC{\ell}, \coefA{\ell}{x}, \coefR{\ell} \in \bbQ^{\gnndim{\ell} \times \gnndim{\ell-1}}}$, 
the \emph{coefficient matrices} of the $\ell^{th}$ layer; and
$\coefb{\ell} \in \bbQ^{\gnndim{\ell}}$, the \emph{bias vector} of the $\ell^{th}$ layer. We also require  $\gnndim{i}>0$ for each $0<i \leq L$.
\end{definition}


\begin{definition}[Feature of a GNN, computation of GNNs] \label{def:gnncomp}
    For a GNN $\cA$ and a graph $\cG = \tuple{V, E, \set{U_c}_{c \in \intsinterval{n}}}$,
    the \emph{computation of $\cA$ on $\cG$} is a sequence of functions
    $\set{\feat{\ell}_\cG: V \to \bbQ^{\gnndim{\ell}}}_{\ell \in \intinterval{0}{L}}$ defined inductively on $\ell$.
    For $v \in V$, $\feat{\ell}_\cG(v)$ is called the \emph{$\ell$-feature vector of $v$} w.r.t. $\cA$, and 
    $\feat{\ell}_{\cG, i}(v)$ denotes the $i^{th}$ entry of $\feat{\ell}_\cG(v)$.
    For brevity we sometimes omit $\ell$ when it is either clear from context or irrelevant. 
    We also shorten ``feature vector'' to ``feature'' at times. 

    For $\ell = 0$, $\feat{0}_{\cG}(v)$ is the  encoding of the colors of the vertex $v$, with a one at an index $i$ exactly when $U_i(v)$ holds.
    This is also called the \emph{initial feature vector} of $v$.
    Formally, for $1 \le i \le \gnndim{0}$,
    if $v \in U_i$, then $\feat{0}_{\cG, i}(v) = 1$; otherwise, $\feat{0}_{\cG, i}(v) = 0$.
    For $1 \le \ell \le L$,
        $\feat{\ell}_\cG(v) \ :=\  \trrelu{\bfwl{\ell}(v)}$ 
    where
    \begin{equation*}
        \bfwl{\ell}(v) \ :=\ 
        \coefC{\ell} \feat{\ell-1}_\cG(v) + \!\!\!\!\!\!
        \sum_{x \in\set{\out, \inc}}\!\!
        \left(
            \coefA{\ell}{x} \!\!\!\!\!\!
            \sum_{u \in \nbr{x, \cG}(v)} \!\!\!\!\!
            \feat{\ell-1}_\cG(u)
        \right) 
        +
        \coefR{\ell} \sum_{u \in V} \feat{\ell-1}_\cG(u) +
        \coefb{\ell}.
    \end{equation*}
\end{definition}

\begin{definition}[Pure layers and non-mixing GNNs, outgoing-only]
    We say that the $\ell^{th}$ layer of a GNN is \emph{pure global} if both
    $\coefA{\ell}{\out}$ and $\coefA{\ell}{\inc}$ are zero.
    It is \emph{pure local} if $\coefR{\ell}$ is zero.
    A layer is \emph{pure} if it is either pure global or pure local.
    Intuitively, this means that the features in a layer depend either only on global aggregation or only on local neighborhood aggregation, but never on both.
 A graph neural network is \emph{non-mixing} if all of its layers are pure. It is \emph{outgoing-only} if the coefficients of matrices corresponding to $x=\inc$ are all zero. 
\end{definition}

\begin{definition}[Boolean semantics of GNNs]
    For a GNN $\cA$,
    a graph $\cG$,
    and a node $v$ in $\cG$,
    we say that $\cA$ \emph{accepts} the tuple $\tuple{\cG, v}$,
    if $\feat{L}_{\cG, 1}(v) \ge 1/2$.
\end{definition}

With the Boolean semantics, we can define several verification problems for GNNs.
The most basic one is the satisfiability problem.
\begin{definition}[Satisfiability of a GNN]
    For a GNN $\cA$,
    we say that $\cA$ is \emph{satisfiable},
    if there exist a graph $\cG$ and a node $v$ in $\cG$, such that $\cA$ accepts $\tuple{\cG, v}$.
\end{definition}

\begin{corollary}\label{cor:gnn_sat}
    Satisfiability for non-mixing GNNs with $\trrelu$ activations is decidable.
   
\end{corollary}

The derivation of Corollary \ref{cor:gnn_sat} is analogous to Theorem~4.12 in~\cite{ourtocl}, which translates GNNs into equiexpressive logical formulas.
Observe that, when the input GNN is non-mixing, the construction in~\cite{ourtocl} produces a $\ourlogic$ formula.
Hence, the satisfiability problem for non-mixing GNNs reduces to the corresponding satisfiability problem for $\ourlogic$,
which is decidable by Theorem~\ref{thm:ourlogic}.

%% file: 3-proofoutline.tex

\section{High-level proof outline}
\label{sec:outline}

We now give an overview of the proof of Theorem \ref{thm:ourlogic}.

We now give a normal form for the logic. We begin by defining  atomic global cardinality constraints, 
which are the simplest form of global Presburger quantifiers.

\begin{definition}[Global Cardinality Constraint] \label{def:gcc}
    An atomic \emph{global cardinality constraint ($\GCC$)} over $\vocab$ is a formula of the form
    \[
    \sumcon_{i \in \intsinterval{k}} \lambda_i \cdot \termg{U_i(x)} \ge \delta
    \]
    where $\lambda_i, \delta \in \bbZ$ and $U_i$ are unary predicates in $\vocab$.
    A $\GCC$ over $\vocab$ is built from atomic $\GCC$s using logical connectives.
    We say that a $\vocab$-structure $\cA$ satisfies the atomic $\GCC$ if
    the inequality $\sum_{i \in \intsinterval{k}} \lambda_i \cdot \abs{U_i^{\cA}} \ge \delta$ holds.
    The semantics of general $\GCC$s is straightforward.
\end{definition}

Every $\ourlogic$ sentence can be effectively converted into 
a conjunction of a $\gptwo$ sentence in normal form and a $\gcc$.
The proof is standard, so we defer it to the appendix.

{\bf Convention:} for brevity, we often write $\pi(x) \models \phi(x)$ where  $\pi$ is a $1$-type and $\phi$ is quantifier-free, to indicate
that for every loopless finite structure, every element satisfying $\pi$ satisfies $\phi$. It is straightforward to check this effectively.

\begin{restatable}[Normal form]{lemma}{normalform}\label{lem:normal}
Every $\ourlogic$ sentence $\asentenceinourlogic$ can be effectively converted to 
a sentence of the form $\anfgptwosentence \land \agcc$ over an extended vocabulary where 
$\anfgptwosentence$ is a $\GPtwo$ sentence in \emph{normal form}:
\[
\anfgptwosentence
\ :=\ 
\forall x~ \quantifierfreepartphi(x)\ \land\ 
\bigwedge_{i \in \intsinterval{m}} \forall x \forall y\ \left(R_i(x,y) \land x \neq y\right) \to \alpha_i(x, y)
\ \land\ 
\bigwedge_{i \in \intsinterval{n}} \forall x\ U_i(x) \to P_i(x),
\]
with $\quantifierfreepartphi$ is quantifier-free and asserts that $\neg R(x,x)$ holds for each binary relation $R$, each $R_i(x,y)$ a binary atom,
each $\alpha_i(x,y)$ quantifier-free, each $U_i$ a unary atom,
each $P_i(x)$ is a ``lower bound Presburger quantifier'', of the form 
\[
\sum_{t \in \intsinterval{s_i}} \lambda_{i,k} \cdot \presby{(x \neq y) \land R_{i,t}(x, y)}  \ \ge\ \delta_i \] 
with each $R_{i,t}(x,y)$ a binary atom, 
while $\agcc$ is a $\gcc$ defined in Def.~\ref{def:gcc},
such that $\asentenceinourlogic$ has a model of size $k$ if and only if
$\anfgptwosentence \land \agcc$ has a model of size $k$. 
\end{restatable}

In what follows we fix a sentence $\anfgptwosentence\land \agcc$ in the normal form as above, 
and let $\univpartphi(x)$  be the formula obtained from $\anfgptwosentence$ by discarding 
the outer quantifier $\forall x$, i.e.,
$\quantifierfreepartphi(x)\ \land\ 
\bigwedge_{i \in \intsinterval{m}} \forall y\ \left(R_i(x,y) \land x \neq y\right) \to \alpha_i(x, y)
\ \land\ 
\bigwedge_{i \in \intsinterval{n}} U_i(x) \to P_i(x)$.

\begin{definition}[$1$-type, $2$-type,dual $2$-type] \label{def:oneandtwotype}
A \emph{1-type} (resp. \emph{2-type}) (w.r.t. $\vocab$) is a maximal consistent subset of
$\bigcup_{i \in \intsinterval{n}} \set{U_i(x), \neg U_i(x)}$
and \\
$\bigcup_{i \in \intsinterval{m}} \set{R_i(x, y), \neg R_i(x, y), R_i(y, x), \neg R_i(y, x)}$, respectively.

The \emph{null 2-type} is the 2-type in which every atom appears negatively.
For a given 2-type $\eta$, let $\dual{\eta}$ denote the dual of $\eta$, where the variables $x$ and $y$ are swapped.

For a structure $\cG$, the 1-type of an element $v$ is the unique 1-type $\pi$ that $v$ satisfies in $\cG$. 
Similarly, the 2-type of a pair $(u, v)$ of elements, where $u \neq v$, is the unique 2-type that $(u, v)$ satisfies in $\cG$.
Note that if the 2-type of $(u, v)$ is $\eta$, then the 2-type of $(v, u)$ is $\dual{\eta}$.
Further, if the 2-type of $(u,v)$ is the null 2-type, so is the 2-type of $(v,u)$.

An element with 1-type $\pi$ is also called a \emph{$\pi$-element}.
An element $v$ is an \emph{$\eta$-neighbor} of $u$ if the 2-type of $(u,v)$ is $\eta$.
It is called a \emph{$(\pi,\eta)$-neighbor} of $u$ if it is a $\pi$-element and an $\eta$-neighbor of $u$.
\end{definition}

Both 1- and 2-types can be identified with the conjunction  of their constituent formulas. 
When viewed as formulas, we write types as $\pi(x)$ and $\eta(x, y)$.
By default, a $1$-type will be a formula with $x$ as free variable; we write
$\pi(y)$ to denote the same $1$-type but substituting $x$ with $y$ in $\pi(x)$.

\begin{definition}[Asymmetric types, Conformant types] \label{def:conformanttypeasymmetrictype}
A 1-type $\pi$ is \emph{conformant} (w.r.t. $\anfgptwosentence$)
 if $\pi(x) \wedge \neg R(x,x) \models \quantifierfreepartphi(x)$. 

An \emph{asymmetric $2$-type} is a $2$-type $\eta$ which is distinct from its dual $2$-type.

A 2-type $\eta$ is \emph{conformant with 1-types $\pi_1$ and $\pi_2$} (w.r.t. $\anfgptwosentence$),
if it is asymmetric, and for every $i \in \intsinterval{m}$,
\begin{equation*}
    \begin{aligned}
        \pi_1(x) \land \eta(x,y) \land \pi_2(y) &\models (R_i(x,y)\land x\neq y) \to \alpha_i(x,y) \\
        \pi_2(x) \land \dual{\eta}(x,y) \land \pi_1(y) &\models (R_i(x,y)\land x\neq y) \to \alpha_i(x,y).
    \end{aligned}
\end{equation*}
Note that while the null 2-type is conformant with any choice of 1-types $\pi_{1}$ and $\pi_{2}$, 
it does not affect the satisfaction of local Presburger quantifiers, 
which merely count elements adjacent via some binary relation.

Thus, throughout the remainder of this paper, 
by default \emph{we will use the term ``2-types'' to mean ``non-null 2-types'', and explicitly highlight whenever null $2$-types are allowed}.

We let $\onetypescmp{\anfgptwosentence}$ be the set of conformant 1-types,
and $\twotypescmp{\anfgptwosentence}{\pi_1}{\pi_2}$ be 
the set of asymmetric 2-types conformant with $\pi_1$ and $\pi_2$.
Both sets can be constructed effectively from $\anfgptwosentence$ and $\pi_1,\pi_2$.
Since we fix $\anfgptwosentence$, we will omit the superscript $\anfgptwosentence$
and write $\onetypes$ and $\twotypes_{\pi_1,\pi_2}$.
\end{definition}

\begin{definition}[Asymmetric structure, Asymmetric satisfiability] \label{def:asymmetricstructuresat}
An \emph{asymmetric structure} is one where every pair of distinct elements realizes an asymmetric $2$-type as in Definition \ref{def:conformanttypeasymmetrictype} or a null 2-type.
The \emph{asymmetric satisfiability problem} asks whether a sentence in $\ourlogic$ is satisfiable in an asymmetric (finite) structure.
\end{definition}

\begin{proposition} \label{prop:aymmetricisok} For any vocabulary $\vocab$, there is an expansion
to $\vocab'$ so that every $\vocab$ structure can be expanded to an asymmetric $\vocab'$ structure.
In particular, the satisfiability problem for $\ourlogic$ reduces to the asymmetric satisfiability problem.
\end{proposition}
\begin{proof}
Suppose $\astructure$ is a structure over $\vocab$ witnessing $\phi$. 
Let $R$ be a fresh binary predicate not in $\vocab$.
Let $\astructure'$ be the expansion of $\astructure$ to the vocabulary $\vocab \cup \set{R}$ 
in which the predicate $R$ is interpreted as follows:
for every two distinct elements $u,v$ such that $(u,v)$ satisfies some relation symbol,
we pick arbitrarily a pair $(u,v)$ and set $(u,v)$ to be in $R$ and
the reverse $(v,u)$ to be not in $R$.
Then, the structure $\astructure'$ is asymmetric as required.
\end{proof}

By Proposition \ref{prop:aymmetricisok}, \emph{it suffices to show that asymmetric satisfiability is decidable}, and
this is what we will do.

Our goal is to characterize a structure using a vector of integers. 
We start with characterizing a single element in a structure.

\begin{definition}[Behavior]\label{def:behavior}
A \emph{behavior} $\abehav$ is a vector of non-negative integers indexed
by a 2-type and a 1-type.
We write $\abehav(\eta, \pi)$ to denote the entry for 2-type $\eta$ and 1-type $\pi$.
We say that a $\abehav$ is the behavior of an element $v$, 
if for each 2-type $\eta$ and target 1-type $\pi$, 
$\abehav(\eta, \pi)$ is
the number of $(\pi,\eta)$-neighbors of $v$. 
\end{definition}


Observe that whether an element satisfies the $\GPtwo$ formula $\univpartphi(x)$ 
depends only on its 1-type and its behavior.
More precisely, let $\cG_1$ and $\cG_2$ be structures, and 
let $v_1$ be an element in $\cG_1$ and $v_2$ be an element in $\cG_2$.
If $v_1$ and $v_2$ have the same 1-type and the same behavior, then
$\cG_1 \models \univpartphi(v_1)$ if and only if
$\cG_2 \models \univpartphi(v_2)$.
This motivates the following notion of a \emph{conformant behavior}, 
which formalizes when a behavior is compatible with $\univpartphi(x)$.

\begin{definition}[Conformant behavior and $\bhset{\pi}$]\label{def:conf_behavior}
We say that a behavior $\abehav$ is \emph{conformant} with $\pi$ and $\anfgptwosentence$ (or equivalently, with $\pi$ and $\univpartphi(x)$)
if there exists a structure $\cG$ and an element $v$ in $\cG$ with behavior $\abehav$ and 1-type $\pi$
such that
\begin{compactenum}    
    \item
        $\cG \models \univpartphi(v)$;
    \item
        for every element $u$ in $\cG$, $u$ has a conformant 1-type, and for every pair of elements $u_1, u_2$ in $\cG$ with $u_1 \neq u_2$, $(u_1, u_2)$ has conformant 2-type or null 2-type.
\end{compactenum}

We denote this set as $\bhset{\pi}$, the set of \emph{behaviors conformant with} $\pi$ and $\univpartphi(x)$.
\end{definition}

Recall the definition of hybrid linear set, a strengthening of semilinearity from Definition \ref{def:semilinearpresburger}.
The following result is easily shown (see Appendix \ref{app:almostcsl}):
\begin{restatable}{proposition}{almostcsl} \label{prop:conformantsemilinear} For any 1-type $\pi$, $\univpartphi(x)$, the set of conformant behaviors is hybrid linear, and a  hybrid linear representation can be effectively generated from $\pi$ and $\univpartphi(x)$.
\end{restatable}

Accordingly, we define the following bases for conformant behaviors, which can also be computed effectively.

\begin{definition}[Basis for conformant behaviors, $\bhsethomo{\pi}$, $\bhbase{\pi}{i}$, and $\bhperiod{\pi}{i}$]\label{def:localbehbasis}
For every 1-type $\pi\in\onetypes$,
we construct $\numbhbases{\pi}$ vectors of non-negative integers $\bhbase{\pi}{1},\ldots,\bhbase{\pi}{\numbhbases{\pi}}$
and a set of $\numbhperiods{\pi}$ vectors of non-negative integers 
$\bhP{\pi} = \set{\bhperiod{\pi}{1}, \ldots, \bhperiod{\pi}{\numbhperiods{\pi}}}$,
where every $\bhperiod{\pi}{1}$ is not the zero vector,
such that
$\bhset{\pi} = \bigcup_{i \in \intsinterval{\numbhbases{\pi}}} \cL(\bhbase{\pi}{i}; \bhP{\pi})$.
We let $\bhsethomo{\pi}$ be the set $\cL_\homo(\bhP{\pi}) = \cL(\bfzero; \bhP{\pi})$.
\end{definition}
Note that $\bhset{\pi}$ forms a union of linear sets sharing the same finite set of period vectors. 
This follows from the fact that $\bhset{\pi}$ is defined by 
a quantifier-free Presburger formula without disjunction --
it is a conjunction of linear (in)equalities in the Presburger quantifiers.

Fixing all these vectors $\bhbase{\pi}{i}$'s and $\bhperiod{\pi}{j}$'s,
we define the notion of \emph{local behavior vector},
which records its base index and the multiplicities of its periods.

\begin{definition}[Local behavior vector ($\lbv$), realizing $\albv$, valid, homogeneously valid]\label{def:lbv}
    Let $\pi$ be a $1$-type.
    A \emph{local behavior vector ($\lbv$)} $\albv$ of a 1-type $\pi$ is a vector of non-negative integers
  with
    $\numbhbases{\pi} + \numbhperiods{\pi}$ entries.
    We denote its entries by $\albv(x, i)$ for $x \in \set{\basetext, \periodtext}$ and
    $1 \le i \le \numbhxs{\pi}$.

    We say that $\albv$ is \emph{valid} if there is exactly one $i$ such that
    $\albv(\basetext, i) = 1$, and for $j \neq i$, $\albv(\basetext, j) = 0$.
    This unique index $i$ is called \emph{the base index of $\albv$}.
    We say that $\albv$ is \emph{homogeneously valid} if $\albv(\basetext, i) = 0$ for every $i$.

    {\noindent}The behavior generated by $\albv$ is:
    $\eval{\albv} := 
    \sum\limits_{i \in \intsinterval{\numbhbases{\pi}}} \albv(\basetext, i) \cdot \bhbase{\pi}{i} +
    \sum\limits_{i \in \intsinterval{\numbhperiods{\pi}}} \albv(\periodtext, i) \cdot \bhperiod{\pi}{i}$.

    Given a structure $\cG$, an element \emph{$u$ realizes $\albv$} if its behavior is $\eval{\albv}$.
    In this case, we say that $\albv$ is an $\lbv$ of $u$. 
\end{definition}

\begin{remark}[Non-unique representations] \label{rem:nonuniquelbv} Although we have used the term ``basis'' for the representation of our hybrid linear set in Definition \ref{def:localbehbasis},  we emphasize that this does not give a \emph{unique representation}, as one would have from a basis for a vector space:
The $\lbv$ realized by a vertex $u$ is not necessarily unique.
Indeed, it is possible for two distinct valid $\lbv$s to generate the same behavior.
We can say that an element realizes \emph{some} valid $\lbv$ of 1-type $\pi$ whenever
its behavior is in the set $\bhset{\pi}$.
\end{remark}

Having discussed local behavior vectors, which capture local neighborhoods, 
we introduce the notion of \emph{global behavior vector} ($\gbv$).
While an $\lbv$ represents a behavior of a given element,
a $\gbv$ will give a global description of a structure.
Intuitively, a $\gbv$ records the multiplicities of behaviors in a structure.
It will do this not in terms of the number of occurrences of each conformant behavior -- 
since there are infinitely many -- but in terms of the base and period components 
given by the hybrid linear representation.

More precisely, for each 1-type, a $\gbv$ counts the occurrences 
of the base and period components of the conformant behaviors.
In doing so, however, it loses the information about the connectivity of the structure and 
which period components belong to the same element.

\begin{definition}[Global behavior vector ($\gbv$), slice]\label{def:gbv}
A \emph{global behavior vector ($\gbv$)} $\agbv$ is a vector with
$\sum_{\pi \in \onetypes} \left(\numbhbases{\pi} + \numbhperiods{\pi}\right)$ entries.
We denote its entries by $\agbv(\pi, x, i)$ for $\pi \in \onetypes$, $x \in \set{\basetext, \periodtext}$, and
$1 \le i \le \numbhxs{\pi}$.

The \emph{slice} of $\agbv$ on 1-type $\pi$, denoted by $\agbv(\pi,-)$, is a vector with 
$\numbhbases{\pi} + \numbhperiods{\pi}$ entries where for $x \in \set{\basetext, \periodtext}$, and
$1 \le i \le \numbhxs{\pi}$, $\agbv(\pi,-)(x,i) = \agbv(\pi,x,i)$.
\end{definition}

We illustrate the different notions of behavior in Figure \ref{fig:beh}.

We now define the notion of realization between structures and $\gbv$s.
Intuitively, a $\gbv$ $\agbv$ is realized by a structure $\cG$ if 
$\agbv$ is the sum of  $\lbv$s of all vertices in $\cG$.

\begin{definition}[Structure realizing a $\gbv$, realization witness]\label{def:struct_realize_gbv}
    Let $\cG$ be a structure.
    We say that  $\cG$ \emph{realizes} a $\gbv$ $\agbv$ if 
    there exists a mapping $\rep$ from the vertices of $\cG$ to valid $\lbv$s such that
    \begin{inparaenum}
        \item for every element $u$, $u$ realizes $\rep(u)$ (Def. \ref{def:lbv}) and
        \item for every $\pi \in \onetypes$, 
        the slice of $\agbv$ on $\pi$ is precisely the sum of $\rep(u)$ over all elements $u$ of 1-type $\pi$ in $\cG$.
    \end{inparaenum}

    The mapping $\rep$ is called a \emph{witness of the realization}.
\end{definition}

We can show that a structure $\cG$ satisfies $\anfgptwosentence$ if and only if it realizes some $\gbv$.
Indeed, if $\cG$ does not satisfy $\anfgptwosentence$, then 
there exists an element $u$ of $\cG$ with 1-type $\pi$ that does not satisfy $\univpartphi(x)$,  
and hence, its behavior is not in $\bhset{\pi}$.
Thus, such a mapping does not exist, and $\cG$ realizes no $\gbv$.
Conversely, if $\cG$ satisfies $\anfgptwosentence$, then every element $u$ of $\cG$ 
satisfies $\univpartphi(x)$, and in particular,
its behavior is in $\bhset{\pi}$, where $\pi$ is the 1-type of $u$.
Thus, the mapping $\rep$ that maps every element $u$ to its $\lbv$
is the witness for $\cG$ realizing some $\gbv$.

Note that since there may exist multiple $\lbv$s that generate the same behavior, 
a structure may realize multiple $\gbv$s. 
Our goal is to reduce the satisfiability of $\anfgptwosentence \land \agcc$
to the existence of a realizable $\gbv$ for a given hybrid linear set,
which in turn is reduced to the satisfiability of Presburger formulas.

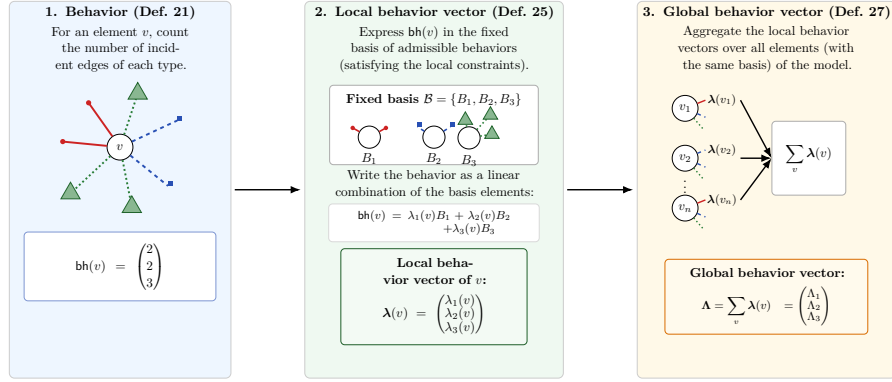
\begin{figure}[htbp]
    \centering
    \resizebox{1.0\textwidth}{!}{%
        \input{fig/behaviorfigure}
    }
    \caption{Illustration of behavior, local behavior vector, and global behavior vector.}
    \label{fig:beh}
\end{figure}

\myparagraph{Overapproximation of $\gbv$s realized by models of formulas; D-structures}
It is easy to see that a $\gcc$, which constrains cardinalities of elements in a structure,
implies a linear constraint over a $\gbv$, with variables corresponding to the numbers $\agbv(\pi,x,i)$.
We can combine this with some other sanity conditions to derive an \emph{overapproximation Presburger formula} 
for $\gbv$s realized by models of $\anfgptwosentence \land \agcc$.
This is formally defined in Section~\ref{subsec:overapproximation}. 
An example of a sanity condition is a \emph{matching constraint} saying
that the sum of the entries corresponding to a given 2-type is the same as the sum of entries for
the dual of the 2-type as defined in Def.~\ref{def:oneandtwotype}.

However, not all $\gbv$s in the overapproximation are realizable.
This brings us to the main and most involved part of our proof:  
deciding whether there is a realizable $\gbv$ in a given hybrid linear set.

Let us fix a representation $\frb$ and $\frP = \set{\frp_1, \ldots, \frp_\numouterperiods}$, and let $\cL(\frb; \frP)$
be the solution set of the corresponding overapproximation.
We will show that 
necessary and sufficient conditions 
for $\agbv$ in $\cL(\frb; \frP)$ to be realizable by a structure can be encoded
as a Presburger formula.

First, we show that every  $\agbv$ satisfying the overapproximation is realizable by what we call a \emph{decomposed structure} (\emph{D-structure}).
Intuitively, a D-structure is a relaxation of a structure
where we allow phantom elements and multiple 2-types between two elements.
For technical convenience, we define D-structures in terms of directed multigraphs.

\begin{definition}[D-structure, phantom vertex]\label{def:D-structure}
    A \emph{D-structure} is a nonempty directed, node- and edge-labelled \emph{multigraph} 
    $\tuple{V, V^\phan, E, \onetypeof, \twotypeof}$, where
    \begin{inparaenum}
        \item
            $V$ is the set of \emph{concrete vertices} and
            $V^\phan$ is the set of \emph{phantom vertices},
            where we require $V$ and $V^\phan$ to be disjoint.
        \item
            $E$ is the \emph{multiset} of directed edges,
        \item
            $\onetypeof$ labels vertices in $V \cup V^\phan$ by 1-types, and
        \item
            $\twotypeof$ is a mapping that assigns a $2$-type to each edge $e\in E$. 
    \end{inparaenum}


\begin{definition}[implicit $\eta$-edge, explicit $\eta$-edge, $\eta$-edge] \label{def:implicitexplicitetaedge}
We refer to an edge $e$ from $u$ to $v$ as an \emph{explicit $\eta$-edge} if $\twotypeof(e) = \eta$.
We will also refer to it 
as an \emph{implicit $\dual{\eta}$-edge from $v$ to $u$}.
An \emph{$\eta$-edge from $u$ to $v$} is either 
an explicit $\eta$-edge from $u$ to $v$, or an implicit $\eta$-edge from $v$ to $u$.
\end{definition}

    The \emph{behavior} of a vertex $v \in V \cup V^\phan$, denoted by $\bh{v}$,
    is the vector whose $(\eta,\pi)$-entry is the number of distinct $\eta$-edges from $v$ to vertices $u$ with $\onetypeof(u)=\pi$.

    A D-structure is \emph{phantomless} if it contains no phantom vertices, i.e.,
    $V^\phan = \emptyset$.

  A D-structure is asymmetric if all two-types associated with edges between distinct nodes are asymmetric.
\end{definition}

A D-structure is a kind of pre-model in which there are two obstacles to being a proper model. 
First, there can be \emph{multi-edges connecting two vertices}: multiple directed edges, possibly in different directions and possibly assigned to different $2$-types,  which a model does not allow. 
Second, the behavior that we ultimately want for a node in a real structure might be split between a concrete node and one or more phantoms.
Phantoms represent partial behaviors that may not be conformant,
but can be merged with a concrete vertex while maintaining conformance. 
In certain situations -- and under further restrictions -- 
a subset of the phantoms can be pumped to make a D-structure larger, 
which is useful for enabling swapping, which in turn will allow the removal of multi-edges.
Lastly, we emphasize that 
\emph{$D$-structures allow self-loops: edges or multi-edges where the source and target are the same}.
Note that a phantomless D-structure 
without multi-edges and self-loops corresponds to a structure where
the $2$-type of a non-edge $(u,v)$ is the null $2$-type.

\begin{definition}[Realization of a $\gbv$ in a D-structure]\label{def:realizingdstructure}
We extend the notion of a structure realizing a $\gbv$ to D-structures. This requires a realization witness map as in Def. \ref{def:struct_realize_gbv}, but mapping \emph{the concrete vertices to valid $\lbv$s} and 
\emph{the phantom vertices to homogeneously valid $\lbv$s}, where these notions are defined in Definition \ref{def:lbv}.
\end{definition}

The intention is that when we merge a concrete vertex with 1-type $\pi$ with
a non-adjacent phantom vertex with the same 1-type $\pi$,
its behavior will stay in the hybrid linear set $\bhset{\pi}$.

We observe that the following operations on D-structures preserve $\gbv$s, and also preserve asymmetry of a structure. 
They will be used to transform  D-structures to structures.
\begin{restatable}{lemma}{opspreserve}\label{lem:operation}
For a D-structure $\cG$, each of the following operations preserves the realized $\gbv$s.
\begin{inparaenum}
\item \itemtitle{Merge}
Merge a phantom vertex $v$ into a concrete or phantom vertex $u$, where $\onetypeof(v) = \onetypeof(u)$. We will typically apply this where we want to  preserve the property that there are no self-loops. In such a case, we must further require that $v$ and $u$ are not adjacent.
\item \itemtitle{Swap}
Let $e_1$ be an $\eta$-edge from $v_1$ to $u_1$, and $e_2$ an $\eta$-edge from $v_2$ to $u_2$ with
$\onetypeof(v_1) = \onetypeof(v_2)$ and $\onetypeof(u_1) = \onetypeof(u_2)$.
Suppose that $e_1$ and $e_2$ are distinct.
Remove the edges $e_1$ and $e_2$.
Add an edge $e_3$ from $v_1$ to $u_2$ and an edge $e_4$ from $v_2$ to $u_1$.
Assign the 2-types so that
$\twotypeof(e_3) = \twotypeof(e_4) = \eta$.
\end{inparaenum}
\end{restatable}

\myparagraph{Conditions for realizability of a linear set of $\gbv$s}
We now introduce the conditions under which a D-structure whose $\gbv$ is in $\cL(\frb; \frP)$ 
can be transformed into another structure whose (possibly different) $\gbv$ is also in $\cL(\frb; \frP)$.
Then we show how to encode these conditions in a Presburger formula.

From $\frb$ and $\frP$ we can deduce which 1-types have a bounded number of vertices for every structure realizing a $\gbv$ in $\cL(\frb;\frP)$.
These are called \emph{king} 1-types. 
The rest are called \emph{peasant} 1-types.
We will be interested in the scenario where a D-structure realizing a $\agbv$ in $\cL(\frb; \frP)$ 
can be decomposed into two disjoint D-structures --
the \emph{kernel} $\cK$ and the \emph{village} $\cV$ -- 
with $\cK$ realizing some $\gbv$ in $\cL(\frb; \frP)$ and
$\cV$ realizing some $\gbv$ in $\cL_{\homo}(\frP)$.
The kernel has a bounded number of vertices and a bounded number of edges.
Both bounds can be computed effectively from $\frb$ and $\frP$. 

In the decomposition, we can impose additional conditions on the kernel such as:
\begin{compactenum}
\item[K1.]
Only phantom vertices may have peasant 1-types.
We call concrete vertices with king 1-types \emph{kings},
non-king concrete vertices \emph{nobles},
and phantom vertices \emph{phantom nobles}.
\item[K2.]
Multi-edges cannot involve kings. 
Kings cannot have self-loops.
\end{compactenum}
A village which is compatible with a given kernel will have the following properties:
\begin{compactenum}
    \item[V1.] 
        All phantom vertices have king 1-types,
        while concrete vertices in the village have peasant 1-types.
        We call the phantom vertices \emph{shadows} and the concrete vertices \emph{peasants}.
    \item[V2.]
        There are no edges between shadows.
        Multi-edges cannot involve shadows.
        There are no self-loops in the village. 
    \item[V3.]
        There is a $1$-type preserving bijection between village shadows and kings in the  kernel.
    \item[V4.]
        For a phantom noble in the kernel, there is a peasant in the village with the same $1$-type.
    \item[V5.]    
        For every multi-edge in the kernel, there is a pair of peasants in the village with  1-types matching the endpoints of the edge in the kernel and the same 2-type as the edge.
    \item[V6.]
        For every self-loop with 2-type $\eta$ in the vertex with 1-type $\pi$ in the kernel,
        there is a pair of peasants in the village with $\pi$ and an edge with $\eta$ between them.
\end{compactenum}

See Figure~\ref{fig:kernel_village} for an illustration.

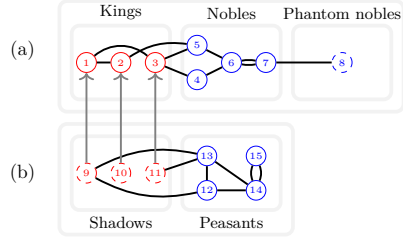
\begin{figure}[h!]
    \centering
    \resizebox{0.4\textwidth}{!}{
        \input{fig/kernel_village}
    }
    \caption{
        Illustration of (a) a kernel and (b) a compatible village.
        Concrete vertices are drawn with solid outlines (vertices 1-7 and 12-15), and 
        phantom vertices are drawn with dashed outlines (8-11).
        Vertices of king 1-types are shown in red (1-3 and 9-11), while 
        vertices of peasant 1-types are shown in blue (4-8 and 12-15).
        The gray arrows indicate the bijection between kings and shadows in the condition \emph{shadow--king correspondence}.
    }
    \label{fig:kernel_village}
\end{figure}

To obtain a structure from a kernel $\cK$ and village $\cV$ compatible with $\cK$, we have a \emph{pumping theorem} (see Theorem \ref{thm:pump}).
We create enough copies of villages so that we can:
\begin{compactenum}
    \item 
        Merge the shadows in the villages with the kings in the kernel
        according to the bijection defined in V3.
Note that, by conditions K2 and V2, no multi-edges involving kings
        are created by this merging process.

    \item 
        Merge each phantom noble in the kernel with an appropriate peasant
        in the village. Condition V4 guarantees that such a peasant exists.
After these first two steps, every vertex satisfies the local
        conditions, and no phantom vertices remain. It remains to eliminate
        the multi-edges.

    \item
        Eliminate the multi-edges arising from the peasants in the villages.
        Sufficient pumping guarantees that we can always find suitable swap
        partners.

    \item 
        Eliminate the remaining multi-edges arising from the phantom nobles
        in the kernel. Condition V5 guarantees that, after enough pumping
        of the villages, suitable swap partners are available.

    \item 
        Eliminate the remaining self-loop arising from the nobles and phantom nobles
        in the kernel. Condition V6 guarantees that, after enough pumping
        of the villages, suitable swap partners are available.
\end{compactenum}
See Figure~\ref{fig:pumping} for an illustration.

\begin{figure}[ht]
    \centering
    \resizebox{0.8\textwidth}{!}{
        \input{fig/pumping}
    }
    \caption{
        Illustration of the pumping argument.
        The construction starts with one kernel and three disjoint copies of a village.
        The green edges all have the same 2-type and connect vertices of the same 1-types.
        Similarly, the magenta edges all have the same 2-type and connect vertices of the same 1-types.
        The vertices marked with $\times$ have the same 1-type.
        In Steps~1 and~2,
        we merge the corresponding shadows in the three village copies with the kings in the kernel.
        We also merge each phantom noble in the kernel with a peasant of the same 1-type in one of the village copies.
        Condition~V4 guarantees the existence of such a peasant.
        In Steps~3 and~4,
        we apply the swapping operation to the green and magenta edges to eliminate the resulting multi-edges.
        Since the green (respectively, magenta) edges originate from corresponding edges in different copies of the same village, they all have the same 2-type.
        Condition~V5 guarantees that a suitable edge is always available for each swapping operation.
    }
    \label{fig:pumping}
\end{figure}
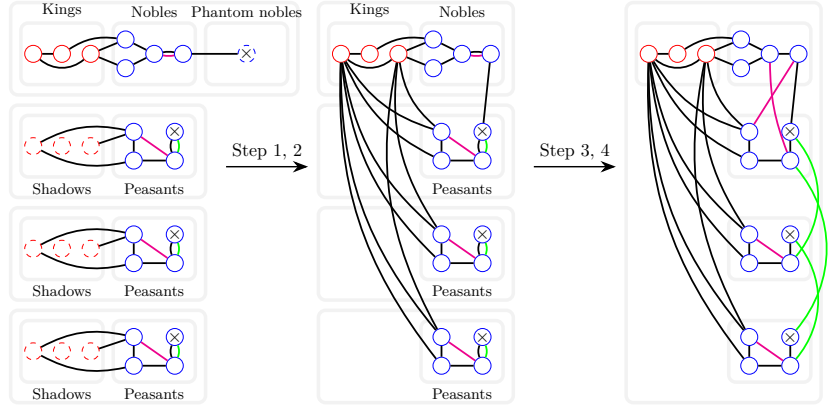

For a fixed kernel $\cK$, we construct a \emph{Presburger formula describing a decomposition into a village for the kernel}, or just ``village formula'' below. 
Informally, the village formula will be a  linear system whose solutions
represent a village $\cV$ compatible with the kernel $\cK$ such that the union $\cK\cup \cV$ 
realizes a $\gbv$ in $\cL(\frb;\frP)$.
Our main theorem then breaks down into two further components:
\begin{inparaenum}
\item 
    A \emph{satisfiability-to-solvability theorem} (Theorem \ref{thm:sattosolve}):
    Existence of a structure realizing a $\gbv$ in $\cL(\frb;\frP)$ implies existence of a kernel and solvability of the corresponding constraints.
\item 
    A \emph{solvability-to-village-existence theorem} (Theorem \ref{thm:solvetovillage}):
    Given a kernel $\cK$, the solvability of the corresponding constraints implies existence of a village compatible with $\cK$.
    Note that, by the pumping theorem, this implies the existence of a structure realizing a $\gbv$ in $\cL(\frb;\frP)$.
\end{inparaenum}

Combined with the \emph{pumping theorem}, which allows us to obtain a model from a kernel and compatible village,
we complete the proof.
The details are in Section~\ref{sec:fullproofmain}.

%% file: fig/behaviorfigure.tex
\usetikzlibrary{arrows.meta,positioning,shapes.geometric,calc}

\definecolor{panelblue}{RGB}{238,246,255}
\definecolor{panelgreen}{RGB}{239,249,241}
\definecolor{panelyellow}{RGB}{255,249,232}
\definecolor{edgeRed}{RGB}{205,35,35}
\definecolor{edgeBlue}{RGB}{40,80,180}
\definecolor{edgeGreen}{RGB}{40,125,55}

\tikzset{
  >=Latex,
  panel/.style={
    draw=black!20,
    rounded corners=4pt,
    very thin,
    inner sep=10pt,
    minimum height=10.2cm,
  },
  title/.style={
    font=\bfseries\large,
    align=center
  },
  body/.style={
    font=\normalsize,
    align=center
  },
  elem/.style={
    circle,
    draw=black,
    fill=white,
    minimum size=7mm,
    inner sep=0pt
  },
  basisnode/.style={
    circle,
    draw=black,
    fill=white,
    minimum size=6mm,
    inner sep=0pt
  },
  localbox/.style={
    draw=edgeGreen!75!black,
    fill=panelgreen!75,
    rounded corners=3pt,
    inner sep=7pt,
    align=center
  },
  globalbox/.style={
    draw=orange!85!black,
    fill=panelyellow!80,
    rounded corners=3pt,
    inner sep=7pt,
    align=center
  },
  flow/.style={
    -Latex,
    line width=1.1pt
  }
}

\begin{tikzpicture}[x=1cm,y=1cm]

\begin{scope}[shift={(0,0)}]

  \node[
    panel,
    fill=panelblue,
    minimum width=6cm,
    anchor=north west
  ] (P1) at (0,0) {};

  \node[title,text width=5.2cm] at (3,-0.3)
    {1. Behavior (Def.~\ref{def:behavior})};

  \node[body,text width=5.0cm] at (3,-1.3)
    {For an element $v$, count the number
     of incident edges of each type.};

  \node[elem] (x) at (3,-3.95) {$v$};

  \draw[edgeRed,line width=1.5pt]
    (x) -- ++(-1.55,0.15) coordinate (r1);

  \draw[edgeRed,line width=1.5pt]
    (x) -- ++(-0.85,1.20) coordinate (r2);

  \fill[edgeRed] (r1) circle (2pt);
  \fill[edgeRed] (r2) circle (2pt);

  \draw[edgeBlue,dashed,line width=1.5pt]
    (x) -- ++(1.55,0.72) coordinate (b1);

  \draw[edgeBlue,dashed,line width=1.5pt]
    (x) -- ++(1.30,-1.05) coordinate (b2);

  \fill[edgeBlue] (b1) rectangle ++(0.13,0.13);
  \fill[edgeBlue] (b2) rectangle ++(0.13,0.13);

  \draw[edgeGreen,dotted,line width=1.5pt]
    (x) -- ++(0.48,1.45) coordinate (g1);

  \draw[edgeGreen,dotted,line width=1.5pt]
    (x) -- ++(0.32,-1.62) coordinate (g2);

  \draw[edgeGreen,dotted,line width=1.5pt]
    (x) -- ++(-1.34,-1.25) coordinate (g3);

  \foreach \p in {g1,g2,g3}{
    \node[
      regular polygon,
      regular polygon sides=3,
      draw=edgeGreen!80!black,
      fill=edgeGreen!55,
      minimum size=6mm,
      inner sep=0pt
    ] at (\p) {};
  }

  \node[
    draw=edgeBlue!60,
    fill=white,
    rounded corners=3pt,
    align=center,
    text width=4.7cm,
    inner sep=7pt
  ] at (3,-7.18)
  {
    $\displaystyle
    \bh{v}=
    \begin{pmatrix}
      2\\
      2\\
      3
    \end{pmatrix}$
  };

\end{scope}

\begin{scope}[shift={(8,0)}]

  \node[
    panel,
    fill=panelgreen,
    minimum width=7cm,
    anchor=north west
  ] (P2) at (0,0) {};

  \node[title,text width=10cm] at (3.5,-0.3)
    {2. Local behavior vector (Def.~\ref{def:lbv})};

  \node[body,text width=5.5cm] at (3.5,-1.30)
    {Express $\bh{v}$ in the fixed basis of
     admissible behaviors\\
     (satisfying the local constraints).};

  \node[
    draw=gray!65,
    fill=white,
    rounded corners=3pt,
    minimum width=5.65cm,
    minimum height=2.15cm
  ] (basis) at (3.5,-3.35) {};

  \node[font=\bfseries] at (3.5,-2.66)
    {Fixed basis $\mathcal B=\{B_1,B_2,B_3\}$};

  \node[basisnode] (c1) at (1.75,-3.68) {};

  \draw[edgeRed,line width=1.25pt]
    (c1) -- ++(-0.45,0.25);

  \draw[edgeRed,line width=1.25pt]
    (c1) -- ++(0.45,0.25);

  \fill[edgeRed]
    ($(c1)+(-0.45,0.25)$) circle (1.7pt);

  \fill[edgeRed]
    ($(c1)+(0.45,0.25)$) circle (1.7pt);

  \node at (1.75,-4.25) {$B_1$};

  \node[basisnode] (c2) at (3.5,-3.68) {};

  \draw[edgeBlue,dashed,line width=1.25pt]
    (c2) -- ++(-0.45,0.28);

  \draw[edgeBlue,dashed,line width=1.25pt]
    (c2) -- ++(0.45,0.28);

  \fill[edgeBlue]
    ($(c2)+(-0.45,0.28)$)
    rectangle ++(0.11,0.11);

  \fill[edgeBlue]
    ($(c2)+(0.45,0.28)$)
    rectangle ++(0.11,0.11);

  \node at (3.5,-4.25) {$B_2$};

  \node[basisnode] (c3) at (4.45,-3.70) {};

  \draw[edgeGreen,dotted,line width=1.25pt]
    (c3) -- ++(-0.10,0.48);

  \draw[edgeGreen,dotted,line width=1.25pt]
    (c3) -- ++(0.62,0.10);

  \draw[edgeGreen,dotted,line width=1.25pt]
    (c3) -- ++(0.58,0.62);

  \foreach \dx/\dy in {
    -0.10/0.48,
     0.62/0.10,
     0.58/0.62
  }{
    \node[
      regular polygon,
      regular polygon sides=3,
      draw=edgeGreen!80!black,
      fill=edgeGreen!55,
      minimum size=4.5mm,
      inner sep=0pt
    ] at ($(c3)+(\dx,\dy)$) {};
  }

  \node at (4.45,-4.32) {$B_3$};

  \node[body,text width=5.4cm] at (3.5,-4.92)
    {Write the behavior as a linear combination
     of the basis elements:};

  \node[
    draw=black!15,
    fill=white,
    rounded corners=3pt,
    text width=5.35cm,
    inner sep=5pt,
    align=center
  ] at (3.5,-6)
  {
    \small
    $\displaystyle
    \bh{v}=
    \lambda_1(v)B_1+
    \lambda_2(v)B_2$\\[-1mm]
    $\displaystyle
    \hspace{1.9cm}+
    \lambda_3(v)B_3$
  };

  \node[
    localbox,
    text width=4.55cm
  ] at (3.5,-8)
  {
    {\bfseries Local behavior vector of $v$:}\\[2mm]
    $\displaystyle
    \boldsymbol{\lambda}(v)=
    \begin{pmatrix}
      \lambda_1(v)\\[-1mm]
      \lambda_2(v)\\[-1mm]
      \lambda_3(v)
    \end{pmatrix}$
  };

\end{scope}

\begin{scope}[shift={(17,0)}]

  \node[
    panel,
    fill=panelyellow,
    minimum width=7cm,
    anchor=north west
  ] (P3) at (0,0) {};

  \node[title,text width=10cm] at (3.50,-0.3)
    {3. Global behavior vector (Def.~\ref{def:gbv})};

  \node[body,text width=6.2cm] at (3.50,-1.3)
    {Aggregate the local behavior vectors
     over all elements (with the same basis)
     of the model.};

  \node[elem] (x1) at (1.30,-2.90) {$v_1$};
  \node[elem] (x2) at (1.30,-4.25) {$v_2$};
  \node[elem] (xn) at (1.30,-5.60) {$v_n$};

  \draw[edgeRed,line width=1.1pt]
    (x1) -- ++(0.55,0.22);

  \draw[edgeBlue,dashed,line width=1.1pt]
    (x1) -- ++(0.55,-0.24);

  \draw[edgeGreen,dotted,line width=1.1pt]
    (x1) -- ++(0.48,-0.55);

  \draw[edgeBlue,dashed,line width=1.1pt]
    (x2) -- ++(0.55,0.22);

  \draw[edgeBlue,dashed,line width=1.1pt]
    (x2) -- ++(0.55,-0.24);

  \draw[edgeGreen,dotted,line width=1.1pt]
    (x2) -- ++(0.48,-0.55);

  \draw[edgeRed,line width=1.1pt]
    (xn) -- ++(0.55,0.22);

  \draw[edgeBlue,dashed,line width=1.1pt]
    (xn) -- ++(0.55,-0.24);

  \draw[edgeGreen,dotted,line width=1.1pt]
    (xn) -- ++(0.48,-0.55);

  \node[font=\small] at (2.30,-2.68)
    {$\boldsymbol{\lambda}(v_1)$};

  \node[font=\small] at (2.30,-4.03)
    {$\boldsymbol{\lambda}(v_2)$};

  \node[font=\small] at (2.30,-5.38)
    {$\boldsymbol{\lambda}(v_n)$};

  \node[font=\large] at (1.30,-4.92)
    {$\vdots$};

  \node[
    draw=gray!65,
    fill=white,
    rounded corners=3pt,
    minimum width=2.0cm,
    minimum height=2.0cm
  ] (sum) at (4.65,-4.25)
    {$\displaystyle\sum_v\boldsymbol{\lambda}(v)$};

  \draw[flow] (2.80,-2.68) -- (sum.west);
  \draw[flow] (2.80,-4.25) -- (sum.west);
  \draw[flow] (2.80,-5.38) -- (sum.west);

  \node[
    globalbox,
    text width=5cm
  ] at (3.5,-8)
  {
    {\bfseries Global behavior vector:}\\[2mm]
    {\small
    $\displaystyle
    \begin{aligned}
    \boldsymbol{\Lambda}
      &=\sum_v\boldsymbol{\lambda}(v)
      &=\begin{pmatrix}
          \Lambda_1\\[-1mm]
          \Lambda_2\\[-1mm]
          \Lambda_3
        \end{pmatrix}
    \end{aligned}$}
  };

\end{scope}


\draw[flow] ([xshift=0.1cm]P1.east) -- ([xshift=-0.1cm]P2.west);

\draw[flow] ([xshift=0.1cm]P2.east) -- ([xshift=-0.1cm]P3.west);

\end{tikzpicture}

%% file: fig/kernel_village.tex

\begin{tikzpicture}[scale=.9]
    \begin{scope}[shift={(0, 2.5)}]
        \draw[bdr] (-0.25, -0.25) rectangle (6.75, 2.00);
        \node at (-1, 1) {(a)};

        \begin{scope}[shift={(0, 0)}]
            \draw[bdr] (0.0, 0.0) rectangle (2.0, 1.5);
            \node[anchor=south] at (1.0, 1.5) {\footnotesize Kings};

            \node[newvtx, red] (kk1) at (0.3, 0.75) {\tiny 1};
            \node[newvtx, red] (kk2) at (1.0, 0.75) {\tiny 2};
            \node[newvtx, red] (kk3) at (1.7, 0.75) {\tiny 3};

            \draw[edge]               (kk1) to (kk2);
            \draw[edge, bend left=40] (kk1) to (kk3);
        \end{scope}

        \begin{scope}[shift={(2.25, 0)}]
            \draw[bdr] (0.0, 0.0) rectangle (2.0, 1.5);
            \node[anchor=south] at (1.0, 1.5) {\footnotesize Nobles};

            \node[newvtx, blue] (kn1) at (0.3, 0.4) {\tiny 4};
            \node[newvtx, blue] (kn2) at (0.3, 1.1) {\tiny 5};
            \node[newvtx, blue] (kn3) at (1.0, 0.75) {\tiny 6};
            \node[newvtx, blue] (kn4) at (1.7, 0.75) {\tiny 7};

            \draw[edge]                (kk3) to (kn1);
            \draw[edge]                (kk3) to (kn2);
            \draw[edge, bend left=20]  (kk2) to (kn2);

            \draw[edge]                (kn1) to (kn3);
            \draw[edge]                (kn2) to (kn3);
            \draw[edge, bend left=10]  (kn3) to (kn4);
            \draw[edge, bend right=10] (kn3) to (kn4);
        \end{scope}

        \begin{scope}[shift={(4.5, 0)}]
            \draw[bdr] (0.0, 0.0) rectangle (2.0, 1.5);
            \node[anchor=south] at (1.0, 1.5) {\footnotesize Phantom nobles};

            \node[newvtx, blue, dashed] (kp1) at (1.0, 0.75) {\tiny 8};

            \draw[edge] (kp1) to (kn4);
        \end{scope}
    \end{scope}

    \begin{scope}[shift={(0, 0)}]
        \draw[bdr] (-0.25, -0.25) rectangle (4.50, 2.00);
        \node at (-1, 1) {(b)};

        \begin{scope}[shift={(0, 0.25)}]
            \draw[bdr] (0.0, 0.0) rectangle (2.0, 1.5);
            \node[anchor=north] at (1.0, 0.0) {\footnotesize Shadows};

            \node[newvtx, red, dashed] (vs1) at (0.3, 0.75) {\tiny 9};
            \node[newvtx, red, dashed, inner sep = 1pt] (vs2) at (1.0, 0.75) {\tiny 10};
            \node[newvtx, red, dashed, inner sep = 1pt] (vs3) at (1.7, 0.75) {\tiny 11};
        \end{scope}

        \begin{scope}[shift={(2.25, 0.25)}]
            \draw[bdr] (0.0, 0.0) rectangle (2.0, 1.5);
            \node[anchor=north] at (1.0, 0.0) {\footnotesize Peasants};

            \node[newvtx, blue, inner sep = 1pt] (vp1) at (0.5, 0.4) {\tiny 12};
            \node[newvtx, blue, inner sep = 1pt] (vp2) at (0.5, 1.1) {\tiny 13};
            \node[newvtx, blue, inner sep = 1pt] (vp3) at (1.5, 0.4) {\tiny 14};
            \node[newvtx, blue, inner sep = 1pt] (vp4) at (1.5, 1.1) {\tiny 15};

            \draw[edge]                (vp1) to (vp2);
            \draw[edge]                (vp1) to (vp3);
            \draw[edge]                (vp2) to (vp3);
            \draw[edge, bend right=20] (vp3) to (vp4);
            \draw[edge, bend left=20]  (vp3) to (vp4);

            \draw[edge, bend right=20] (vs1) to (vp1);
            \draw[edge, bend left=20]  (vs1) to (vp2);
            \draw[edge]                (vs3) to (vp2);
        \end{scope}
    \end{scope}

    \draw[line width=1pt, gray, ->] (vs1) to (kk1);
    \draw[line width=1pt, gray, ->] (vs2) to (kk2);
    \draw[line width=1pt, gray, ->] (vs3) to (kk3);
\end{tikzpicture}

%% file: fig/pumping.tex
\begin{tikzpicture}[scale=.8]
    \begin{scope}[shift={(0, 0)}]
        \begin{scope}[shift={(0, 5)}]
            \draw[bdr] (-0.25, -0.25) rectangle (6.75, 2.00);

            \begin{scope}[shift={(0, 0)}]
                \draw[bdr] (0.0, 0.0) rectangle (2.0, 1.5);
                \node[anchor=south] at (1.0, 1.5) {\footnotesize Kings};

                \node[vtx, red] (k1k1) at (0.3, 0.75) {};
                \node[vtx, red] (k1k2) at (1.0, 0.75) {};
                \node[vtx, red] (k1k3) at (1.7, 0.75) {};

                \draw[edge]                (k1k1) to (k1k2);
                \draw[edge, bend right=40] (k1k1) to (k1k3);
            \end{scope}

            \begin{scope}[shift={(2.25, 0)}]
                \draw[bdr] (0.0, 0.0) rectangle (2.0, 1.5);
                \node[anchor=south] at (1.0, 1.5) {\footnotesize Nobles};

                \node[vtx, blue] (k1n1) at (0.3, 0.4) {};
                \node[vtx, blue] (k1n2) at (0.3, 1.1) {};
                \node[vtx, blue] (k1n3) at (1.0, 0.75) {};
                \node[vtx, blue] (k1n4) at (1.7, 0.75) {};

                \draw[edge]                (k1k3) to (k1n1);
                \draw[edge]                (k1k3) to (k1n2);
                \draw[edge, bend left=20]  (k1k2) to (k1n2);

                \draw[edge]                (k1n1) to (k1n3);
                \draw[edge]                (k1n2) to (k1n3);
                \draw[edge, bend left=10]  (k1n3) to (k1n4);
                \draw[edge, bend right=10, magenta] (k1n3) to (k1n4);
            \end{scope}

            \begin{scope}[shift={(4.5, 0)}]
                \draw[bdr] (0.0, 0.0) rectangle (2.0, 1.5);
                \node[anchor=south] at (1.0, 1.5) {\footnotesize Phantom nobles};

                \node[vtx, blue, dashed] (k1p1) at (1.0, 0.75) {};

                \draw[edge] (k1p1) to (k1n4);
            \end{scope}
        \end{scope}

        \begin{scope}[shift={(0, 2.5)}]
            \draw[bdr] (-0.25, -0.25) rectangle (4.50, 2.00);

            \begin{scope}[shift={(0, 0.25)}]
                \draw[bdr] (0.0, 0.0) rectangle (2.0, 1.5);
                \node[anchor=north] at (1.0, 0.0) {\footnotesize Shadows};

                \node[vtx, red, dashed] (v2s1) at (0.3, 0.75) {};
                \node[vtx, red, dashed] (v2s2) at (1.0, 0.75) {};
                \node[vtx, red, dashed] (v2s3) at (1.7, 0.75) {};
            \end{scope}

            \begin{scope}[shift={(2.25, 0.25)}]
                \draw[bdr] (0.0, 0.0) rectangle (2.0, 1.5);
                \node[anchor=north] at (1.0, 0.0) {\footnotesize Peasants};

                \node[vtx, blue] (v2p1) at (0.5, 0.4) {};
                \node[vtx, blue] (v2p2) at (0.5, 1.1) {};
                \node[vtx, blue] (v2p3) at (1.5, 0.4) {};
                \node[vtx, blue] (v2p4) at (1.5, 1.1) {};

                \draw[edge]                (v2p1) to (v2p2);
                \draw[edge]                (v2p1) to (v2p3);
                \draw[edge, magenta]       (v2p2) to (v2p3);
                \draw[edge, bend right=20, green] (v2p3) to (v2p4);
                \draw[edge, bend left=20]  (v2p3) to (v2p4);

                \draw[edge, bend right=20] (v2s1) to (v2p1);
                \draw[edge, bend left=20]  (v2s1) to (v2p2);
                \draw[edge]                (v2s3) to (v2p2);
            \end{scope}
        \end{scope}

        \begin{scope}[shift={(0, 0)}]
            \draw[bdr] (-0.25, -0.25) rectangle (4.50, 2.00);

            \begin{scope}[shift={(0, 0.25)}]
                \draw[bdr] (0.0, 0.0) rectangle (2.0, 1.5);
                \node[anchor=north] at (1.0, 0.0) {\footnotesize Shadows};

                \node[vtx, red, dashed] (v3s1) at (0.3, 0.75) {};
                \node[vtx, red, dashed] (v3s2) at (1.0, 0.75) {};
                \node[vtx, red, dashed] (v3s3) at (1.7, 0.75) {};
            \end{scope}

            \begin{scope}[shift={(2.25, 0.25)}]
                \draw[bdr] (0.0, 0.0) rectangle (2.0, 1.5);
                \node[anchor=north] at (1.0, 0.0) {\footnotesize Peasants};

                \node[vtx, blue] (v3p1) at (0.5, 0.4) {};
                \node[vtx, blue] (v3p2) at (0.5, 1.1) {};
                \node[vtx, blue] (v3p3) at (1.5, 0.4) {};
                \node[vtx, blue] (v3p4) at (1.5, 1.1) {};

                \draw[edge]                (v3p1) to (v3p2);
                \draw[edge]                (v3p1) to (v3p3);
                \draw[edge, magenta]       (v3p2) to (v3p3);
                \draw[edge, bend right=20, green] (v3p3) to (v3p4);
                \draw[edge, bend left=20]  (v3p3) to (v3p4);

                \draw[edge, bend right=20] (v3s1) to (v3p1);
                \draw[edge, bend left=20]  (v3s1) to (v3p2);
                \draw[edge]                (v3s3) to (v3p2);
            \end{scope}
        \end{scope}

        \begin{scope}[shift={(0, -2.5)}]
            \draw[bdr] (-0.25, -0.25) rectangle (4.50, 2.00);

            \begin{scope}[shift={(0, 0.25)}]
                \draw[bdr] (0.0, 0.0) rectangle (2.0, 1.5);
                \node[anchor=north] at (1.0, 0.0) {\footnotesize Shadows};

                \node[vtx, red, dashed] (vxs1) at (0.3, 0.75) {};
                \node[vtx, red, dashed] (vxs2) at (1.0, 0.75) {};
                \node[vtx, red, dashed] (vxs3) at (1.7, 0.75) {};
            \end{scope}

            \begin{scope}[shift={(2.25, 0.25)}]
                \draw[bdr] (0.0, 0.0) rectangle (2.0, 1.5);
                \node[anchor=north] at (1.0, 0.0) {\footnotesize Peasants};

                \node[vtx, blue] (vxp1) at (0.5, 0.4) {};
                \node[vtx, blue] (vxp2) at (0.5, 1.1) {};
                \node[vtx, blue] (vxp3) at (1.5, 0.4) {};
                \node[vtx, blue] (vxp4) at (1.5, 1.1) {};

                \draw[edge]                (vxp1) to (vxp2);
                \draw[edge]                (vxp1) to (vxp3);
                \draw[edge, magenta]       (vxp2) to (vxp3);
                \draw[edge, bend right=20, green] (vxp3) to (vxp4);
                \draw[edge, bend left=20]  (vxp3) to (vxp4);

                \draw[edge, bend right=20] (vxs1) to (vxp1);
                \draw[edge, bend left=20]  (vxs1) to (vxp2);
                \draw[edge]                (vxs3) to (vxp2);
            \end{scope}
        \end{scope}

        \node at (k1p1) {$\times$};
        \node at (v2p4) {$\times$};
        \node at (v3p4) {$\times$};
        \node at (vxp4) {$\times$};
    \end{scope}

    \begin{scope}[shift={(7.5, 0)}]
        \begin{scope}[shift={(0, 5)}]
            \draw[bdr] (-0.25, -0.25) rectangle (4.5, 2.00);

            \begin{scope}[shift={(0, 0)}]
                \draw[bdr] (0.0, 0.0) rectangle (2.0, 1.5);
                \node[anchor=south] at (1.0, 1.5) {\footnotesize Kings};

                \node[vtx, red] (k4k1) at (0.3, 0.75) {};
                \node[vtx, red] (k4k2) at (1.0, 0.75) {};
                \node[vtx, red] (k4k3) at (1.7, 0.75) {};

                \draw[edge]                (k4k1) to (k4k2);
                \draw[edge, bend right=40] (k4k1) to (k4k3);
            \end{scope}

            \begin{scope}[shift={(2.25, 0)}]
                \draw[bdr] (0.0, 0.0) rectangle (2.0, 1.5);
                \node[anchor=south] at (1.0, 1.5) {\footnotesize Nobles};

                \node[vtx, blue] (k4n1) at (0.3, 0.4) {};
                \node[vtx, blue] (k4n2) at (0.3, 1.1) {};
                \node[vtx, blue] (k4n3) at (1.0, 0.75) {};
                \node[vtx, blue] (k4n4) at (1.7, 0.75) {};

                \draw[edge]                (k4k3) to (k4n1);
                \draw[edge]                (k4k3) to (k4n2);
                \draw[edge, bend left=20]  (k4k2) to (k4n2);

                \draw[edge]                (k4n1) to (k4n3);
                \draw[edge]                (k4n2) to (k4n3);
                \draw[edge, bend left=10]  (k4n3) to (k4n4);
                \draw[edge, bend right=10, magenta] (k4n3) to (k4n4);
            \end{scope}

            %
            %
        \end{scope}

        \begin{scope}[shift={(0, 2.5)}]
            \draw[bdr] (-0.25, -0.25) rectangle (4.50, 2.00);

            %

            \begin{scope}[shift={(2.25, 0.25)}]
                \draw[bdr] (0.0, 0.0) rectangle (2.0, 1.5);
                \node[anchor=north] at (1.0, 0.0) {\footnotesize Peasants};

                \node[vtx, blue] (v5p1) at (0.5, 0.4) {};
                \node[vtx, blue] (v5p2) at (0.5, 1.1) {};
                \node[vtx, blue] (v5p3) at (1.5, 0.4) {};
                \node[vtx, blue] (v5p4) at (1.5, 1.1) {};

                \draw[edge]                (v5p1) to (v5p2);
                \draw[edge]                (v5p1) to (v5p3);
                \draw[edge, magenta]       (v5p2) to (v5p3);
                \draw[edge, bend right=20, green] (v5p3) to (v5p4);
                \draw[edge, bend left=20]  (v5p3) to (v5p4);

            \end{scope}
        \end{scope}

        \begin{scope}[shift={(0, 0)}]
            \draw[bdr] (-0.25, -0.25) rectangle (4.50, 2.00);

            %

            \begin{scope}[shift={(2.25, 0.25)}]
                \draw[bdr] (0.0, 0.0) rectangle (2.0, 1.5);
                \node[anchor=north] at (1.0, 0.0) {\footnotesize Peasants};

                \node[vtx, blue] (v6p1) at (0.5, 0.4) {};
                \node[vtx, blue] (v6p2) at (0.5, 1.1) {};
                \node[vtx, blue] (v6p3) at (1.5, 0.4) {};
                \node[vtx, blue] (v6p4) at (1.5, 1.1) {};

                \draw[edge]                (v6p1) to (v6p2);
                \draw[edge]                (v6p1) to (v6p3);
                \draw[edge, magenta]       (v6p2) to (v6p3);
                \draw[edge, bend right=20, green] (v6p3) to (v6p4);
                \draw[edge, bend left=20]  (v6p3) to (v6p4);
            \end{scope}
        \end{scope}

        \begin{scope}[shift={(0, -2.5)}]
            \draw[bdr] (-0.25, -0.25) rectangle (4.50, 2.00);

            %

            \begin{scope}[shift={(2.25, 0.25)}]
                \draw[bdr] (0.0, 0.0) rectangle (2.0, 1.5);
                \node[anchor=north] at (1.0, 0.0) {\footnotesize Peasants};

                \node[vtx, blue] (vyp1) at (0.5, 0.4) {};
                \node[vtx, blue] (vyp2) at (0.5, 1.1) {};
                \node[vtx, blue] (vyp3) at (1.5, 0.4) {};
                \node[vtx, blue] (vyp4) at (1.5, 1.1) {};

                \draw[edge]                (vyp1) to (vyp2);
                \draw[edge]                (vyp1) to (vyp3);
                \draw[edge, magenta]       (vyp2) to (vyp3);
                \draw[edge, bend right=20, green] (vyp3) to (vyp4);
                \draw[edge, bend left=20]  (vyp3) to (vyp4);
            \end{scope}
        \end{scope}
        \draw[edge, bend right=20] (k4k1) to (v6p1);
        \draw[edge, bend right=20] (k4k1) to (v6p2);
        \draw[edge, bend right=20] (k4k3) to (v6p2);

        \draw[edge, bend right=20] (k4k1) to (v5p1);
        \draw[edge, bend right=20] (k4k1) to (v5p2);
        \draw[edge, bend right=20] (k4k3) to (v5p2);

        \draw[edge, bend right=20] (k4k1) to (vyp1);
        \draw[edge, bend right=20] (k4k1) to (vyp2);
        \draw[edge, bend right=20] (k4k3) to (vyp2);

        \draw[edge] (k4n4) to (v5p4);

        \node at (v5p4) {$\times$};
        \node at (v6p4) {$\times$};
        \node at (vyp4) {$\times$};
    \end{scope}

    \begin{scope}[shift={(15.00, 0)}]
        \begin{scope}[shift={(0, 5)}]

            \begin{scope}[shift={(0, 0)}]
                \draw[bdr] (0.0, 0.0) rectangle (2.0, 1.5);

                \node[vtx, red] (k7k1) at (0.3, 0.75) {};
                \node[vtx, red] (k7k2) at (1.0, 0.75) {};
                \node[vtx, red] (k7k3) at (1.7, 0.75) {};

                \draw[edge]               (k7k1) to (k7k2);
                \draw[edge, bend right=40] (k7k1) to (k7k3);
            \end{scope}

            \begin{scope}[shift={(2.25, 0)}]
                \draw[bdr] (0.0, 0.0) rectangle (2.0, 1.5);

                \node[vtx, blue] (k7n1) at (0.3, 0.4) {};
                \node[vtx, blue] (k7n2) at (0.3, 1.1) {};
                \node[vtx, blue] (k7n3) at (1.0, 0.75) {};
                \node[vtx, blue] (k7n4) at (1.7, 0.75) {};

                \draw[edge]                (k7n1) to (k7n3);
                \draw[edge]                (k7n2) to (k7n3);
                \draw[edge]                (k7n3) to (k7n4);

                \draw[edge]                (k7k3) to (k7n1);
                \draw[edge]                (k7k3) to (k7n2);
                \draw[edge, bend left=20]  (k7k2) to (k7n2);
            \end{scope}
        \end{scope}

        \begin{scope}[shift={(0, 2.5)}]

            %

            \begin{scope}[shift={(2.25, 0.25)}]
                \draw[bdr] (0.0, 0.0) rectangle (2.0, 1.5);

                \node[vtx, blue] (v8p1) at (0.5, 0.4) {};
                \node[vtx, blue] (v8p2) at (0.5, 1.1) {};
                \node[vtx, blue] (v8p3) at (1.5, 0.4) {};
                \node[vtx, blue] (v8p4) at (1.5, 1.1) {};

                \draw[edge]                (v8p1) to (v8p2);
                \draw[edge]                (v8p1) to (v8p3);
                \draw[edge]                (v8p3) to (v8p4);

            \end{scope}
        \end{scope}

        \begin{scope}[shift={(0, 0)}]

            %

            \begin{scope}[shift={(2.25, 0.25)}]
                \draw[bdr] (0.0, 0.0) rectangle (2.0, 1.5);

                \node[vtx, blue] (v9p1) at (0.5, 0.4) {};
                \node[vtx, blue] (v9p2) at (0.5, 1.1) {};
                \node[vtx, blue] (v9p3) at (1.5, 0.4) {};
                \node[vtx, blue] (v9p4) at (1.5, 1.1) {};

                \draw[edge]          (v9p1) to (v9p2);
                \draw[edge]          (v9p1) to (v9p3);
                \draw[edge, magenta] (v9p2) to (v9p3);
                \draw[edge]          (v9p3) to (v9p4);
            \end{scope}
        \end{scope}

        \begin{scope}[shift={(0, -2.5)}]
            \draw[bdr] (-0.25, -0.25) rectangle (4.50, 9.50);

            %

            \begin{scope}[shift={(2.25, 0.25)}]
                \draw[bdr] (0.0, 0.0) rectangle (2.0, 1.5);

                \node[vtx, blue] (vzp1) at (0.5, 0.4) {};
                \node[vtx, blue] (vzp2) at (0.5, 1.1) {};
                \node[vtx, blue] (vzp3) at (1.5, 0.4) {};
                \node[vtx, blue] (vzp4) at (1.5, 1.1) {};

                \draw[edge]          (vzp1) to (vzp2);
                \draw[edge]          (vzp1) to (vzp3);
                \draw[edge, magenta] (vzp2) to (vzp3);
                \draw[edge]          (vzp3) to (vzp4);
            \end{scope}
        \end{scope}
    \end{scope}

    \draw[edge]                (k7n4) to (v8p4);

    \draw[edge, bend right=20] (k7k1) to (v9p1);
    \draw[edge, bend right=20] (k7k1) to (v9p2);
    \draw[edge, bend right=20] (k7k3) to (v9p2);

    \draw[edge, bend right=20] (k7k1) to (v8p1);
    \draw[edge, bend right=20] (k7k1) to (v8p2);
    \draw[edge, bend right=20] (k7k3) to (v8p2);

    \draw[edge, bend right=20] (k7k1) to (vzp1);
    \draw[edge, bend right=20] (k7k1) to (vzp2);
    \draw[edge, bend right=20] (k7k3) to (vzp2);

    \draw[edge, bend right=40, green] (v9p3) to (v8p4);
    \draw[edge, bend right=40, green] (vzp4) to (v8p3);
    \draw[edge, bend right=40, green] (vzp3) to (v9p4);

    \draw[edge, magenta, bend right=10] (k7n3) to (v8p3);
    \draw[edge, magenta]                (v8p2) to (k7n4);

    \node at (v8p4) {$\times$};
    \node at (v9p4) {$\times$};
    \node at (vzp4) {$\times$};

    \draw[
        -{Stealth[length=3mm]},
        thick
    ]
    (5, 3) -- node[above]{Step 1, 2} (7, 3);

    \draw[
        -{Stealth[length=3mm]},
        thick
    ]
    (12.5, 3) -- node[above]{Step 3, 4} (14.5, 3);
\end{tikzpicture}

%% file: 4-behaviorchar.tex

\section{Proof of Theorem~\ref{thm:ourlogic}} \label{sec:fullproofmain}

We now flesh out the outline, referring to the  notation 
from Table~\ref{table:notation}.
\input{notationtable}


We need some terminology to describe components of $\gbv$s.
\begin{definition}[$\gbv$ incorporating $1$- and $2$-types, basely or periodically]\label{def:gbvincorporates}
Let $\agbv$ be a $\gbv$.
    \begin{compactitem}
        \item
            $\agbv$ \emph{incorporates a $1$-type $\pi$ basely} if
            $\agbv(\pi, \basetext, i)>0$ for some $1 \le i\le\numbhbases{\pi}$.

        \item
            $\agbv$ \emph{incorporates a $1$-type $\pi$ periodically} if
            $\agbv(\pi, \periodtext, i)>0$ for some $1 \le i\le\numbhperiods{\pi}$.

        \item
            $\agbv$ \emph{incorporates a 1-type $\pi$} if
            $\agbv$ incorporates $\pi$ either basely or periodically.

        \item
            $\agbv$ \emph{incorporates a $2$-type $\eta$ between $\pi_1$ and $\pi_2$} if either
            \begin{compactitem}
                \item
                    there exists $1 \le i \le \numbhbases{\pi_1}$ such that
                    $\agbv(\pi_1, \basetext, i) > 0$ and
                    $\bhbase{\pi_1}{i}(\eta, \pi_2) > 0$, or
                \item
                    there exists $1 \le i \le \numbhperiods{\pi_1}$ such that
                    $\agbv(\pi_1, \periodtext, i) > 0$ and
                    $\bhperiod{\pi_1}{i}(\eta, \pi_2) > 0$.
            \end{compactitem}
    \end{compactitem}
\end{definition}

Intuitively, a $\gbv$ records which $1$-types and 2-types must appear in a realization,
and whether the corresponding $1$-types are witnessed by concrete or phantom vertices.

We record the following simple consequences of the definitions:
\begin{lemma}\label{lem:propertiesincorporating}
Let $\agbv$ be a $\gbv$ and $\cG$ be a D-structure that realizes $\agbv$.
\begin{compactitem}
\item
If $\agbv$ incorporates a $1$-type $\pi$ \emph{basely},
then there exists a \emph{concrete} vertex in $\cG$ realizing $\pi$.

\item
If $\agbv$ incorporates a 1-type $\pi$ \emph{periodically but not basely},
then there exists a \emph{phantom} vertex in $\cG$ realizing $\pi$.

\item%
If $\agbv$ incorporates a 2-type $\eta$ between $\pi_1$ and $\pi_2$,
then there exist nodes $u$ and $v$ and an $\eta$-edge from $v$ to $u$ in $\cG$ (as in Definition \ref{def:implicitexplicitetaedge})
with $\onetypeof(v) = \pi_1$ and $\onetypeof(u) = \pi_2$.
\end{compactitem}
\end{lemma}
\begin{proof} 
We start with the first item. Note that by Definition \ref{def:realizingdstructure}, the witness map
assigns each phantom vertex a homogeneously valid $\lbv$: one where the coefficients corresponding to \emph{base} vectors are all zero. In particular this is true for vertices with $1$-type $\pi$.
We also know from Definition \ref{def:realizingdstructure} that the slice of $\agbv$ on $\pi$ is the sum
of the $\lbv$s associated (via the witness map) to vertices with $1$-type $\pi$. If there are only phantom vertices
with $1$-type $\pi$, this means that for each component of the slice of $\agbv$ on $\pi$, given by some index $i$,
that component of the slice is the sum of the witness-defined $\lbv$s of phantom vertices, and for a base component this would be zero by homogeneous validity.
But this would mean that each base component of the slice of $\agbv$ on $\pi$ is zero, contradicting the fact
that $\pi$ is incorporated basely.

For the second part, suppose that $\agbv$ incorporates a 1-type $\pi$ periodically but not basely.
We first show that there is no concrete vertex $v$ in $\cG$ with 1-type $\pi$.
Suppose, for contradiction, that $v$ is such a vertex.
Then $\rep(v)$ is a valid $\lbv$, which implies that there exists some $1 \le i \le \numbhbases{\pi}$ such that
$\rep(v)(\basetext, i) > 0$.
Thus, $\agbv(\pi, \basetext, i) \ge \rep(v)(\basetext, i) > 0$,
contradicting the assumption that $\agbv$ does not incorporate $\pi$ basely.
We next show that there exists a phantom vertex $u$ in $\cG$ with 1-type $\pi$.
Again, suppose for contradiction that there is no such vertex.
Since there is also no concrete vertex with 1-type $\pi$, there is no vertex with 1-type $\pi$ in $\cG$.
It follows that $\agbv(\pi, \periodtext, j) = 0$ for every $1 \le j \le \numbhperiods{\pi}$.
This contradicts the assumption that $\agbv$ incorporates $\pi$ periodically.

For the third part, suppose $\agbv$ incorporates a 2-type $\eta$ between $\pi_1$ and $\pi_2$. 
The definition requires one of two options,
and the proofs are similar in both cases.
Let us only consider the first option:  
there is a base component present that includes an $\eta$ edge to $\pi_2$: 
there exists $1 \le i \le \numbhbases{\pi_1}$ such that
$\agbv(\pi_1, \basetext, i) > 0$ and
$\bhbase{\pi_1}{i}(\eta, \pi_2) > 0$. Fix such an $i$. For $\pi_1$, the witness map must assign
some vertex $v$ a $\lbv$ in which the $i^{th}$ base component is non-zero: otherwise $\agbv$ 
could not have a total for the base component above zero. 
Since this base component has a non-zero entry for $\eta$ and $\pi_2$ and $v$ realizes the $\lbv$,
there must be a $u$ as required, completing the argument in this case.
\end{proof}

We need a few more quantities associated with a $\gbv$.
\begin{definition}[Size and base size of a $\gbv$]\label{def:basesize}
    Let $\agbv$ be a $\gbv$.
    The \emph{size} of $\agbv$, denoted by $\abs{\agbv}$,
    the \emph{base size} of $\agbv$, denoted by $\projectionbase{\agbv}$, and
    the \emph{$\pi$-base size} of $\agbv$ for 1-type $\pi$, denoted by $\projectionbasepi{\agbv}{\pi}$,
    are defined by
    \begin{equation*}
        \begin{aligned}
            \abs{\agbv}
            \ :=\ &
            \sum\limits_{\pi \in \onetypes}
            \sum\limits_{i \in \intsinterval{\numbhbases{\pi}}}
            \agbv(\pi, \basetext, i)
            +
            \sum\limits_{\pi \in \onetypes}
            \sum\limits_{i \in \intsinterval{\numbhperiods{\pi}}}
            \agbv(\pi, \periodtext, i) \\
            \projectionbase{\agbv}
            \ :=\ &
            \sum\limits_{\pi \in \onetypes}
            \sum\limits_{i \in \intsinterval{\numbhbases{\pi}}}
            \agbv(\pi, \basetext, i) \\
            \projectionbasepi{\agbv}{\pi}
            \ :=\ &
            \sum\limits_{i \in \intsinterval{\numbhbases{\pi}}}
            \agbv(\pi, \basetext, i).
        \end{aligned}
    \end{equation*}
\end{definition}

Intuitively, the \emph{size} of a $\gbv$ $\agbv$ counts the total number of components,
while the \emph{base size} counts only base components.
In particular, the number of concrete vertices in a D-structure realizing $\agbv$ 
is the base size of $\agbv$.
\begin{lemma}\label{lem:num_concreate}
Let $\agbv$ be a $\gbv$, and let $\cG$ be a D-structure realizing $\agbv$.
Then the number of concrete vertices in $\cG$ is exactly $\projectionbase{\agbv}$, and the number of concrete vertices with 1-type $\pi$ in $\cG$ is exactly $\projectionbasepi{\agbv}{\pi}$.
\end{lemma}

\begin{proof}
    Let $\rep$ be a witness of $\cG$ realizing $\agbv$.
    Recall that $\rep$ maps every concrete vertex of $\cG$ with $1$-type $\pi$ to a valid $\lbv$ $\albv$ of $\pi$. And validity means that there exists exactly one $1 \le i \le \numbhbases{\pi}$ such that
    $\albv(\basetext,i)=1$, and $\albv(\basetext,j)=0$ for every $1 \le j \le \numbhbases{\pi}$ with $j\neq i$.

    For every $1 \le i \le \numbhbases{\pi}$, let $V_{\pi,i}$ be the set of concrete vertices $v$ of $\cG$ such that $v$ has $1$-type $\pi$ and $\rep(v)(\basetext,i)=1$.
    Clearly, the sets $V_{\pi,i}$ partition the set of concrete vertices of $\cG$ with $1$-type $\pi$.
    Let $V_\pi$ denote this set.

    Observe that $\abs{V_{\pi,i}} = \sum_{v\in V_\pi}\rep(v)(\basetext,i)$.
    Moreover, by the definition of a realization witness,
    $\agbv(\pi,\basetext,i) = \sum_{v\in V_\pi}\rep(v)(\basetext,i)$.
    Hence,
    \begin{equation*}
        \abs{V_\pi}
        =
        \sum_{i\in\intsinterval{\numbhbases{\pi}}}\abs{V_{\pi,i}}
        =
        \sum_{i\in\intsinterval{\numbhbases{\pi}}}
        \agbv(\pi,\basetext,i)
        =
        \projectionbasepi{\agbv}{\pi}.
    \end{equation*}

    Similarly, the total number of concrete vertices is exactly
    \begin{equation*}
        \sum_{\pi\in\onetypes}\abs{V_\pi}
        =
        \sum_{\pi\in\onetypes}\projectionbasepi{\agbv}{\pi}
        =
        \projectionbase{\agbv}.
    \end{equation*}
\end{proof}

Since the entries of every $\gbv$ are non-negative integers, the size is additive:
$\abs{\agbv+\agbv'} = \abs{\agbv} + \abs{\agbv'}$.
The same property also holds for the base size.

\subsection{Overapproximation of $\gbv$s for a $\ourlogic$ sentence}
\label{subsec:overapproximation}

We now move to the next part of the proof outline, obtaining an \emph{overapproximation Presburger formula}
which will be denoted by $\overapprox{\anfgptwosentence\land\agcc}(\frx)$.
Here, $\frx$ is a tuple of variables used to encode a $\gbv$;
 we index its coordinates in the same way as the entries of a $\gbv$.
\begin{definition}[$\overapprox{\anfgptwosentence \land \agcc}$, $\match$, $\structcon$, $\liftgcc$] \label{def:overapprox}%

We define\\
    $\overapprox{\anfgptwosentence \land \agcc}(\frx)
    \ :=\ 
    \match(\frx)
    \land
    \structcon(\frx)
    \land
    \liftgcc(\frx)$,
where
    \begin{align*}
        \match(\frx)
        \ :=\ &
        \bigwedge_{\substack{\pi_1, \pi_2 \in \onetypes\\ \eta \in \twotypescmpi{\pi_1}{\pi_2}}} 
        \Bigg(
            \sumcon_{\pi_1, \eta, \pi_2}(\frx)
            \ =\ 
            \sumcon_{\pi_2, \tildeeta, \pi_1}(\frx)
        \Bigg) \\
        \sumcon_{\pi_1, \eta, \pi_2}(\frx)
        \ :=\ &
        \Bigg(
            \sum\limits_{i \in \intsinterval{\numbhbases{\pi_1}}}   \frx(\pi_1, \basetext, i)   \cdot \bhbase{\pi_1}{i}(\eta, \pi_2) +
            \sum\limits_{i \in \intsinterval{\numbhperiods{\pi_1}}} \frx(\pi_1, \periodtext, i) \cdot \bhperiod{\pi_1}{i}(\eta, \pi_2)
        \Bigg) \\
        \structcon(\frx)
        \ :=\ &
        \bigwedge_{\pi \in \onetypes} 
        \Bigg(
            \sum\limits_{i \in \intsinterval{\numbhperiods{\pi}}} \frx(\pi, \periodtext, i) > 0
        \Bigg)
        \to
        \Bigg(
            \sum\limits_{i \in \intsinterval{\numbhbases{\pi}}} \frx(\pi, \basetext, i) > 0
        \Bigg) 
    \end{align*}
and $\liftgcc(\frx)$ is obtained from $\agcc$ by replacing every term $\termg{U_i(x)}$ with 
$\sum\limits_{\pi \in \onetypes_i} \sum\limits_{t \in \intsinterval{\numbhbases{\pi}}} \frx(\pi, \basetext, t)$,
where $\onetypes_i$ is the set of conformant 1-types that contain $U_i$.
\end{definition}


The constraint $\structcon(\frx)$ requires that
if a 1-type $\pi$ is incorporated periodically, then $\pi$ must also be incorporated basely. 
This property always holds in a structure, since structures contain no phantom vertices.
The formula $\liftgcc(\frx)$ encodes the $\GCC$s.
It is routine to show that $\overapprox{\anfgptwosentence \land \agcc}(\frx)$ is an overapproximation of
the set of $\gbv$s realized by models of $\anfgptwosentence \land \agcc$.

\begin{lemma}\label{lem:overapprox}
    Let $\cG$ be structure.
    If $\cG$ is a model of $\anfgptwosentence \land \agcc$,
    then for every $\gbv$ $\agbv$ realized by $\cG$,
    $\agbv$ is a solution to $\overapprox{\anfgptwosentence\land\agcc}(\frx)$.
\end{lemma}

\begin{proof}
Let $\cG$ be a model of $\anfgptwosentence \land \agcc$.
In particular, $\univpartphi(a)$ holds in $\cG$ for every element $a$ in $\cG$,
and hence its behavior is in $\bhset{\pi}$ with $\pi$ the 1-type of $a$
and therefore its $\lbv$ is valid.

Let $\rep$ be the witness that maps each element $a$ to its $\lbv$
and let $\agbv$ be the $\gbv$ whose slice on each 1-type $\pi$ is the sum of the $\lbv$ associated by the witness map of $\pi$-elements.
We verify that $\gbv$ satisfies each condition in $\overapprox{\anfgptwosentence\land\agcc}(\frx)$.
\begin{compactitem}
\item\itemtitle{$\match(\frx)$}
Observe that
$\sumcon_{\pi_1, \eta, \pi_2}(\agbv)$
is the number of $\eta$-edges (Definition \ref{def:implicitexplicitetaedge}) from $u$ to $v$ such that
$\onetypeof(u) = \pi_1$ and
$\onetypeof(v) = \pi_2$.
Note that if $e$ is an $\eta$-edge from $u$ to $v$,
then it is also an $\dual{\eta}$-edge from $v$ to $u$,
it follows that
$\sumcon_{\pi_1, \eta, \pi_2}(\agbv)$
and
$\sumcon_{\pi_2, \dual{\eta}, \pi_1}(\agbv)$
are the same. 
Hence, $\match(\agbv)$ holds.

We comment that for this part, we just use that $\agbv$ is a global behavior vector of a structure, regardless of whether the structure satisfies any formula. 
\item\itemtitle{$\structcon(\frx)$}
For every 1-type $\pi \in \onetypes$,
suppose that $\agbv(\pi, \periodtext, j) > 0$ for some $1 \le j \le \numbhperiods{\pi}$.
By the definition of the realization witness for a $\gbv$ in Definition \ref{def:struct_realize_gbv},
the value $\agbv(\pi, \periodtext, j)$ is the sum of the $(\periodtext, j)$-entries of the witness of vertices in $\cG$ satisfying 1-type $\pi$,
which implies that there exists a vertex $v$ with $\rep(v)(\periodtext, j) > 0$.
Since $\rep(v)$ is a valid $\lbv$ (Def. \ref{def:lbv}),
there exists $1 \le i \le \numbhbases{\pi}$, such that $\rep(v)(\basetext, i) > 0$.
By the definition of realization witness of $\gbv$, this implies that
$\sum\limits_{i \in \intsinterval{\numbhbases{\pi}}} \agbv(\pi, \basetext, i) > 0$.

\item\itemtitle{$\liftgcc(\frx)$}
Suppose that $\cG$ satisfies the atomic $\gcc$
$\Sigma_{i \in \intsinterval{k}} \lambda_i \cdot \presbx{U_i(x)} \ge \delta$.
By the semantics of $\gcc$s,
it holds that
$\sum\limits_{i \in \intsinterval{k}} \lambda_i \cdot \abs{U_i^\cG} \ge \delta$.
For every element $v$ in $\cG$ satisfying 1-type $\pi$,
since $\rep(v)$ is a valid $\lbv$,
we have
$\sum\limits_{i \in \intsinterval{\numbhbases{\pi}}} \rep(v)(\basetext, i) = 1$.
Therefore
$\abs{U_i^\cG} = \sum\limits_{\pi \in \onetypes_i} \sum\limits_{t \in \intsinterval{\numbhbases{\pi}}} \agbv(\pi, \basetext, t)$.
The case for a general $\gcc$ can be proven inductively.
It follows that the corresponding lifted $\gcc$ constraint for atomic $\gcc$ holds.
The result for a general $\gcc$ then follows by induction on its structure.
\end{compactitem}
\end{proof}

Although $\overapprox{\anfgptwosentence \land \agcc}(\frx)$ is only an overapproximation,
the only ``excess'' $\gbv$s satisfied in it are those that are not realizable in any structure:
whenever a structure $\cG$ realizes a $\gbv$ satisfying $\overapprox{\anfgptwosentence \land \agcc}(\frx)$,
 $\cG$ is a model of the $\ourlogic$ sentence $\anfgptwosentence \land \agcc$:

\begin{lemma}\label{lem:pseudo_model_tight}
Let $\cG$ be a structure. If $\cG$ realizes some $\gbv$
$\agbv \in \sol{\overapprox{\anfgptwosentence \land \agcc}(\frx)}$,
then $\cG$ is a model of $\anfgptwosentence \land \agcc$.
\end{lemma}

\begin{proof}
Recall that $\anfgptwosentence$ is a $\gptwo$ sentence in normal form,
and $\univpartphi(x)$ is the formula obtained from $\varphi$ by removing the outermost quantifier $\forall x$.
By the definition of conformant behaviors, 
the behavior realized by an element $v$ of $\cG$ is conformant (w.r.t $\anfgptwosentence$)  if and only if $\cG \models \univpartphi(v)$. 

Let $\rep$ be a witness to the realization of $\agbv$ by $\cG$.
By definition, $\rep(v)$ is a valid $\lbv$ for every $v$ in $\cG$.
Hence, the behavior realized by $v$, namely $\eval{\rep(v)}$, is conformant. 
It follows that $\cG \models \univpartphi(v)$ for every $v$ in $\cG$, and therefore
$\cG \models \anfgptwosentence$.

It remains to show that $\cG \models \agcc$.
We use the same correspondence as in the argument for $\liftgcc(\frx)$ in Lemma~\ref{lem:overapprox}.
For every unary predicate $U_i$, 
the validity of $\rep(v)$ ensures that each element of $\cG$ contributes exactly once 
to the base component corresponding to its $1$-type. Consequently,
$\abs{U_i^\cG}
=
\sum\limits_{\pi \in \onetypes_i}
\sum\limits_{t \in \intsinterval{\numbhbases{\pi}}}
\agbv(\pi,\basetext,t)$,
where $\onetypes_i$ is the set of conformant $1$-types containing $U_i$.
Thus, for every atomic $\gcc$ of the form
$\sum\limits_{i \in \intsinterval{k}}
\lambda_i \cdot \termg{U_i(x)} \ge \delta$,
its interpretation in $\cG$ satisfies
$\sum\limits_{i \in \intsinterval{k}}
\lambda_i \cdot \abs{U_i^\cG}
=
\sum\limits_{i \in \intsinterval{k}}
\lambda_i
\sum\limits_{\pi \in \onetypes_i}
\sum\limits_{t \in \intsinterval{\numbhbases{\pi}}}
\agbv(\pi,\basetext,t)$.
Hence, an atomic $\gcc$ holds in $\cG$ if and only if its corresponding lifted constraint holds for $\agbv$.
By induction on the Boolean structure of $\agcc$, we conclude that
$\cG \models \agcc$.
\end{proof}

Finally, we note that if we allow $D$-structures rather than structures, there is no overapproximation:
We begin with the following building block for the matching constraint involving a pair of 1-types and a 2-type.

\begin{restatable}[Rewiring lemma]{lemma}{rewiring} \label{lem:rewiring}
    Let $\pi_1$, $\pi_2$ be 1-types and $\eta$ be an asymmetric 2-type as in Definition \ref{def:conformanttypeasymmetrictype}.
    Let $\cG$ be a D-structure, 
    $V_1$ be the set of concrete and phantom vertices of $\cG$ with 1-type $\pi_1$, and
    $V_2$ be the set of concrete and phantom vertices of $\cG$ with 1-type $\pi_2$.
    Suppose that for every $v \in V_1$, $\bh{v}(\eta, \pi_2) = 0$; and for every
    $u \in V_2$, $\bh{u}(\dual{\eta}, \pi_1) = 0$.
    Let $f_1: V_1 \to \bbN$ and $f_2: V_2 \to \bbN$ be functions satisfying the ``local match equation''
        $\sum_{v \in V_1} f_1(v)
        \ =\ 
        \sum_{v \in V_2} f_2(v)$.

    Then there exists a D-structure $\cH$ with the same sets of concrete and phantom vertices as $\cG$,
    obtained by adding edges of 2-type $\eta$ between vertices in $V_1$ and $V_2$,
    such that:
    \begin{itemize}
        \item
            For every vertex $v$, $\onetypeof_\cH(v) = \onetypeof_\cG(v)$.
        \item
            For every vertex $v \in V_1$,
            $\bhg{\cH}{v}(\eta, \pi_2) = f_1(v)$,
            while all other entries of $\bhg{\cH}{v}$ coincide with the
            corresponding entries of $\bhg{\cG}{v}$.
        \item
            For every vertex $v \in V_2$,
            $\bhg{\cH}{v}(\dual{\eta}, \pi_1) = f_2(v)$,
            while all other entries of $\bhg{\cH}{v}$ coincide with the
            corresponding entries of $\bhg{\cG}{v}$.
        \item
            For every vertex $v$ such that $\onetypeof_\cH(v) \not\in \set{\pi_1, \pi_2}$,
            $\bhg{\cH}{v} = \bhg{\cG}{v}$.
    \end{itemize}

    Furthermore, if $\cG$ is asymmetric, the new structure is asymmetric as well.
\end{restatable}

Intuitively, the functions $f_1$ and $f_2$ specify the intended behavior entries of the vertices in $V_1$ and $V_2$, respectively:
$f_1(v)$ specifies the intended $(\eta,\pi_2)$-entry for each $v \in V_1$,
while $f_2(v)$ specifies the intended $(\dual{\eta},\pi_1)$-entry for each $v \in V_2$.
We show that, whenever these functions satisfy the ``local match equation'',
the D-structure can be extended so that every vertex in $V_1 \cup V_2$ realizes its intended behavior entry.

One application of Lemma~\ref{lem:rewiring} is the following lemma concerning $\match$.

\begin{restatable}[Assembling a D-structure from a $\match$ solution]{lemma}{constructDstruct}
\label{lem:match}
For every $\agbv$ satisfying $\match$, 
there exists a D-structure $\cG$ realizing $\agbv$. 
In particular, this holds for any $\agbv$
satisfying $\overapprox{\anfgptwosentence \land \agcc}$.
\end{restatable}

The proof is also deferred to Appendix \ref{app:dstructurefrommatch}.

The gap between $\sol{\overapprox{\anfgptwosentence \land \agcc}(\frx)}$ and 
the set of $\gbv$s realizing models of $\anfgptwosentence \land \agcc$ is that
$\sol{\overapprox{\anfgptwosentence \land \agcc}(\frx)}$ may contain $\gbv$s that 
are not realized by any structure.
Thus, by Lemma~\ref{lem:pseudo_model_tight},
we can reduce the finite satisfiability problem of $\anfgptwosentence \land \agcc$ to 
the \emph{realizability problem of $\sol{\overapprox{\anfgptwosentence \land \agcc}(\frx)}$}.

Note that $\overapprox{\anfgptwosentence \land \agcc}(\frx)$ is an existential Presburger formula.
Hence, its solution set is semilinear.
We discuss the property of the solution set of $\overapprox{\anfgptwosentence \land \agcc}(\frx)$ 
in the next section. 

\subsection{King and peasant $1$-types, and concretifying bases}

We next investigate the properties of a linear subset of $\sol{\overapprox{\anfgptwosentence \land \agcc}(\frx)}$, 
identifying which $1$-types have bounded occurrences. 
We have our overapproximation $\overapprox{\anfgptwosentence \land \agcc}(\frx)$, and we choose
one linear set within it.
This is specified by  a $\gbv$ $\frb$ along with a finite set of $\gbv$s
$\frP = \set{\frp_1, \ldots, \frp_p}$  -- 
that is, $\cL(\frb;\frP)\subseteq \sol{\overapprox{\anfgptwosentence \land \agcc}(\frx)}$.
We write $\numouterperiods$ for $p$: 
the number of ``outer periods'', as opposed to the basis vectors we had for local behaviors.
And we write $\cL(\frb;\frP)$ for the linear set,
the set of $\gbv$s in the form $\frb + \sum\limits_{i \in \intsinterval{\numouterperiods}} n_i\frp_i$.

To characterize which $\gbv$s in $\cL(\frb;\frP)$ are realized by structures,
we introduce the notion of king and peasant 1-types w.r.t. $\frb$ and $\frP$. Here we refer to the definition of a $\gbv$ incorporating a one type
basely from Definition \ref{def:gbvincorporates}.
\begin{definition}[King 1-type, peasant 1-type w.r.t. $\frb$ and $\frP$]\label{def:kingandpeasantonetypeforbcvs}
\hfill
\begin{compactitem}
\item
A 1-type $\pi$ is called a \emph{king 1-type} w.r.t. $\frb$ and $\frP$ if 
$\pi$ is incorporated basely by $\frb$ and is not incorporated basely  by any $\frp_i\in \frP$.
\item
A 1-type $\pi$ is called a \emph{peasant 1-type} w.r.t. $\frb$ and $\frP$ if 
$\pi$ is incorporated basely by some $\frp_i \in \frP$.
\end{compactitem}

We denote the set of king 1-types w.r.t. $\frb$ and $\frP$ by $\onetypeking{\frb}{\frP}$
and the set of peasant 1-types w.r.t. $\frb$ and $\frP$ by $\onetypepeasant{\frb}{\frP}$.
\end{definition}

The following notion, \emph{concretifying}, addresses the requirement that
the existence of a phantom vertex with 1-type $\pi$ implies
the existence of a concrete vertex with the same 1-type.
This is crucial when we want to transform a D-structure to a structure.
Recall that in a D-structure that realizes a $\gbv$,
the witnessed-assigned $\lbv$ of a phantom vertex must be homogeneously valid so that it can be merged
to a concrete vertex with the same 1-type.


\begin{definition}[Concretifying $\frb$ and $\frP$]\label{def:concretifying}
We say that $\frb$ and $\frP$ are \emph{concretifying} if the following condition holds.
\begin{compactitem}
\item \itemtitle{No unsupported phantoms in $\frP$}
For every 1-type $\pi$,
for every $1 \le i \le \numouterperiods$,
if $\pi$ is incorporated periodically in $\frp_i$, 
then $\pi$ is also incorporated basely in $\frb$ or in some $\frp_j\in\frP$.
\end{compactitem}
\end{definition}
We can then prove that every basis of a linear subset of the solution set of 
$\overapprox{\anfgptwosentence \land \agcc}(\frx)$ is concretifying.

The following lemma follows from the fact that 
$\cL(\frb;\frP)\subseteq \sol{\overapprox{\anfgptwosentence \land \agcc}(\frx)}$.
\begin{lemma}\label{lem:sol_concretifying}
 Let $\frb$ be a $\gbv$ 
 and let
$\frP = \set{\frp_1, \ldots, \frp_{\numouterperiods}}$
 be a finite set of $\gbv$s 
 such that
 $\cL(\frb; \frP)
 \ \subseteq\ 
 \sol{\overapprox{\anfgptwosentence \land \agcc}(\frx)}$.
 Then $\frb$ and $\frP$ are concretifying.
\end{lemma}
\begin{proof}
We verify that $\frb$ and $\frP$ satisfy the condition of concretifying.
\begin{compactitem}
\item \itemtitle{No unsupported phantoms in $\frP$}
For each $1 \le i \le \numouterperiods$, we have
$\frb + \frp_i \in \sol{\overapprox{\anfgptwosentence \land \agcc}(\frx)}$.
Hence $\frb + \frp_i$ satisfies $\structcon(\frx)$, from which the required condition follows.
\end{compactitem}
\end{proof}

\begin{lemma}\label{lem:concretifying}
Suppose that $\frb$ and $\frP$ are concretifying.
For every $\frb' \in \cL(\frb; \frP)$ and $\frP' \subseteq \frP$,
$\frb'$ and $\frP'$ are concretifying.
\end{lemma}

\begin{proof}
    Let $\frb' := \frb + \sum\limits_{i \in \intsinterval{\numouterperiods}} a_i \frp_i$.
    We verify that $\frb'$ and $\frP'$ satisfy each condition.
    \begin{compactitem}

        \item \itemtitle{No unsupported phantoms in $\frP'$}
            The property follows immediately from the observation that
            if $\frb$ incorporates $\pi$ basely, then $\frb'$ also incorporates $\pi$ basely.
    \end{compactitem}
\end{proof}

Recall that, for every $\gbv$ $\agbv \in \cL(\frb;\frP)$, the base vector $\frb$ appears only once.
Therefore, if a 1-type $\pi$ is  incorporated basely only in $\frb$ and not basely in any period vector in $\frP$,
then every structure realizing $\agbv$ contains only a bounded number of vertices of 1-type $\pi$.
We therefore call these \emph{king $1$-types} (Def. \ref{def:kingandpeasantonetypeforbcvs}): as we promised in Section \ref{sec:outline}, such types are controlled entirely by the bounded base part $\frb$, so their cardinality is bounded.
Moreover, this bound can be lifted to the edges between king $1$-types.
Since there are no multi-edges between vertices of a structure, the number of such edges is also bounded.
As a consequence, we can also upper bound the number of period $\gbv$s whose induced behavior incorporates such an edge.
\begin{lemma}[Bounding entries for king types] \label{lem:bound_king}
    For every $1 \le t \le \numouterperiods$,
    if $\frp_t$ incorporates a 2-type $\eta$ between king 1-types $\pi_1$ and $\pi_2$,
    then for every
    $\agbv :=
    \frb + \sum\limits_{i \in \intsinterval{\numouterperiods}} a_i \frp_i
    \in \cL(\frb; \frP)$
    with $a_t > \left(\projectionbase{\frb}\right)^2$,
    there is no structure realizing $\agbv$.
\end{lemma}

\begin{proof}
    We argue by contradiction.
    Suppose that $\cG$ is a structure realizing $\agbv$ with $a_t > \left(\projectionbase{\frb}\right)^2$,
    and $\frp_t$ incorporates a 2-type $\eta$ between king 1-types $\pi_1$ and $\pi_2$.

    Since $\pi_1$ and $\pi_2$ are king 1-types,
    they are incorporated basely only in $\frb$.
    Hence, by Lemma~\ref{lem:num_concreate}, $\cG$ contains at most $\projectionbase{\frb}$ vertices of 1-type $\pi_1$ or $\pi_2$.

    On the other hand, $a_t > \left(\projectionbase{\frb}\right)^2$ implies that $\cG$ contains
    more than $\left(\projectionbase{\frb}\right)^2$ edges with 2-type $\eta$ from vertices with 1-type $\pi_1$ to vertices with 1-type $\pi_2$.
    However, since $\cG$ is a structure, it contains no multi-edges.
    Therefore, there can be at most $\left(\projectionbase{\frb}\right)^2$
    edges from vertices with 1-types $\pi_1$ to vertices with 1-types $\pi_2$, a contradiction.
\end{proof}

The lemma tells us that if we are interested in structures, we can impose the following restriction on the family of $\gbv$s, \emph{strongly concretifying} $\frb$ and $\frP$.

\begin{definition}[Strongly concretifying $\frb$ and $\frP$]\label{def:strong_concretifying}
We say that $\frb$ and $\frP$ are \emph{strongly concretifying} if they are concretifying and satisfy: 
\begin{compactitem}
\item \itemtitle{No edges incorporated between king 1-types in $\frP$}
For every $\frp_t \in \frP$, 
every pair of king 1-types $\pi_1$ and $\pi_2$,
and every 2-type $\eta$,
$\frp_t$ does not incorporate the $2$-type $\eta$ between $\pi_1$ and $\pi_2$, 
in the sense of Definition~\ref{def:gbvincorporates}.
\end{compactitem}
A linear set $\cL(\frb;\frP)$ is \emph{strongly concretifying}
if $\frb$ and $\frP$ are strongly concretifying.
\end{definition}

Unlike with concretifying, we cannot prove that 
every linear subset of the solution set of $\overapprox{\anfgptwosentence \land \agcc}(\frx)$ 
is strongly concretifying.
Instead, we show that the realizable $\gbv$s in this solution set can be covered by 
finitely many strongly concretifying linear sets,
the representations of which are computable from $\overapprox{\anfgptwosentence \land \agcc}(\frx)$.

The key idea is as follows.
Suppose that some period vector $\frp_i$ does not satisfy the required condition.
Then, by Lemma~\ref{lem:bound_king}, there exists a threshold such that no $\gbv$ containing sufficiently many copies of $\frp_i$ can be realized by a structure.
Consequently, it suffices to consider only finitely many copies of $\frp_i$.
By explicitly enumerating all possible numbers of copies below the threshold, we can replace the original linear set with finitely many linear sets, each of which satisfies the required condition.

\begin{lemma}\label{lem:strongly_concretifying}
Suppose that $\frb$ and $\frP$ are concretifying but not strongly concretifying.
Then, for every realizable $\gbv$ $\agbv \in \cL(\frb; \frP)$,
there exist $\frb' \in \cL(\frb; \frP)$ and $\frP' \subsetneq \frP$
such that
\begin{compactitem}
\item
$\frb'$ and $\frP'$ are strongly concretifying;
\item
$\abs{\frb'} \le \left(2M\abs{\frb}\right)^{2^{\numouterperiods}}$,
where $M$ is the maximum of $\abs{\frp_i}$ for $1 \le i \le \numouterperiods$;
\item
$\agbv \in \cL(\frb'; \frP')$.
\end{compactitem}
\end{lemma}

\begin{proof}
Let $\agbv := \frb + \sum\limits_{i\in\intsinterval{\numouterperiods}} a_i\frp_i$.
Choose an index $1\le t\le\numouterperiods$ such that $\frp_t$ violates the condition 
\emph{no edges between king 1-types in period vectors}.
Since $\agbv$ is realizable, Lemma~\ref{lem:bound_king} implies that
$a_t
\ \le\ 
\projectionbase{\frb}^2$.

We absorb all copies of $\frp_t$ into the base vector by defining
$ \frb_1
\ :=\ 
\frb+a_t\frp_t,
\quad
\frP_1
\ :=\ 
\frP\setminus\set{\frp_t}$.
Then $\agbv \in \cL(\frb_1; \frP_1)$, and
$\abs{\frb_1} = \abs{\frb} + a_t\abs{\frp_t}
\le \abs{\frb} + \projectionbase{\frb}^2 \abs{\frp_t} \ \le\ 2M\abs{\frb}^2$,
where $M$ is the maximum of $\abs{\frp_i}$ for $1 \le i \le \numouterperiods$.
Moreover, by Lemma~\ref{lem:concretifying}, the $\frb_1$ and $\frP_1$ are still concretifying.

Note that the cardinality $\frP'$ is strictly less than the cardinality of $\frP$.
We may therefore repeat the same argument.
Each iteration removes one period vector from $\frP$
while preserving the containment of $\agbv$ 
and the concretifying property by Lemma~\ref{lem:concretifying}.
Hence, after at most $\numouterperiods$ iterations,
we obtain $\frb'$ and $\frP'$ such that no period vector in $\frP'$ violates the strongly concretifying condition.
Therefore, $\frb'$ and $\frP'$ are strongly concretifying.

It remains to bound the size of $\frb'$.
The sequence $\frb_0, \frb_1, \ldots, \frb_k$
with $k \le \numouterperiods$,
$\abs{\frb_0} = \abs{\frb}$, and
$\abs{\frb_k} = \abs{\frb'}$
satisfies the recurrence
$\abs{\frb_{i+1}} \le 2M\abs{\frb_i}^2$.
Consequently, $\abs{\frb'} = \abs{\frb_k} \le \left(2M\abs{\frb}\right)^{2^\numouterperiods}$.
\end{proof}

\begin{lemma}\label{lem:strongly_concretifying_cover}
    Suppose that $\frb$ and $\frP$ are concretifying but not strongly concretifying.
    Then we can construct linear sets of $\gbv$s
    $\set{
        \cL(\frb_1; \frP_1),
        \ldots,
        \cL(\frb_k; \frP_k)
    }$,
    such that
    \begin{compactitem}
        \item
            for every $1 \le i \le k$, $\frb_i$ and $\frP_i$ are strongly concretifying and
            $\cL(\frb_i; \frP_i) \subseteq \cL(\frb; \frP)$;

        \item
            every realizable $\gbv$ in $\cL(\frb; \frP)$ belongs to one of these linear sets.
    \end{compactitem}
\end{lemma}

\begin{proof}
    By Lemma~\ref{lem:strongly_concretifying}, every realizable $\gbv$ in $\cL(\frb; \frP)$ belongs to a linear set
    $\cL(\frb'; \frP')$, where $\frb'$ and $\frP'$ are strongly concretifying,
    $\abs{\frb'}$ is bounded by the bound given in Lemma~\ref{lem:strongly_concretifying},
    and $\frP' \subseteq \frP$.

    We construct the finite family of representations of linear sets of $\gbv$s using the following procedure.
    We enumerate all $\frb' \in \cL(\frb; \frP)$ satisfying this size bound, together with all subsets $\frP' \subseteq \frP$.
    For each pair of $\frb'$ and $\frP'$, we test whether $\frb'$ and $\frP'$ are strongly concretifying and retain those that satisfy the condition.
    Since the strongly concretifying property is effectively decidable, this procedure constructs a finite family of linear sets effectively.

    It is clear that every linear set is contained in $\cL(\frb; \frP)$ by construction.
    Moreover, by Lemma~\ref{lem:strongly_concretifying}, every realizable $\gbv$ in $\cL(\frb; \frP)$ belongs to one of these linear sets.
\end{proof}

\begin{lemma}\label{lem:sol_semilinear_strong}
    Given a  sentence $\anfgptwosentence \land \agcc$ in normal form, 
    we can  construct a finite family of linear sets of $\gbv$s
    $\set{
        \cL(\frb_1; \frP_1),
        \ldots,
        \cL(\frb_m; \frP_m)
    }$,
    such that
    \begin{compactitem}
        \item
            for every $1 \le i \le m$,
            $\frb_i$ and $\frP_i$ are strongly concretifying and\
            $\cL(\frb_i; \frP_i)$ is contained in $\sol{\overapprox{\anfgptwosentence \land \agcc}(\frx)}$

        \item
            every $\gbv$ realized by a model of $\anfgptwosentence \land \agcc$ belongs to one of these linear sets.
    \end{compactitem}
\end{lemma}

\begin{proof}
    Since $\overapprox{\anfgptwosentence \land \agcc}(\frx)$ is an existential Presburger formula,
    its solution set is semilinear, and a semilinear representation can be computed effectively.
    By Lemma~\ref{lem:sol_concretifying}, every linear set in this representation is concretifying.

    We then apply Lemma~\ref{lem:strongly_concretifying_cover} to each concretifying linear set.
    For each such linear set, the lemma produces a finite family of strongly concretifying linear sets whose union contains every realizable $\gbv$ in the original linear set.
    Applying this transformation to every linear set in the semilinear representation yields the desired finite family.

    Since every step of the construction is effective, the resulting family can be computed effectively.
\end{proof}


%% file: notationtable.tex

\begin{table}[h!]
\center
\begin{tblr}{
        colspec={c|l},
        hlines
    }
    $\onetypes$                   & set of 1-types  conformant with $\anfgptwosentence$ (Def. \ref{def:conformanttypeasymmetrictype}) \\
    $\twotypescmpi{\pi_1}{\pi_2}$ & set of non-null 2-types conformant with $\anfgptwosentence$ and 1-types $\pi_1$ and $\pi_2$  (Def. \ref{def:conformanttypeasymmetrictype}) \\
    $\bhset{\pi}$                 & set of behaviors of 1-type $\pi$ (Def. \ref{def:localbehbasis}) \\
    $\bhsethomo{\pi}$             & homogeneous set of behaviors of 1-type $\pi$ (Def. \ref{def:localbehbasis})\\
    $\bhbase{\pi}{i}$             & base vector for the behaviors of 1-type $\pi$ (Def. \ref{def:localbehbasis})\\
    $\bhperiod{\pi}{i}$           & period vector for the behaviors of 1-type $\pi$ (Def. \ref{def:localbehbasis})\\
    $\numbhbases{\pi}$            & number of  base vectors for the behaviors of 1-type $\pi$ (Def. \ref{def:localbehbasis}) \\ 
    $\numbhperiods{\pi}$          & number of period vectors for the behaviors of 1-type $\pi$  (Def. \ref{def:localbehbasis}) \\
    $\eval{\albv}$                & evaluation of the local behavior vector  $\albv$  (Def. \ref{def:lbv})\\
\end{tblr}
\caption{Summary of notation.}
\label{table:notation}
\end{table}


%% file: 4-realization.tex

\subsection{Kernels, villages, and village systems}\label{subsec:realizable-gbv}

We continue to \emph{fix a representation for one of the linear components of our  overapproximating semilinear set, consisting of $\gbv$ $\frb$ and
a finite set of $\gbv$s $\frP = \set{\frp_1, \ldots, \frp_\numouterperiods}$.} Note that this
determines \emph{a set of king and peasant $1$-types} (Definition \ref{def:kingandpeasantonetypeforbcvs}), which we are thus also fixing.

We can assume that $\frb$ and $\frP$ are strongly concretifying by Lemma~\ref{lem:sol_semilinear_strong}. 
We aim to determine the condition under which there exists a structure realizing some $\gbv$ in $\cL(\frb; \frP)$.

Recall the intuition from Section \ref{sec:outline}.
We want to show that there exists a structure realizing some $\gbv \in \cL(\frb;\frP)$ if and only if
it can be decomposed into two parts: a bounded-size \emph{kernel} and an unbounded-size \emph{village}.
Although the size of the village is unbounded,
we will further show that the existence of a village can be characterized by a bounded-size existential Presburger formula,
which can be constructed from $\frb$ and $\frP$ effectively.

We now formalize the intuitions behind kernel and village, which were illustrated in Figure \ref{fig:kernel_village}.
\begin{definition}[Kernel, king, noble, and phantom noble] \label{def:kernalkingnoblephantomnoble}
    We say that a D-structure $\cK$ is a \emph{kernel} of $\frb$ and $\frP$ if it is asymmetric and
    realizes a $\gbv$ in $\cL(\frb;\frP)$ and satisfies the following conditions.
    \begin{compactitem}
        \item[\emph{K1.}]
            $\cK$ contains no phantom vertices realizing 
            1-types in $\onetypeking{\frb}{\frP}$. 

            We call the concrete vertices of $\cK$ realizing 1-types in $\onetypeking{\frb}{\frP}$ \emph{kings}.
            The remaining concrete vertices realizing 1-types in $\onetypepeasant{\frb}{\frP}$ are called \emph{nobles},
            while the phantom vertices realizing  1-types in $\onetypepeasant{\frb}{\frP}$ are called \emph{phantom nobles}.

        \item[\emph{K2.}]
            Referring to the definitions in K1 above, no multi-edge of $\cK$ involves a king.
            No edge of $\cK$ connects a king and a phantom noble.
           A king  of $\cK$ has no self-loop.
    \end{compactitem}
\end{definition}

\begin{definition}[Village, shadow, and peasant] \label{def:villageshadowpeasant}
    We say that a D-structure  $\cV$ is a \emph{village} of $\frb$ and $\frP$ if it is asymmetric,
    it realizes a $\gbv$ in $\cL_\homo(\frP)$, and satisfies the following conditions.
    \begin{compactitem}

        \item[\emph{V1.}]
            Every phantom vertex of $\cV$ realizes a  1-type in $\onetypeking{\frb}{\frP}$.
            Every concrete vertex of $\cV$ realizes a  1-type in $\onetypepeasant{\frb}{\frP}$.

            We call the phantom vertices of $\cV$ \emph{shadows},
            and the concrete vertices of $\cV$ \emph{peasants}.

        \item[\emph{V2.}]
            No multi-edge of $\cV$ involves a shadow.
            No edge of $\cV$ connects two shadows.
        There is no self-loop in $\cV$.
    \end{compactitem}
\end{definition}

\begin{definition}[Compatibility between kernel and village] \label{def:kernel_village_comp}
    Let $\cK$ be a kernel of $\frb$ and $\frP$ and
    let $\cV$ be a village of $\frb$ and $\frP$.
    We say that $\cK$ and $\cV$ are compatible if: 
    \begin{compactitem}
        \item[\emph{V3.}] \itemtitle{Shadow--king correspondence}
            There exists a bijection
            from shadows of $\cV$ to kings of $\cK$
            that preserves 1-types. Equivalently, for every $1$-type, the number of shadows of $\cV$ that satisfy it in the village is the same as the number of kings that satisfy it in the kernel.

        \item[\emph{V4.}] \itemtitle{Absorption of phantom nobles in $\cK$}
            For every phantom noble of $\cK$ realizing the 1-type $\pi$,
            $\cV$ contains a peasant realizing $\pi$.

        \item[\emph{V5.}] \itemtitle{Absorption of multi-edges in $\cK$}
            For every vertex $v_1$ and $v_2$ in $\cK$,
            if there exists distinct $\eta_1$-edge $e_1$ and $\eta_2$-edge $e_2$ from $v_1$ to $v_2$,
            then $\cV$ contains an $\eta$-edge $e$ from $u_1$ to $u_2$ which satisfies 
            $\onetypeof(u_1) = \onetypeof(v_1)$,
            $\onetypeof(u_2) = \onetypeof(v_2)$, and
            $\twotypeof(e) \in \set{ \twotypeof(e_1), \twotypeof(e_2) }$.
         
    \item[\emph{V6.}] \itemtitle{Absorption of self-loops in $\cK$}
            For every self-loop $e$ for $v$ in $\cK$,
            $\cV$ contains an edge $e'$, such that
            letting $u_1$ be the source vertex of $e'$ and $u_2$ be the destination of $e'$ , we have: 
                $\onetypeof(u_1)
                \ =\ 
                \onetypeof(v), 
                \quad
                \onetypeof(u_2)
                \ =\ 
                \onetypeof(v)$,
            and
                $\twotypeof(e') = \twotypeof(e)$.
        
    \end{compactitem}
\end{definition}

We now formally state the pumping theorem mentioned in the proof sketch:

\begin{theorem}[Pumping] \label{thm:pump}
    Let $\cK$ be a kernel of $\frb$ and $\frP$, and let $\cV$ be a village of $\frb$ and $\frP$.
    If $\cK$ and $\cV$ are compatible, then there exists an asymmetric structure $\cG$ realizing a $\gbv$ in $\cL(\frb;\frP)$.
\end{theorem}

\begin{proof}
Let $k$ be the number of kings in $\cK$. 
The construction consists of several steps.
We begin with the disjoint union of $\cK$ and copies of $\cV$. We then perform the following operations
\begin{compactitem}
\item \itemtitle{Eliminating shadows.}
For each copy of $\cV$ and each shadow $v$ in that copy,
merge $v$ into the king $f(v)$ of $\cK$, where $f$ is the bijection from the shadows of $\cV$ to the kings of $\cK$ that preserves 1-types, given
by (V3).
Note that there cannot be an edge from a shadow $v$ to $f(v)$.
Moreover, for each king $u$ in $\cK$, for every copy of the village, there is exactly one shadow from that copy merged into $u$.
Since there are no edges between shadows in different copies,
the merge procedure cannot create a self-loop, as required for a structure.

\item \itemtitle{Eliminating phantom nobles.}
For each phantom noble in $\cK$ of 1-type $\pi$, merge it into a peasant of 1-type $\pi$ in one of the copies of $\cV$.

We do not have edges between phantom nobles and peasants. With sufficiently many copies (as many as there are phantom nobles), we can ensure
that we do not merge two phantom nobles to the same peasant. Thus again there is no self-loop created in this process.

\item \itemtitle{Eliminating multi-edges in the villages.}
For each multi-edge in a copy of $\cV$, repeatedly apply the swapping operation to the corresponding pair of vertices across the other copies.
Since there are sufficiently many copies, every multi-edge can be eliminated.

\item \itemtitle{Eliminating multi-edges in the kernel.}
For each multi-edge in $\cK$, apply the swapping operation using 
a suitable pair of peasants from one of the copies of $\cV$ with the corresponding 1-types.

\item \itemtitle{Eliminating self-loops in the kernel.}
For each self-loop in $\cK$, apply the swapping operation, using 
a suitable pair of peasants from one of the copies of $\cV$ with the corresponding 1-types.
\end{compactitem}

Let us check that these operations are possible and achieve the required goals.
Because $\cK$ and $\cV$ are compatible, 
the \emph{absorption of phantom nobles in $\cK$} (V4), 
\emph{absorption of multi-edges in $\cK$} (V5) and 
\emph{absorption of self-loops in $\cK$} (V6)
conditions
ensure that the required peasants in $\cV$ always exist.
Moreover, since the numbers of phantom nobles and multi-edges in $\cK$ are finite, by taking sufficiently many copies of $\cV$,
every phantom noble and every multi-edge in $\cK$ can be eliminated.
Consequently, the resulting D-structure is phantom-free and without multi-edges, and hence is a structure.

Since $\cK$ realizes a $\gbv$ in $\cL(\frb;\frP)$ and $\cV$ realizes a $\gbv$ in $\cL_\homo(\frP)$,
the disjoint union of $\cK$ and sufficiently many copies of $\cV$ also realizes a $\gbv$ in $\cL(\frb;\frP)$.
By Lemma~\ref{lem:operation}, each of the above operations preserves the $\gbv$ realized by the D-structure.
Therefore, the resulting structure also realizes a $\gbv$ in $\cL(\frb;\frP)$.
\end{proof}

As outlined in Section~\ref{sec:outline}, our next goal is to construct an existential Presburger formula, called the \emph{village system},
and then show a three-way equivalence:
existence of a village for a given kernel, solvability of the village system, and existence of a structure realizing an appropriate $\gbv$.

We first establish a basic property of kernels of $\frb$ and $\frP$.
\begin{lemma}\label{lem:kernel}
    Let $\cG$ be a D-structure realizing a $\gbv$ $\agbv \in \cL(\frb;\frP)$.
    For every king 1-type $\pi$,
    $\cG$ contains exactly $\projectionbasepi{\frb}{\pi}$ concrete vertices of 1-type $\pi$.

    Moreover, let $\rep$ be a realization witness of $\cG$ for $\agbv$.
    Then, for every king 1-type $\pi$, for every $1 \le i \le \numbhbases{\pi}$,
    there are exactly $\frb(\pi,\basetext,i)$ vertices $v$ of 1-type $\pi$ such that
    $\rep(v)(\basetext,i)=1$.
\end{lemma}

\begin{proof}
    By the definition of a king 1-type, for every
    $1 \le j \le \numouterperiods$, we have $\projectionbasepi{\frp_j}{\pi}=0$.
    By Lemma~\ref{lem:num_concreate}, the number of vertices of 1-type $\pi$ in $\cG$ is
    \begin{equation*}
        \projectionbasepi{\agbv}{\pi}
        \ =\ 
        \projectionbasepi{\frb}{\pi}
        +
        \sum_{j \in \intsinterval{\numouterperiods}} a_j\projectionbasepi{\frp_j}{\pi}
        \ =\ 
        \projectionbasepi{\frb}{\pi}.
    \end{equation*}

    For the second claim, since $\rep$ is a realization witness for $\agbv$, we have
    \begin{equation*}
        \agbv(\pi,\basetext,i)
        \ =\ 
        \sum_{\substack{v \in V\\
                        \onetypeof(v)=\pi}}
            \rep(v)(\basetext, i).
    \end{equation*}
    Since $\pi$ is a king 1-type, for every $1 \le j \le \numouterperiods$, $\frp_j(\pi,\basetext,i)=0$.
    Hence $\agbv(\pi,\basetext,i) = \frb(\pi,\basetext,i)$.
    Moreover, since $\rep(v)$ is valid $\lbv$ for every vertex $v$, each entry $\rep(v)(\basetext,i)$ is either $0$ or $1$.
    Therefore, $\agbv(\pi,\basetext,i)$ is precisely the number of vertices $v$ of 1-type $\pi$ satisfying
    $\rep(v)(\basetext,i) = 1$.
    Consequently, there are exactly $\frb(\pi,\basetext,i)$ such vertices.
\end{proof}

\begin{definition}[Base Behavior of an $\lbv$] \label{def:basebehavior}
Given $\lbv$ $\albv$ of 1-type $\pi$, its \emph{base behavior} is 
$\sum_{j \in \intsinterval{\pi}} \albv(\basetext, j) \cdot \bhbase{\pi}{j}$.
\end{definition}
We note a consequence of the lemma, which we will use throughout the subsequent arguments.
\begin{corollary} \label{cor:kingonetypesindependentofdstructure}
    For every D-structure realizing a $\gbv$ $\agbv \in \cL(\frb;\frP)$,
    the multiset of 1-types and the corresponding base behaviors  of its concrete vertices with king 1-types is the same. 
    In particular, this multiset is determined solely by $\frb$ and $\frP$, independently of the choice of D-structure and realization witness.
\end{corollary}

\begin{definition}[$k$, $\pibar_1,\ldots,\pibar_k$,
    and $\bfbbar_1,\ldots,\bfbbar_k$] \label{def:kpibarbfpibar}
 Based on this observation, for convenience,
    we can define $k := \sum_{\pi \text{ a king 1-type}} \projectionbasepi{\frb}{\pi}$ be the number of concrete vertices with king 1-types in some (equivalently, in any, by Corollary \ref{cor:kingonetypesindependentofdstructure}) D-structure realizing a $\gbv$ $\agbv \in \cL(\frb;\frP)$. 
    We fix an enumeration $\pibar_1,\ldots,\pibar_k$ of these 1-types
    and $\bfbbar_1,\ldots,\bfbbar_k$ their corresponding base behaviors of $\rep(v)$.
\end{definition}
To obtain a Presburger description of a village $\cV$, we describe its shadows and peasants separately.

To describe shadows,
recall that if a village and a kernel $\cK$ are compatible,
the number of shadows is fixed and their 1-types depend only on the kernel $\cK$.
Thus, we can enumerate the shadows and their 1-types and encode their $\lbv$s by variables.

Recall that the number of possible $\lbv$s is unbounded.
However, for a village to be compatible with a fixed kernel
the number of shadows and the 1-types of shadows are fixed.
Thus, the number of possible behaviors of vertices in the village toward the shadows is bounded.
This motivates the notion of king-reaching $\lbv$ in Definition~\ref{def:king-restricted-lbvwellbehaved}.

\begin{definition}[King-reaching]
A behavior $\abehave$ is \emph{king-reaching} with respect to $\frb$ and $\frP$
if there exist a 2-type $\eta$ and a king 1-type
$\pi \in \onetypeking{\frb}{\frP}$
such that $\abehave(\eta, \pi) > 0$.
\end{definition}

\begin{definition}[King-restricted $\lbv$, well-behaved]\label{def:king-restricted-lbvwellbehaved}
Let $\pi$ be a peasant 1-type in $\onetypepeasant{\frb}{\frP}$.
A \emph{king-restricted $\lbv$ ($\krlbv$)} of $\pi$ with respect to $\frb$ and $\frP$ 
is a pair
$\akrlbv := \tuple{\albv,\tuple{\eta_1, \ldots, \eta_k}}$,
where $\albv$ is a valid $\lbv$ and each $\eta_i$ is an ordinary (non-null) 2-type or the null $2$-type.

We use the same notation $\eval{\akrlbv}$ to indicate the behavior generated by $\albv$ as in Definition~\ref{def:lbv}.
For $x \in \set{\basetext, \periodtext}$ and $1 \le i \le \numbhxs{\pi}$,
we write $\akrlbv(x, i)$ for $\albv(x, i)$.
For $1 \le j \le k$, we write $\twotypetoking{j}(\akrlbv)$ for $\eta_j$.

We say that $\akrlbv$ is \emph{well-behaved} if it satisfies the following conditions.
\begin{compactitem}
\item \itemtitle{King-reaching Periods}
For every $1 \le j \le \numbhperiods{\pi}$,
if $\albv(\periodtext, j) > 0$,
then the behavior $\bhperiod{\pi}{j}$ is king-reaching,

\item \itemtitle{Cardinality linking}
For every 2-type $\eta$ and every king 1-type $\pi'$,
$\eval{\akrlbv}(\eta, \pi')$
equals the number of indices
$1 \le j \le k$
satisfying
$\twotypetoking{j}(\akrlbv) = \eta$ and
$\pibar_j = \pi'$.
\end{compactitem}

\end{definition}

Informally, $\akrlbv$ is the king-restricted behavior of a peasant if
the peasant has behavior $\albv$ and is connected to a shadow of 1-type $\pibar_j$ by 
an edge of 2-type $\twotypetoking{j}(\akrlbv)$ for every $1 \le j \le k$, or they are not connected to
such a shadow at all in case 2-type $\twotypetoking{j}(\akrlbv)$ is the null $2$-type.

Intuitively, a peasant in a village is king-reaching  if
it is connected to a shadow by an edge. 
For example, in Figure~\ref{fig:kernel_village}
vertices $12$ and $13$ are king-reaching, while 
vertices $14$ and $15$ are not.
Suppose $\pibar_1, \ldots, \pibar_k$ are the 1-types of the shadows.

\begin{lemma}
    Let $\pi$ be a peasant $1$-type.
    We can compute, from $\pi$, $\pibar_1 \ldots \pibar_k$, and the bases, a bound on the number of well-behaved $\krlbv$s of $\pi$ with respect to $\frb$ and $\frP$. 
\end{lemma}

\begin{proof}
    Let $\akrlbv = \tuple{
        \albv,
        \tuple{
            \eta_1, \ldots, \eta_k
            }
        }$
    be a well-behaved $\krlbv$ of $\pi$ with respect to $\frb$ and $\frP$. 
    Each $\eta_1, \ldots, \eta_k$ can take only finitely many values, since the set of
    $2$-types, supplemented with the null $2$-type, is fixed.


    We next bound the size of $\albv$.
    Recall that well-behavedness implies that $\albv$ is valid (Definition \ref{def:lbv}),
    and thus there exists exactly one $i$ such that $\albv(\basetext, i) = 1$ and the rest are set to $0$. Thus, there are only finitely many options for the base component of $\albv$.
    It remains to bound the entries of $\albv$.
    By (King-reaching Periods),
    if the behavior $\bhperiod{\pi}{j}$ is not king-reaching,
    we have $\albv(\periodtext, j) = 0$. Thus we only need to compute a bound in the king-reaching case.
    Suppose that the behavior $\bhperiod{\pi}{j}$ is king-reaching and
    $\bhperiod{\pi}{j}(\eta, \pi') > 0$ for some 2-type $\eta$ and king 1-type $\pi'$.
    By the definition of behavior generated by $\akrlbv$,
    $\albv(\pi, j) \le \albv(\pi, j) \cdot \bhperiod{\pi}{j}(\eta, \pi') \le \eval{\akrlbv}(\eta, \pi')$.
    By (Cardinality linking), we have that
    $\eval{\akrlbv}(\eta, \pi')$ is the number of indices
    $1 \le j \le k$
    satisfying
    $\twotypetoking{j}(\akrlbv) = \eta$ and
    $\pibar_j = \pi'$,
    which is bounded by $k$.
 
    
\end{proof}

\begin{definition}[$\cS^\pi_{\frb, \frP}$] \label{def:krlbvs}
    Let $\cS^\pi_{\frb, \frP}$
    denote the set of well-behaved $\krlbv$s of $\pi$ w.r.t. $\frb$ and $\frP$. 
\end{definition}

The idea is that whenever we have a village satisfying a $\gbv$ compatible with a fixed kernel, we will
have a D-structure where each concrete element satisfies a well-behaved $\krlbv$, and hence the number of $\krlbv$s of peasants is bounded. This will allow us to describe the D-structure by specifying
the quantity of each $\krlbv$. From such a D-structure we can create a village for the kernel.

We are now ready to introduce the village system,
which captures the conditions for the existence of a village with a specified list of shadows and their 1-types.

\begin{definition}[Village system, Weak village system, Homogeneous village system]\label{def:village_system}
    The village system with respect to $\frb$ and $\frP$ is an existential Presburger formula $\villagesystem_{\frb, \frP}$ defined as follows.
    \begin{compactitem}
        \item \itemtitle{Variables}
            The variables of $\villagesystem_{\frb, \frP}$ are:
            \begin{compactitem}
                \item
                    $x_i$ for $1 \le i \le \numouterperiods$;
                \item
                    $y_{i, j}$ for $1 \le i \le k$ and
                    $1 \le j \le \numbhperiods{\pibar_i}$;
                \item
                    $z^\pi_{\akrlbv}$ for every peasant 1-type $\pi$ (relative to $\frb$ and $\frP$, which are fixed
                    throughout Section \ref{subsec:realizable-gbv}), 
                    and every $\akrlbv \in \cS^\pi_{\frb, \frP}$.
            \end{compactitem}

            Intuitively,
            \begin{compactitem}
                \item
                    $x_i$ encodes the number of occurrences of the $\gbv$ period $\frp_i$ of the village;

                \item
                    for $1 \le i \le k$, $y_{i, 1}, \ldots, y_{i, \numbhperiods{\pibar_i}}$
                    encodes the intended $\lbv$ of the $i^{th}$ shadow in the village, which is homogeneously valid (in the sense of Definition \ref{def:lbv})
        
                \item
                    $z^\pi_{\akrlbv}$ encodes the number of peasants with 1-type $\pi$ and $\krlbv$ $\akrlbv$ in the village.
            \end{compactitem}

            Thus $ \sum\limits_{i \in \intsinterval{\numouterperiods}} x_i\frp_i$ represents the intended $\gbv$ of the village,
            and for $1 \le i \le k$,
            $\sum\limits_{j \in \intsinterval{\numbhperiods{\pibar_i}}} y_{i, j}\bhperiod{\pibar_i}{j}$
            represents the intended behavior of the $i^{th}$ shadow in the village.
    \end{compactitem}

    $\villagesystem_{\frb, \frP}$ consists of three types of constraints.
    The first type ensures that the behavior vector bases and periods of the vertices in the intended village agree with those specified by the intended $\gbv$.
    \begin{compactitem}
        \item \itemtitle{Counting shadows constraint}
            For every king 1-type $\pi$ and
            $1 \le j \le \numbhperiods{\pi}$,
            $\villagesystem_{\frb, \frP}$ contains the constraints
                $\frb(\pi, \periodtext, j) + 
                \sum\limits_{i \in \intsinterval{\numouterperiods}}
                x_i\frp_i(\pi, \periodtext, j)
                \ =\ 
                \sum\limits_{i \in I_\pi} y_{i, j}$,\\
            where $I_\pi \subseteq \intsinterval{k}$, for a king $1$-type $\pi$
            denotes the set of indices $i$ such that $\pibar_i = \pi$.

            The constraint guarantees that the number of behavior vector periods contributed by shadows is exactly the same as the number required by the intended  $\gbv$.

        \item \itemtitle{Counting peasant bases constraint}
            For each peasant 1-type $\pi$ and
            $1 \le j \le \numbhbases{\pi}$,
            $\villagesystem_{\frb, \frP}$ contains the constraint\\
                $\frb(\pi, \basetext, j) +
                \sum\limits_{i \in \intsinterval{\numouterperiods}}
                x_i\frp_i(\pi, \basetext, j)
                \ =\ 
                \sum\limits_{\akrlbv \in \cS^\pi_{\frb, \frP}}
                \akrlbv(\basetext, j) \cdot z^\pi_{\akrlbv}$.

            The constraint guarantees that the number of the behavior-vector base components contributed by peasants 
            is exactly the same as the number required by the intended  $\gbv$. 

        \item \itemtitle{Counting peasant periods constraint}
            For every peasant 1-type $\pi$ and
            $1 \le j \le \numbhperiods{\pi}$,
            if $\bhperiod{\pi}{j}$ is king-reaching,
            then
            $\villagesystem_{\frb, \frP}$ contains the constraints\\
                $\frb(\pi, \periodtext, j) + 
                \sum\limits_{i \in \intsinterval{\numouterperiods}}
                x_i\frp_i(\pi, \periodtext, j)
                \ =\ 
                \sum\limits_{\akrlbv \in \cS^\pi_{\frb, \frP}}
                \akrlbv(\periodtext, j) \cdot z^\pi_{\akrlbv}$.

            The constraint guarantees that the number of behavior vector periods contributed by peasants is exactly the same as the number required by the intended $\gbv$.
    \end{compactitem}

\medskip

    The second type of constraint 
is the {\bf (Shadow--peasant matching constraint)}:
            For $1 \le i \le k$,
            every 2-type $\eta$, and
            every peasant 1-type $\pi$,
            $\villagesystem_{\frb, \frP}$ contains the constraint\\
                $
                \bfbbar_i(\eta, \pi) + 
                \sum\limits_{j \in \intsinterval{\numbhperiods{\pibar_i}}}
                y_{i, j} \bhperiod{\pibar_i}{j}(\eta, \pi)
                \ =\ 
                \sum\limits_{\akrlbv \in \cS^{\pi}_{\frb, \frP}}
                \kappa_{\akrlbv, i, \eta} \cdot z^\pi_{\akrlbv}$,
            where $\kappa_{\akrlbv, i, \eta}$ is $1$ if $\twotypetoking{i}(\akrlbv) = \dual{\eta}$, otherwise, it is $0$.

            This constraint guarantees that the number of edges incident to the $i^{th}$ shadow equals
            the total number of edges incident to that shadow as specified by the $\krlbv$s.

    The last type of constraint enforces an additional structural property of a village.
    Recall from property (V1) of Definition \ref{def:villageshadowpeasant}  that a village contains no phantom vertices of peasant 1-type.
    If a structure realizes a $\gbv$ in which a peasant 1-type $\pi$ is incorporated periodically but not basely,
    then by Lemma~\ref{lem:propertiesincorporating},
    it contains a phantom vertex with 1-type $\pi$,
    which implies that the structure cannot be a village.
    Since $\frb$ and $\frP$ are concretifying,
    if the peasant 1-type $\pi$ is incorporated periodically in $\frP_i$,
    then $\pi$ is incorporated basely in some $\frP_j$.
    Therefore, the following condition implies
    that every peasant $1$-type of $\frb$ and $\frP$ is incorporated basely by the intended $\gbv$ $\agbv = \sum_{i \in \intsinterval{\numouterperiods}} x_i \frp_i$.

   {\bf (Absorption of phantom vertices with peasant 1-type constraint)}
            For every $1 \le i \le \numouterperiods$,
          $x_i > 0$.
          
\medskip

    We also define the \emph{homogeneous village system},
    denoted by $\villagesystem_{\frb,\frP}^{\homo}$,
    to be the system obtained by removing the constant terms from the first four constraints.
    
    The \emph{weak village system},
    denoted by $\villagesystem_{\frb,\frP}^{-}$,
    is obtained by removing the  {\bf Absorption of phantom vertices with peasant 1-type constraint}.
\end{definition}

To capture the compatibility between kernel and village,
we need additional conditions.

\begin{definition}[Kernel-friendly solution of the village system] \label{def:village_friendly}
    We say that a solution
       $ \set{
            x_i \gets a_i,
            \quad
            y_{i, j} \gets b_{i, j},
            \quad
            z^\pi_{\akrlbv} \gets c^\pi_{\akrlbv}
        }$
    of the village system $\villagesystem_{\frb, \frP}$ is \emph{$\cK$-friendly}
    if it satisfies the following conditions.
    \begin{compactitem}
        \item \itemtitle{Absorption of phantom nobles in $\cK$ constraint}
            For every phantom noble in $\cK$ realizing the 1-type $\pi$,
            there exists $1 \le i \le \numouterperiods$ such that
            $a_i > 0$ and $\pi$ is incorporated basely by $\frp_i$,
            where $\frp_i$ is the $i^{th}$ element in $\frP$.

        \item \itemtitle{Absorption of multi-edges in $\cK$ constraint}
            For  vertices $v_1$ and $v_2$ in $\cK$,
            if there exists edges $e_1$ and $e_2$ from $v_1$ to $v_2$, 
            then there exists $1 \le i \le \numouterperiods$ such that
            $a_i > 0$ and $\frp_i$ incorporates $\twotypeof(e_1)$ or $\twotypeof(e_2)$
            between $\onetypeof(v_1)$ and $\onetypeof(v_2)$.

        \item \itemtitle{Absorption of self-loops in $\cK$ constraint}
            For every self-loop $e$ on vertex $v$ in $\cK$,
            there exists $1 \le i \le \numouterperiods$ such that
            $a_i > 0$ and $\frp_i$ incorporates $\twotypeof(e)$
            between $\onetypeof(v)$
            and $\onetypeof(v)$.
    \end{compactitem}
\end{definition}

\subsection{From a solution to a village}
\begin{theorem}[Solvability-to-village-existence] \label{thm:solvetovillage}
Let $\frb$ be a $\gbv$, and $\frP$ a finite set of $\gbv$s
such that $\frb$ and $\frP$ are strongly concretifying and $\cL(\frb ; \frP) \subseteq \sol{\overapprox{\anfgptwosentence \land \agcc}(\frx)}$.
Let $\cK$ be a kernel of $\frb$ and $\frP$. 
If the homogeneous village system $\villagesystem_{\frb, \frP}^\homo$ has a $\cK$-friendly solution, then
there exists a village compatible with $\cK$.
\end{theorem}

We start with a result implying the existence of a village, without any compatibility condition with regard to a kernel.

\begin{lemma} \label{lem:village_exists}
    If the homogeneous village system $\villagesystem_{\frb, \frP}^\homo$ has a solution,
    then there exists a village of $\frb$ and $\frP$ with $k$ shadows $v^\shadowtext_1, \ldots, v^\shadowtext_k$
    with $\onetypeof(v^\shadowtext_i) = \pibar_i$ for $1 \le i \le k$.
\end{lemma}

\begin{proof}
    Suppose that the homogeneous village system $\villagesystem_{\frb, \frP}^\homo$ has a solution  \\
    $\set{
        x_i \gets a_i,
        \quad
        y_{i,j} \gets b_{i,j},
        \quad
        z^\pi_s \gets c^\pi_s
    }$.
    Let $\agbv := \sum\limits_{i \in \intsinterval{\numouterperiods}} a_i \frp_i$.

    \begin{definition}[$v^\shadowtext_i$,$(\pi, \akrlbv)$ prototype-vertices. $(\pi,i)$ prototype-vertices]  \label{def:vertexflavors}
    We start with an empty D-structure $\cG_0$, and add the following vertices.
    \begin{compactitem}
        \item \itemtitle{$v^\shadowtext_i$}
            Add $k$ fresh phantom vertices $v^\shadowtext_1, \ldots, v^\shadowtext_k$ as the shadows,
            and set the 1-type of $v^\shadowtext_i$ to $\pibar_i$.

        \item \itemtitle{$(\pi, \akrlbv)$ prototype-vertices}
            For every peasant 1-type $\pi$ and every king-restricted $\lbv$
            $\akrlbv \in \cS^\pi_{\frb, \frP}$,
            add $c^\pi_\akrlbv$ fresh concrete vertices with 1-type $\pi$.
            We call these the \emph{$(\pi, \akrlbv)$ prototype-vertices}.

        \item \itemtitle{$(\pi, i)$ prototype-vertices}
            Recall the behavior period  vectors $\bhperiod{\pi}{i}$ from Table \ref{table:notation}. For every peasant 1-type $\pi$ and every
            $1 \le i \le \numbhperiods{\pi}$,
            if the behavior $\bhperiod{\pi}{i}$ is not king-reaching, then add
            $\agbv(\pi, \periodtext, i)$ fresh phantom vertices with 1-type $\pi$.
            We call these the \emph{$(\pi, i)$ prototype-vertices}.
    \end{compactitem}
 \end{definition}

    The intuition is that the vertices $v^\shadowtext_i$ will serve as shadows in the village,
    while the remainder of the construction will ensure that $(\pi, \akrlbv)$ prototype-vertices satisfy $\akrlbv$
    and that $(\pi, i)$ prototype-vertices has behavior $\bhperiod{\pi}{i}$.
    We add edges to $\cG_0$ with the following procedure:
    \begin{compactitem}
        \item \itemtitle{Edges between $v^\shadowtext_i$ and $(\pi, \akrlbv)$ prototype-vertices}
            For every $1 \le i \le k$,
            for every $(\pi, \akrlbv)$ prototype-vertex $u$,
            if $\twotypetoking{i}(\akrlbv)$ is not the null type,
            then add an edge from $u$ to $v^\shadowtext_i$ with 2-type $\eta$.

        \item \itemtitle{Edges within $(\pi, \akrlbv)$ prototype-vertices and $(\pi, i)$ prototype-vertices}
            For every pair of peasant 1-types $\pi_1, \pi_2$ and every asymmetric 2-type $\eta \in \twotypescmpi{\pi_1}{\pi_2}$,
            let $V_1$ be the set of vertices with 1-type $\pi_1$
            and $V_2$ be the set of vertices with 1-type $\pi_2$.
            Let $f_1: V_1 \to \bbN$ and $f_2: V_2 \to \bbN$ be functions defined as
            \begin{equation*}
                \begin{aligned}
                    f_1(v)
                    \ :=\ &
                    \begin{cases}
                        \eval{\akrlbv}(\eta, \pi_2), &\text{if $v$ is a $(\pi_1, \akrlbv)$ prototype-vertex} \\
                        \bhperiod{\pi_1}{i}(\eta, \pi_2), &\text{if $v$ is a $(\pi_1, i)$ prototype-vertex}
                    \end{cases} \\
                    f_2(v)
                    \ :=\ &
                    \begin{cases}
                        \eval{\akrlbv}(\dual{\eta}, \pi_1), &\text{if $v$ is a $(\pi_2, \akrlbv)$ prototype-vertex} \\
                        \bhperiod{\pi_2}{i}(\dual{\eta}, \pi_1), &\text{if $v$ is a $(\pi_2, i)$ prototype-vertex}
                    \end{cases}
                \end{aligned}
            \end{equation*}
            We will show that the conditions of the rewiring lemma, Lemma \ref{lem:rewiring}, are satisfied. 
            Recall that for every $\akrlbv \in \cS^{\pi_1}_{\frb, \frP}$,
            \begin{equation*}
                \begin{aligned}
                    \eval{\akrlbv}
                    \ =\ 
                    &\sum_{i \in \intsinterval{\numbhbases{\pi_1}}} \akrlbv(\basetext, i) \cdot \bhbase{\pi_1}{i} +
                    \sum_{i \in \intsinterval{\numbhperiods{\pi_1}}} \akrlbv(\periodtext, i) \cdot \bhperiod{\pi_1}{i}  \\
                    \ =\ 
                    &\sum_{i \in \intsinterval{\numbhbases{\pi_1}}} \akrlbv(\basetext, i) \cdot \bhbase{\pi_1}{i} +
                    \sum_{i \in \intsinterval{\numbhperiods{\pi_1}}, \text{$\bhperiod{\pi_1}{i}$ is king-reaching}} \akrlbv(\periodtext, i) \cdot \bhperiod{\pi_1}{i}.
                \end{aligned}
            \end{equation*}

            Since $\akrlbv$ is well-behaved (Definition \ref{def:king-restricted-lbvwellbehaved})
            the (King-reaching Periods)  condition guarantees that $\akrlbv(\periodtext, i) = 0$ if $\bhperiod{\pi_1}{i}$ is not king-reaching.
            For $1 \le i \le \numbhbases{\pi_1}$,
            the  {\bf Counting peasant bases constraint} in $\villagesystem^\homo_{\frb, \frP}$ implies:
            \begin{equation*}
                \agbv(\pi_1, \basetext, i)
                \ =\ 
                \sum\limits_{\akrlbv \in \cS^{\pi_1}_{\frb, \frP}}
                \akrlbv(\basetext, i) \cdot c^{\pi_1}_{\akrlbv}. 
            \end{equation*}

            For $1 \le i \le \numbhperiods{\pi_1}$,
            if $\bhperiod{\pi_1}{i}$ is king-reaching,
      the  {\bf Counting peasant periods constraint} in $\villagesystem^\homo_{\frb, \frP}$,
         implies
            \begin{equation*}
                \agbv(\pi_1, \periodtext, i)
                \ =\ 
                \sum\limits_{\akrlbv \in \cS^\pi_{\frb, \frP}}
                \akrlbv(\periodtext, i) \cdot c^\pi_{\akrlbv}.
            \end{equation*}
            Combining the above, the sum of the $(\eta,\pi_2)$-entries of the intended behaviors of all $(\pi_1, \akrlbv)$ prototype-vertices and $(\pi_1,i)$ prototype-vertices is     $\sum_{v \in V_1} f_1(v)$, which is 
          
            \begin{equation*}
                \begin{aligned}
                    &\sum_{\akrlbv \in \cS^{\pi_1}_{\frb, \frP}}
                    c^{\pi_1}_{\akrlbv} \cdot \eval{\akrlbv}(\eta, \pi_2)
                    +
                    \sum_{i \in \intsinterval{\numbhperiods{\pi_1}}, \text{$\bhperiod{\pi_1}{i}$ is not king-reaching}} \agbv(\pi_1, \periodtext, i) \cdot \bhperiod{\pi_1}{i}(\eta, \pi_2) \\
                    \ =\ 
                    &\sum_{i \in \intsinterval{\numbhbases{\pi_1}}} \agbv(\pi_1, \basetext, i) \cdot \bhbase{\pi_1}{i}(\eta, \pi_2) +
                    \sum_{i \in \intsinterval{\numbhperiods{\pi_1}}} \agbv(\pi_1, \periodtext, i) \cdot \bhperiod{\pi_1}{i}(\eta, \pi_2) \\
                    \ =\ 
                    &\sumcon_{\pi_1, \eta, \pi_2}(\agbv).
                \end{aligned}
            \end{equation*}
            Symmetrically, we have $\sum_{v \in V_2} f_2(v) = \sumcon_{\pi_2, \dual{\eta}, \pi_1}(\agbv)$.

            Since each $\frp_i$ satisfies the $\match$ condition,
            $\agbv$ also satisfies the $\match$ condition,
            which implies that $\sumcon_{\pi_1, \eta, \pi_2}(\agbv) = \sumcon_{\pi_2, \dual{\eta}, \pi_1}(\agbv)$.
            Finally we have $\sum_{v \in V_1} f_1(v) = \sum_{v \in V_2} f_2(v)$.

            Thus, by Lemma~\ref{lem:rewiring}, 
            we can add edges of 2-type $\eta$
            between the $V_1$ and $V_2$ so that:

            \medskip (*)
            for every $(\pi_1,\akrlbv)$ prototype-vertex $v$,
            $\bh{v}(\eta,\pi_2) = f_1(v) = \eval{\akrlbv}(\eta,\pi_2)$, and
            for every $(\pi_1, i)$ prototype-vertices $v$,
            $\bh{v}(\eta,\pi_2) = f_1(v) = \bhperiod{\pi_1}{i}(\eta,\pi_2)$.

            Symmetrically, the corresponding equalities hold for the $(\pi_2,\akrlbv)$ and $(\pi_2,i)$ prototype-vertices.

    \end{compactitem}

    We argue that $\cG_0$ satisfies certain properties of a village.
    In the subsequent steps, we will further transform $\cG_0$ into a village. Note that $\cG_0$ is asymmetric as required for a village (Definition \ref{def:villageshadowpeasant}).
    \begin{compactitem}
        \item
            We first argue that $\cG_0$ realizes $\agbv$.
            We define the realization witness $\rep$ on the three kinds of vertices,
            and show that $\bh{v} = \eval{\rep(v)}$ holds for every $v$ in $\cG_0$.
            \begin{compactitem}
                \item \itemtitle{$v^\shadowtext_i$}
                    For every $1 \le i \le k$,
                    we let $\rep(v^\shadowtext_i)(\periodtext, j) = b_{i, j}$ and set all base entries to 0.
                    Note that this implies that $\eval{\rep(v^\shadowtext_i)} = \sum_{j \in \intsinterval{\numbhperiods{\pibar_i}}} b_{i, j} \bhperiod{\pibar_i}{j}$.
                    Clearly, $\rep(v^\shadowtext_i)$ is homogeneously valid.

                    In the construction of $\cG_0$, the only edges incident to $v^\shadowtext_i$ are
                    those introduced according to the $\akrlbv$'s of the
                    $(\pi,\akrlbv)$ prototype-vertices.
                    More precisely, for every 2-type $\eta$ and every 1-type $\pi$,
                    the number of such edges can be counted as follows:
                    \begin{equation*}
                        \bh{v^\shadowtext_i}(\eta, \pi)
                        \ =\ 
                        \sum\limits_{\akrlbv \in \cS^{\pi}_{\frb, \frP}}
                        \kappa_{\akrlbv, i, \eta} \cdot c^\pi_{\akrlbv}
                    \end{equation*}
                    where $\kappa_{\akrlbv,i,\eta}$ is equal to $1$ if
                    $\twotypetoking{i}(\akrlbv) = \dual{\eta}$, and is equal to $0$
                    otherwise.

                    On the other hand, by the {\bf Shadow--peasant matching constraint} in
                    $\villagesystem^\homo_{\frb, \frP}$,
                    we have
                    \begin{equation*}
                        \sum\limits_{j \in \intsinterval{\numbhperiods{\pibar_i}}}
                        b_{i, j} \bhperiod{\pibar_i}{j}(\eta, \pi)
                        \ =\ 
                        \sum\limits_{\akrlbv \in \cS^{\pi}_{\frb, \frP}}
                        \kappa_{\akrlbv, i, \eta} \cdot c^\pi_{\akrlbv}.
                    \end{equation*}
                    By the definition of $\rep(v^\shadowtext_i)$, the left-hand side is
                    precisely $\eval{\rep(v^\shadowtext_i)}(\eta,\pi)$.
                    Since this holds for every 2-type $\eta$ and every 1-type $\pi$, we
                    conclude that $\eval{\rep(v^\shadowtext_i)} = \bh{v^\shadowtext_i}$.
                
                \item \itemtitle{$(\pi,\akrlbv)$ prototype-vertex}
                    For every $(\pi,\akrlbv)$ prototype-vertex $v$, we define  $\rep(v)$ be
                    the $\lbv$ corresponding to $\akrlbv$.
                    
                    For every king 1-type $\pi'$ and every 2-type $\eta$, the (Cardinality linking)
                    condition for a well-behaved $\krlbv$ ensures that
                    $\eval{\rep(v)}(\eta,\pi') = \bh{v}(\eta,\pi')$.
                 For every peasant 1-type $\pi'$ and every 2-type $\eta$,
                    because the edges of 2-type $\eta$ were added according to the procedure described in Lemma~\ref{lem:rewiring},
                    we have $\eval{\rep(v)}(\eta,\pi') = \bh{v}(\eta,\pi')$.
                    Hence, $\eval{\rep(v)} = \bh{v}$.

                \item \itemtitle{$(\pi,i)$ prototype-vertex}
                    For every $(\pi,i)$ prototype-vertex $u$, define $\rep(u)$ by
                    $\rep(u)(\periodtext,i)=1$,
                    with all other entries set to $0$.
                    By construction, there are no edges in $\cG_0$ from $u$ to any
                    shadow. Consequently, for every king 1-type $\pi'$ and every
                    2-type $\eta$,
                    $\bh{u}(\eta,\pi') = 0 = \eval{\rep(u)}(\eta,\pi')$,
                    since $\bhperiod{\pi}{i}$ is not king-reaching.
                    For every peasant 1-type $\pi'$ and every 2-type $\eta$,
                    since the edges of 2-type $\eta$ are added according to the procedure described in Lemma~\ref{lem:rewiring},
                    we have $\eval{\rep(u)}(\eta,\pi') = \bh{u}(\eta,\pi')$.
                    Hence, $\eval{\rep(u)} = \bh{u}$.
            \end{compactitem}
            It remains to show that
            for every $\pi \in \onetypes$, 
            the slice of $\agbv$ on $\pi$ is precisely the sum of the witnessed-assigned $\lbv$s of the elements with 1-type $\pi$.
            This follows straightforwardly from the  {\bf Counting shadows constraint}, {\bf Counting peasant bases constraint}, and {\bf Counting peasant periods constraint}
            in $\villagesystem^\homo_{\frb, \frP}$.

        \item \itemtitle{Partial V1}
            By construction, we see that
            every concrete vertex of $\cG_0$ realizes a 1-type in $\onetypepeasant{\frb}{\frP}$.
            There are still some phantom vertices of $\cG_0$ realizing a 1-type in $\onetypepeasant{\frb}{\frP}$.
            We will remove them in further steps.

        \item \itemtitle{Partial V2}
            By construction,
            there is no multi-edge of $\cG_0$ involving a shadow, and
            no edge of $\cG_0$ connects two shadows.
            There are still some self-loops for vertices with peasant 1-types.
            We will remove them in further steps.
    \end{compactitem}

    We next transform the D-structure $\cG_0$ into the desired village in two steps.
    First, we obtain $\cG_1$ by merging each $(\pi,j)$ prototype-vertex into a
    $(\pi,\akrlbv)$ prototype-vertex of the same 1-type.

    To see that such a $(\pi,\akrlbv)$ prototype-vertex always exists, let $\pi$
    be a peasant 1-type and let $u$ be a $(\pi,j)$ prototype-vertex.
    By definition of peasant 1-type (Definition \ref{def:kingandpeasantonetypeforbcvs}),
    there exists some $1 \le t \le \numouterperiods$ and $1 \le i \le \numbhbases{\pi}$ such that
    $\frp_t(\pi,\basetext,i) > 0$. 
    By the {\bf Absorption of phantom vertices with peasant 1-type constraint}
    in $\villagesystem^\homo_{\frb, \frP}$, we have
    $a_t > 0$.
    By the {\bf Counting peasant bases constraint}
    in $\villagesystem^\homo_{\frb, \frP}$,
      $\sum\limits_{t \in \intsinterval{\numouterperiods}}
        a_t\frp_t(\pi, \basetext, i)
        \ =\ 
        \sum\limits_{\akrlbv \in \cS^\pi_{\frb, \frP}}             \akrlbv(\basetext, i) \cdot c^\pi_{\akrlbv}$.
    The left hand side is positive, and at least one of the vectors being summed on the right is nonzero. Thus there exists $\akrlbv \in \cS^{\pi}_{\frb, \frP}$
    such that $c^\pi_\akrlbv > 0$.
    Therefore our construction contains a $(\pi,\akrlbv)$ prototype-vertex $v$ which $u$ can be merged.

    \begin{claim}[$\cG_1$ properties] \label{clm:g1props}
        $\cG_1$ has the following properties. \begin{inparaenum} \item It  realizes $\agbv$ and is (still) asymmetric. \item It realizes (V1) from the definition of a village (Definition \ref{def:villageshadowpeasant}). \item
        It realizes part  of the property (V2) of a village (``Partial V2''): There are no multi-edges incident to
        a shadow. \end{inparaenum}
    \end{claim}
    Note that we do not claim the other part of (V2):  some vertices of peasant 1-type may still have self-loops.       
    \begin{claimproof}(of Claim \ref{clm:g1props})
        For the first item  note that the merge operation preserves the $\gbv$ realized by the D-structure, by Lemma~\ref{lem:operation}.
        The merge operation also does not introduce new two types between distinct elements, so clearly preserves asymmetry.

        For the second item,  note that every phantom vertex of peasant 1-type is merged with some
        concrete vertex, every concrete vertex of $\cG_1$ realizes a 1-type
        in $\onetypepeasant{\frb}{\frP}$, while every phantom vertex of $\cG_1$
        realizes a 1-type in $\onetypeking{\frb}{\frP}$.
        This implies (V1).

        For the last claim. By construction, there are no edges between shadows and
        $(\pi,i)$ prototype-vertices.
        Therefore, the merge step introduces no multi-edges incident to
        a shadow.
    \end{claimproof}

    Finally, we obtain $\cV$ from $\cG_1$ by eliminating all self-loops.
    Let $\ell$ be the maximum number of self-loops incident to any vertex
    of $\cG_1$.

            
            We take the disjoint union of $2$ copies of
            $\cG_1$, and then merge the copies of each shadow
            $v^\shadowtext_i$ into a single vertex.
        
            We eliminate the self-loops by performing suitable edge swaps
            between different copies. More precisely, for each vertex $v$, we
            resolve its self-loops by swapping the corresponding edges between
            distinct copies of $v$.
            
    We argue that $\cV$ is the desired village. 
    \begin{compactitem}
        \item
            Since $\cG_1$ realizes $\agbv \in \cL_\homo(\frP)$,
            the disjoint union of $2$ 
            copies of $\cG_1$ satisfy some  $2 \cdot \agbv \in \cL_\homo(\frP)$.
            Since the merge and swap operations preserve the $\gbv$ realized by the D-structure by Lemma~\ref{lem:operation},
            this will imply that $\cV$ realizes some $\gbv$ in $\cL_\homo(\frP)$.

        \item \itemtitle{V1}
            By Claim \ref{clm:g1props}, $\cG_1$ satisfies (V1). The construction of $\cV$ from $\cG_1$ preserves (V1).

        \item \itemtitle{V2}
            Claim \ref{clm:g1props} established (V2) except for the requirements on self-loops for $\cG_1$. By construction,
            we removed self-loops by swapping.
            Hence (V2) holds for $\cV$.
    \end{compactitem}

    Note that swapping does not introduce new $2$-types, so does not impact asymmetry. The swapping \emph{may} introduce multi-edges, but this is allowed within a village according to Definition \ref{def:villageshadowpeasant}. This completes the proof of Lemma \ref{lem:village_exists}, 

\end{proof}

We now refine the lemma above to take into account a kernel.

\begin{lemma} \label{lem:village_friendly}
    Let $\cK$ be a kernel of $\frb$ and $\frP$.
    If the village system $\villagesystem^\homo_{\frb, \frP}$ has a $\cK$-friendly solution,
    then there exists a village $\cV$ of $\frb$ and $\frP$,
    with $\cK$ and $\cV$ are compatible.
\end{lemma}

\begin{proof}
    Let $\cV$ be the village constructed in Lemma~\ref{lem:village_exists} from the $\cK$-friendly solution.
    We now show that $\cK$ and $\cV$ are compatible according to Definition \ref{def:kernel_village_comp}.
    \begin{compactitem}
        \item[\emph{V3.}] \itemtitle{Shadow--king correspondence}
            By the construction in Lemma~\ref{lem:village_exists}, 
            $v^\shadowtext_i$ in $\cV$ has 1-type $\pibar_i$,
            which is also the 1-type of $v^\kingtext_i$ in $\cK$.
            Thus, the property holds.

        \item[\emph{V4.}] \itemtitle{Absorption of phantom nobles in $\cK$}
            By the definition of $\cK$-friendly (Definition \ref{def:village_friendly}), in particular by the Absorption of phantom vertices with
            peasant $1$-type constraint,
            for every phantom vertex in $\cK$ with 1-type $\pi$,
            there exists $1 \le i \le \numouterperiods$ such that
            $a_i > 0$ and $\pi$ is incorporated basely by $\frp_i$.

            Since $\cV$ realizes
            $n \cdot \agbv = \sum_{i \in \intsinterval{\numouterperiods}} n \cdot a_i \frp_i$,
            where $n > 0$, we have $n \cdot a_i > 0$.
            Therefore, $\pi$ is also incorporated basely by
            $n \cdot \agbv$.
            By Lemma~\ref{lem:propertiesincorporating},
            $\cV$ contains a vertex $v$ with 1-type $\pi$.
            Note that, by (K1), every phantom vertex in $\cK$ has peasant 1-type,
            which implies that $\pi$ is a peasant 1-type.
            Thus, $v$ is a peasant in $\cV$, and the desired property holds.

        \item[\emph{V5.}] \itemtitle{Absorption of multi-edges in $\cK$}
            By the definition of $\cK$-friendly,
            for vertices $v_1$ and $v_2$ in $\cK$,
            if there exist edges $e_1$ and $e_2$ from $v_1$ to $v_2$,
            then there exists $1 \le i \le \numouterperiods$ such that
            $a_i > 0$ and $\frp_i$ incorporates some
            $\eta \in \set{\twotypeof(e_1), \twotypeof(e_2)}$
            between $\pi_1 = \onetypeof(v_1)$ and
            $\pi_2 = \onetypeof(v_2)$.

            Again, since $\cV$ realizes
            $n \cdot \agbv = \sum_{i \in \intsinterval{\numouterperiods}} n \cdot a_i \frp_i$,
            where $n > 0$, we have $n \cdot a_i > 0$.
            Therefore, $\eta$ is also incorporated between
            $\pi_1$ and $\pi_2$ by $n \cdot \agbv$.
            By Lemma~\ref{lem:propertiesincorporating},
            $\cV$ contains an edge $e$ from a vertex $v$ of 1-type $\pi_1$
            to a vertex $u$ of 1-type $\pi_2$ such that
            $\twotypeof(e) = \eta$.
            Thus, the desired property holds.

        \item[\emph{V6.}] \itemtitle{Absorption of self-loops in $\cK$}
            By the definition of $\cK$-friendly,
            for every self-loop $e$ on a vertex $v$ in $\cK$,
            there exists $1 \le i \le \numouterperiods$ such that
            $a_i > 0$ and $\frp_i$ incorporates
            $\eta = \twotypeof(e)$ between $\pi$ and $\pi$,
            where $\pi = \onetypeof(v)$.

            Again, since $\cV$ realizes
            $n \cdot \agbv = \sum_{i \in \intsinterval{\numouterperiods}} n \cdot a_i \frp_i$,
            where $n > 0$, we have $n \cdot a_i > 0$.
            Therefore, $\eta$ is also incorporated between
            $\pi$ and $\pi$ by $n \cdot \agbv$.
            By Lemma~\ref{lem:propertiesincorporating},
            $\cV$ contains an edge $e$ from a vertex $v$ of 1-type $\pi$
            to a vertex $u$ of 1-type $\pi$ such that
            $\twotypeof(e) = \eta$.
            Thus, the desired property holds.
    \end{compactitem}

\end{proof}

We are now ready to prove Theorem \ref{thm:solvetovillage}.
Since the corresponding village system has a $\cK$-friendly solution,
Lemma~\ref{lem:village_friendly} implies that there exists a village $\cV$ of $\frb$ and $\frP$ that is compatible with $\cK$.

\subsection{Satisfiability-to-solvability theorem: From a structure satisfying the constraints to a solution of the system}
\label{sec:satisfiabilitytosolvability}

Recall from the proof outline that our final algorithm will first guess one of the linear sets within the overapproximation.
It will then guess a small kernel $\cK$ and look for a $\cK$-friendly solution to the
village system. To show one direction of the correctness, we need to argue that if there is a structure realizing a $\gbv$ in the  linear set, then we can ``decompose'' the
structure to obtain a small kernel $\cK$ and a village, along with a $\cK$-friendly solution to the village system. This is the last component of the proof.

\begin{theorem}[Satisfiability-to-solvability theorem] \label{thm:sattosolve}
    Let $\frb$ be a $\gbv$, and let $\frP$ be a finite set of $\gbv$s, with $\frb$ and $\frP$ strongly concretifying. 
    For every asymmetric structure realizing a $\gbv$ in $\cL(\frb;\frP)$,
    there exist \begin{inparaenum} \item  $\frb' \in \cL(\frb;\frP)$ of size bounded by $f(\frb;\frP)$, where $f$ is a computable function;
    \item a subset $\frP' \subseteq \frP$, where  $\frb'$ and $\frP'$ are strongly concretifying;
    \item and a kernel $\cK$ of $\frb'$ and $\frP'$ which realizes the $\gbv$ $\frb'$ \end{inparaenum}
    such that the homogeneous village system $\villagesystem^\homo_{\frb',\frP'}$ has a $\cK$-friendly solution.
\end{theorem}

Note that a kernel realizing a $\gbv$ $\frb'$ must have size bounded by the size of $\frb'$ and the basis for local behaviors,
since the two of these together describe the total number of edges in a D-structure.

\begin{proof}
Recall that a solution to the village system in Definition \ref{def:village_system}
has coefficients
for variables $x_i$, $y_{i,j}$, and $z^\pi_{\akrlbv}$. The $x_i$ are indexed by $1 \le i \le \numouterperiods$, the
                    $y_{i, j}$ by $1 \le i \le k$ and
                    $1 \le j \le \numbhperiods{\pibar_i}$, and the
                    $z^\pi_{\akrlbv}$ by a peasant 1-type $\pi$
                    and $\akrlbv \in \cS^\pi_{\frb, \frP}$.
We first show how to get the bindings of the $x_i$, which will be denoted $a_i$.

We fix a structure $\cG$ realizing a $\gbv$ $\agbv$ in the linear set generated by base vector $\frb$ and period
vectors $\frp_i$ (Definition \ref{def:semilinearpresburger}).

\begin{definition}[$a_i$] \label{def:ai}
We let $a_i$ for $i \in \numouterperiods$ be such that:
    $\agbv
    \ :=\ 
    \frb + \sum\limits_{i \in \intsinterval{\numouterperiods}} a_i \frp_i
    \ \in\ 
    \cL(\frb;\frP)$.
\end{definition}

Let $\rep_\cG$ be a realization witness for $\cG$ realizing $\agbv$.
Following Corollary \ref{cor:kingonetypesindependentofdstructure},
we can set $v^\prekingtext_1,\ldots,v^\prekingtext_k$ to be the concrete vertices of $\cG$ realizing king 1-types such that
for $1 \le t \le k$, $v^\prekingtext_t$ has 1-type $\pibar_i$ and $\rep_\cG(v)$ base behavior $\bfbbar_i$. 
Note that, since $\cG$ is a structure, all its vertices are concrete.

We prove the theorem by transforming $\cG$ so that it has a substructure that is the desired kernel. We will proceed in  several steps. \emph{Each step will preserve asymmetry, so the asymmetry requirement of a kernel will be satisfied.}

\myparagraph{Step 1: decomposing the initial structure}
We first transform $\cG$ into a D-structure $\cH$ by the following splitting procedure.
Let $\cG_0 := \cG$ and $\rep_0 := \rep_\cG$.
We define $\cG_{t+1}$ and $\rep_{t+1}$ recursively.

Suppose that there exists a concrete vertex $v$ of $\cG_t$ with peasant 1-type $\pi$ such that
$\rep_t(v)(\periodtext,i) > 0$ for some $i$ for which $\bhperiod{\pi}{i}$ is not king-reaching.
We introduce a new phantom vertex $u$ of 1-type $\pi$.
For every 2-type $\eta$ and every 1-type $\pi'$,
we select $\bhperiod{\pi}{i}(\eta,\pi')$ many $\eta$-edges from $v$ to vertices of 1-type $\pi'$ and rewire them from $v$ to $u$.

More precisely, for each selected $\eta$-edge $e$ from $v$ to a vertex $w$ (as in Definition \ref{def:implicitexplicitetaedge}), we proceed as follows.
If $e$ is explicit, we remove it and add an explicit $\eta$-edge from $u$ to $w$.
Otherwise, if $e$ is implicit, we remove it and add an explicit $\dual{\eta}$-edge from $w$ to $u$.

We define the corresponding realization witness by setting
$\rep_{t+1}(v)(\periodtext,i) = \rep_t(v)(\periodtext,i)-1$,
while leaving all other entries of $\rep_{t+1}(v)$ unchanged,
and by setting $\rep_{t+1}(u)(\periodtext,i)=1$,
with all other entries of $\rep_{t+1}(u)$ equal to $0$.
For every other vertex $w$, we set $\rep_{t+1}(w)=\rep_t(w)$.

The following properties are invariant under the splitting procedure:
\begin{compactitem}
    \item
        For every vertex $w$ of $\cG_{t+1}$, $\bhg{t+1}{w} = \eval{\rep_{t+1}(w)}$.

        Indeed, let $v$ be the vertex split at $t^{th}$ iteration,
        let $\pi$ and $i$ be the corresponding 1-type and index,
        and let $u$ be the newly added phantom vertex.
        By construction, we transfer edges from $v$ to $u$
        precisely prescribed by $\bhperiod{\pi}{i}$.
        Hence,
        \begin{equation*}
            \bhg{t+1}{v}
            \ =\
            \bhg{t}{v} - \bhperiod{\pi}{i}
            \ =\
            \eval{\rep_t(v)} - \bhperiod{\pi}{i}
            \ =\
            \eval{\rep_{t+1}(v)}.
        \end{equation*}
        Moreover,
            $\bhg{t+1}{u}
            \ =\
            \bhperiod{\pi}{i}
            \ =\
            \eval{\rep_{t+1}(u)}$.
        For every other vertex $w$, the redirection of edges does not change
        its behavior, and $\rep_{t+1}(w) = \rep_t(w)$. Therefore,
            $\bhg{t+1}{w}
            \ =\ 
            \bhg{t}{w}
            \ =\ 
            \eval{\rep_t(w)}
            \ =\ 
            \eval{\rep_{t+1}(w)}$.

    \item
        The D-structure $\cG_{t+1}$ realizes $\agbv$ with witness
        $\rep_{t+1}$.

        Using the same notation as above, we have
        \begin{equation*}
            \rep_{t+1}(v)(\periodtext,i)
            +
            \rep_{t+1}(u)(\periodtext,i)
            \ =\
            \rep_t(v)(\periodtext,i).
        \end{equation*}
        All other witness entries are unchanged. Thus, for every 1-type,
        the sum of the witnesses over all vertices of that 1-type remains unchanged.
        Consequently, $\rep_{t+1}$ is a realization witness for the same $\agbv$.
\end{compactitem}

At each iteration,
the sum of the entries $\rep_t(v)(\periodtext,i)$ over concrete vertices $v$ and indices $i$ for which $\bhperiod{\onetypeof_{\cG_t}(v)}{i}$ is not king-reaching
decreases by one. Hence, the procedure eventually terminates.
Let $\cH$ be the resulting D-structure and let $\rep_\cH$ be the
corresponding realization witness.

By construction, $\cH$ contains the following three kinds of vertices:
\begin{definition}[$v^\prekingtext_i$, $(\pi,\akrlbv)$ split-vertices,$(\pi,i)$ split-vertices] \label{def:splitvertices}
    \leavevmode
\begin{compactitem}
    \item
        \itemtitle{$v^\prekingtext_i$}
        These are the concrete vertices realizing king 1-types.

    \item
        \itemtitle{$(\pi,\akrlbv)$ split-vertices}
        For every concrete vertex $v$ in $\cH$ of peasant 1-type $\pi$,
        we define a $\krlbv$ $\akrlbv$ associated to $v$. Recall from Definition \ref{def:krlbvs} that a $\krlbv$ consists
        of a local behavior vector and a sequence of $2$-types (possibly the null $2$-type). The $\lbv$ will be $\rep_\cH(v)$, and the sequence of 2-types  or null $2$-types will
        be those associated with the edges from $v$ to the vertices $v^\prekingtext_1, \ldots, v^\prekingtext_k$ in $\cH$.
        Since the splitting procedure has terminated,
        for every concrete vertex with peasant 1-type,
        $\rep_\cH(v)(\periodtext,i)=0$,
        whenever $\bhperiod{\pi}{i}$ is not king-reaching.

        We call $v$ a \emph{$(\pi,\akrlbv)$ split-vertex}.

    \item
        \itemtitle{$(\pi,i)$ split-vertices}
        These are the phantom vertices introduced by the splitting procedure.
        For every such vertex $u$ with 1-type $\pi$,
        there exists $1 \le j \le \numbhperiods{\pi}$,
        such that $\rep_\cH(u)(\periodtext,i)=1$,
        with all other entries of $\rep_\cH(u)$ equal to $0$.
        We call $u$ a \emph{$(\pi,i)$ split-vertex}.
\end{compactitem}
\end{definition}

$\cH$ is illustrated in Figure~\ref{fig:step1structuretosolution}.
\begin{figure}[h!]
    \centering
    \input{fig/intuition.tex}
    \caption{
        Illustration of $\cH$ obtained in Step~1 in the proof of Theorem~\ref{thm:sattosolve}. 
        Concrete vertices are drawn with solid outlines, while phantom vertices are drawn with dashed outlines.
        Vertices of king 1-types are shown in red, and vertices of peasant 1-types are shown in blue.
        For simplicity, we assume that there is only one 2-type.
        After the splitting procedure, concrete vertices of peasant 1-types are grouped according to their $\krlbv$s.
        Consequently, all vertices in the same group are adjacent to the same set of vertices of king 1-types and to the same number of vertices of each peasant 1-type.
        For example, the $(\pi_1,\krlbv_{1,1})$ split-vertices are adjacent to $v^{\prekingtext}_1$, $v^{\prekingtext}_2$, and one vertex of 1-type $\pi_1$.
        Similarly, phantom vertices of peasant 1-types are grouped according to the period behaviors they realize.
        For example, each $(\pi_1,1)$ split-vertex is adjacent to exactly one vertex of 1-type $\pi_1$.
        Note that the splitting procedure guarantees that the resulting D-structure contains no multi-edges and no edges between vertices of king 1-types and phantom vertices.
    }
    \label{fig:step1structuretosolution}
\end{figure}
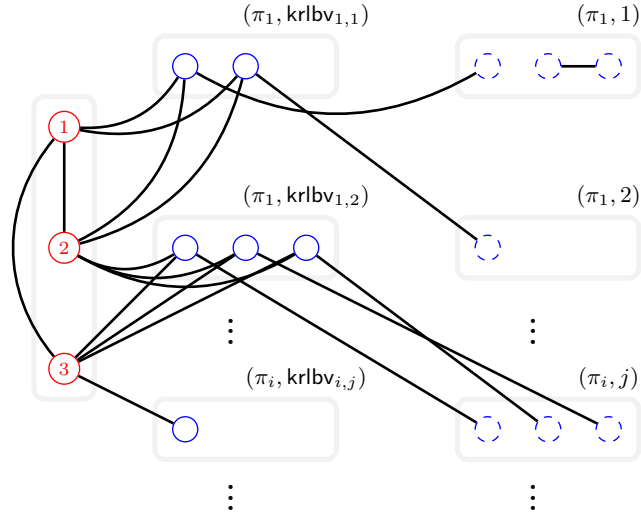

\begin{claim} \label{clm:krlbvwellbehaved}
Each $\akrlbv$ associated to a $(\pi,\akrlbv)$ split-vertex $v$ above is well-behaved, according to Definition \ref{def:king-restricted-lbvwellbehaved}.
\end{claim}
\begin{claimproof}
Recall that well-behavedness requires two properties, according to Definition \ref{def:king-restricted-lbvwellbehaved}.
Just after the definition of $(\pi, \akrlbv)$-split vertex we have established the contraposition of the property King-reaching periods.

Let us consider the other property, Cardinality linking: for every 2-type $\eta$ and every king 1-type $\pi'$,
$\eval{\akrlbv}(\eta, \pi')$
equals the number of indices $1 \le j \le k$ satisfying
$\twotypetoking{j}(\akrlbv) = \eta$ and
$\pibar_j = \pi'$.

Fix king $1$-type $\pi'$ and a $2$-type $\eta$.
Our construction guarantees that for every  $(\pi, \krlbv)$ split vertex $v$,
the number of $\eta$-edges from $v$ to vertices with king 1-type $\pi'$ is
$\eval{\rep_\cH(v)}(\eta, \pi')
= \eval{\akrlbv}(\eta, \pi')
$.
But the number of $\eta$-edges to vertices with $1$-type $\pi'$ is exactly the number of indices
satisfying $\twotypetoking{j}(\akrlbv) = \eta$ and
$\pibar_j = \pi'$.
\end{claimproof}

We next define $b_{i,j}$ and $c^\pi_\akrlbv$,
which form the initial weak assignment.

\begin{definition}[$b_{i,j}$, $c^\pi_\akrlbv$, initial weak assignment] \label{def:weakinitialbij}
For every $1 \le i \le k$ and for $1 \le j \le \numbhperiods{\pibar_i}$, let
$b_{i, j} := \rep_{\cH}(v^\prekingtext_i)(\periodtext, j)$.
For every peasant 1-type $\pi$,
let $c^\pi_\akrlbv$ denote the number of $(\pi, \akrlbv)$ vertices of $\cH$.
By the \emph{weak initial assignment} we mean the assignment to the variables of the village system:
$ \set{
    x_i \gets a_i,
    \quad
    y_{i, j} \gets b_{i, j},
    \quad
    z^\pi_{\akrlbv} \gets c^\pi_{\akrlbv}
}$, where $a_i$ are as in Definition \ref{def:ai}.
\end{definition}

We argue that the weak initial assignment 
forms a solution of the weak village system $\villagesystem_{\frb, \frP}^-$.

\begin{claim} \label{clm:weaksolution}
    The weak initial assignment $ \set{
        x_i \gets a_i,
        \quad
        y_{i, j} \gets b_{i, j},
        \quad
        z^\pi_{\akrlbv} \gets c^\pi_{\akrlbv}
    }$
    is a solution to the weak village system $\villagesystem_{\frb, \frP}^-$.
\end{claim}

\begin{claimproof}
        We start with the {\bf Counting shadows constraint}.
        For every king 1-type $\pi$ and $1 \le j \le \numbhperiods{\pi}$,
        \begin{equation*}
            \agbv(\pi, \periodtext, j)
            \ =\ 
            \frb(\pi, \periodtext, j) +
            \sum\limits_{i \in \intsinterval{\numouterperiods}} a_i\frp_i(\pi, \periodtext, j).
        \end{equation*}
        
        On the other hand, by the definition of realization (Definition~\ref{def:realizingdstructure}),
        the $(\pi, \periodtext, j)$-entry of $\gbv$ is the some of witness of $(\periodtext, j)$-entries of vertices with 1-type $\pi$,
        that is
        \begin{equation*}
            \agbv(\pi, \periodtext, j)
            \ =\ 
            \sum\limits_{i \in I_\pi} \rep_{\cH}(v^\prekingtext_i)(\periodtext, j)
            \ =\ 
            \sum\limits_{i \in I_\pi} b_{i, j}.
        \end{equation*}
        Thus we have
        \begin{equation*}
            \frb(\pi, \periodtext, j) +
            \sum\limits_{i \in \intsinterval{\numouterperiods}} a_i\frp_i(\pi, \periodtext, j)
            \ =\ 
            \sum\limits_{i \in I_\pi} b_{i, j}.
        \end{equation*}
        This completes the proof.

        We now prove the {\bf Counting peasant bases constraint}.
        For every peasant $1$-type $\pi$ and $1 \le j \le \numbhbases{\pi}$,
        we count the number of occurrences of the behavior base $\bhbase{\pi}{j}$,
        \begin{equation*}
            \agbv(\pi, \basetext, j)
            \ =\ 
            \frb(\pi, \basetext, j) +
            \sum\limits_{i \in \intsinterval{\numouterperiods}} a_i \frp_i(\pi, \basetext, j).
        \end{equation*}
        On the other hand, by the definition of witness,
        $\agbv(\pi, \basetext, j)$ is the sum of $(\basetext, j)$-entries of vertices with 1-type $\pi$,
        that is
        \begin{equation*}
            \agbv(\pi, \basetext, j)
            \ =\ 
            \sum\limits_{\akrlbv \in \cS^\pi_{\frb, \frP}}
            \akrlbv(\basetext, j) \cdot c^\pi_{\akrlbv}.
        \end{equation*}

        Hence we conclude as required \\
            $\sum\limits_{i \in \intsinterval{\numouterperiods}} a_i \frp_i(\pi, \basetext, j) +
            \frb(\pi, \basetext, j)
            \ =\ 
            \sum\limits_{\akrlbv \in \cS^\pi_{\frb, \frP}}
            \akrlbv(\basetext, j) \cdot c^\pi_{\akrlbv}$.

        We now prove  the {\bf Counting peasant periods constraint}.
        For every peasant $1$-type $\pi$ and
        $1 \le j \le \numbhperiods{\pi}$,
        if $\bhperiod{\pi}{j}$ is king-reaching,
        we count the number of occurrences of the behavior period $\bhperiod{\pi}{j}$,
        \begin{equation*}
            \agbv(\pi, \periodtext, j)
            \ =\ 
            \frb(\pi, \periodtext, j) +
            \sum\limits_{i \in \intsinterval{\numouterperiods}} a_i \frp_i(\pi, \periodtext, j)
        \end{equation*}
        Reasoning as in the previous constraint,
        \begin{equation*}
            \agbv(\pi, \periodtext, j)
            \ =\ 
            \sum\limits_{\akrlbv \in \cS^\pi_{\frb, \frP}}
            \akrlbv(\periodtext, j) \cdot c^\pi_{\akrlbv}.
        \end{equation*}
        Hence:\\
            $\sum\limits_{i \in \intsinterval{\numouterperiods}} a_i \frp_i(\pi, \periodtext, j) +
            \frb(\pi, \periodtext, j)
             =
            \sum\limits_{\akrlbv \in \cS^\pi_{\frb, \frP}}
            \akrlbv(\periodtext, j) \cdot c^\pi_{\akrlbv}$.

        Finally, we prove the {\bf Shadow--peasant matching constraint}.
        For every $1 \le i \le k$,
        every 2-type $\eta$, and
        every peasant 1-type $\pi$,
        we compare two ways of counting the number of edges with 2-type $\eta$ from $v^\prekingtext_i$ to vertices with peasant 1-type $\pi$.

        On the one hand, this number is determined by the behavior of $v^\prekingtext_i$.
        \begin{equation*}
            \begin{aligned}            
                \bh{v^\prekingtext_i}(\eta, \pi)
                \ =\ &
                \sum\limits_{j \in \intsinterval{\numbhbases{\pibar_i}}}
                \rep_{\cH}(v^\prekingtext_i)(\basetext, j) \cdot \bhbase{\pibar_i}{j}(\eta, \pi) 
                \\
                & \qquad\qquad+
                \sum\limits_{j \in \intsinterval{\numbhperiods{\pibar_i}}}
                \rep_{\cH}(v^\prekingtext_i)(\periodtext, j) \cdot \bhperiod{\pibar_i}{j}(\eta, \pi) \\
                \ =\ &
                \bfbbar_i(\eta, \pi) +
                \sum\limits_{j \in \intsinterval{\numbhperiods{\pibar_i}}}
                b_{i, j} \cdot \bhperiod{\pibar_i}{j}(\eta, \pi).
            \end{aligned}            
        \end{equation*}

        On the other hand,
        recall that the construction guarantees that there are no edges between $(\pi, j)$ split-vertices and $v^\prekingtext_i$.
        We can count the same edges by summing the contribution of $(\pi, \akrlbv)$ split-vertices.
        Recall that $\kappa_{\akrlbv, i, \eta}$ is $1$ if $\twotypetoking{i}(\akrlbv) = \dual{\eta}$ in $\akrlbv$, otherwise, it is $0$.
        Then the total number of such edges is
        $\sum\limits_{\akrlbv \in \cS^{\pi}_{\frb, \frP}} \kappa_{\akrlbv, i, \eta} \cdot c^\pi_{\krlbv}$.
        Thus
        \begin{equation*}
            \bfbbar_i(\eta, \pi) +
            \sum\limits_{j \in \intsinterval{\numbhperiods{\pibar_i}}}
            b_{i, j} \cdot \bhperiod{\pibar_i}{j}(\eta, \pi)
            \ = \ 
            \sum\limits_{\akrlbv \in \cS^{\pi}_{\frb, \frP}} \kappa_{\akrlbv, i, \eta} \cdot c^\pi_{\krlbv}.
        \end{equation*}
\end{claimproof}

We conclude this step by recording the following properties of $\cH$.
\begin{claim}[$\cH$ properties] \label{claim:h}
    \leavevmode
    \begin{compactitem}
        \item
            $\cH$ contains no multi-edges.
            There are no edges between the vertices $v^\prekingtext_1,\ldots,v^\prekingtext_k$
            and any $(\pi,i)$ split-vertex in $\cH$.

        \item
            $\cH$ realizes $\agbv$ with witness $\rep_\cH$.

        \item
            For $1 \le i \le k$, $\bhg{\cH}{v^\prekingtext_i} = \bfbbar + \sum_{j \in \intsinterval{\numbhperiods{\pibar_i}}} b_{i, j} \cdot \bhperiod{\pibar_i}{j}$.
    \end{compactitem}
\end{claim}

\begin{claimproof}
    \leavevmode
    \begin{compactitem}
        \item
            Since $\cG_0$ is a structure,
            it contains no multi-edges.
            The splitting procedure does not introduce any multi-edges,
            and hence $\cH$ contains no multi-edges.

            For every $(\pi,i)$ split-vertex $u$, we have $\bh{u}=\bhperiod{\pi}{i}$.
            Since $\bhperiod{\pi}{i}$ is not king-reaching,
            $u$ is not adjacent to any vertex of king 1-type.

        \item
            By the invariant of the construction, $\cH$ realizes $\agbv$ with witness $\rep_\cH$.
        
        \item
            For $1 \le i \le k$,
            by definition of realization witness,
            we have
            \begin{equation*}
                \begin{aligned}
                    &\bhg{\cH}{v^\prekingtext_i} \\
                    \ =\ &
                    \sum_{j \in \intsinterval{\numbhbases{\pibar_i}}} \rep_\cH(v^\prekingtext_i)(\basetext, j) \cdot \bhbase{\pibar_i}{j}  +
                    \sum_{j \in \intsinterval{\numbhperiods{\pibar_i}}} \rep_\cH(v^\prekingtext_i)(\periodtext, j) \cdot \bhperiod{\pibar_i}{j}.
                \end{aligned}
            \end{equation*}
            Recall that we choose the enumeration of $v^\prekingtext_1, \ldots, v^\prekingtext_k$, see comments after Corollary \ref{cor:kingonetypesindependentofdstructure}.
            Therefore we have
            \begin{equation*}
                \sum_{j \in \intsinterval{\numbhbases{\pibar_i}}} \rep(v^\prekingtext_i)(\basetext, j) \cdot \bhbase{\pibar_i}{j} = \bfbbar_i.
            \end{equation*}
            By Definition~\ref{def:weakinitialbij},
            we have $\rep_\cH(v^\prekingtext_i)(\periodtext, j) = b_{i, j}$.
            Therefore
            \begin{equation*}
                \bhg{\cH}{v^\prekingtext_i}
                \ =\ 
                \bfbbar_i + \sum_{j \in \intsinterval{\numbhperiods{\pibar_i}}} b_{i, j} \cdot \bhperiod{\pibar_i}{j}.
            \end{equation*}
    \end{compactitem}
\end{claimproof}

\myparagraph{Step 2: shrinking the assignment until it is a solution to the village system}
In Claim \ref{clm:weaksolution}, we showed that the initial weak assignment,
consisting of $a_i$, $b_{i,j}$, and $c^\pi_{\akrlbv}$,
is a solution to the weak village system, which is a non-homogeneous integer linear system.
By well-known bounds on minimal solutions in linear systems 
(see Lemma~\ref{lem:minimal} in Appendix \ref{app:solutionsizelemma}), 
there exists another solution
whose values are bounded by an exponential function of the coefficients of the system,
and hence by an exponential function of the entries in $\frb$ and $\frP$
and the entries in the representation of the semilinear $\bhset{\pi}$ for every 1-type $\pi$.

That is, we can define:
\begin{definition}[$a'_i$ $b'_i$, $c'^\pi_{\akrlbv}$] \label{def:a'_ib'_ietc}
    We let $a'_i \le a_i$, $b'_{i,j} \le b_{i,j}$, and $c'^\pi_{\akrlbv} \le c^\pi_{\akrlbv}$ be the bounded solution given
    by Lemma \ref{lem:minimal}.
\end{definition}

\begin{definition}[$a''_i$, $b''_{i,j}$, $c''^\pi_{\akrlbv}$] \label{def:a''ib''ij}
    We define $a''_i := a_i - a'_i \ge 0$,
    $b''_{i,j} := b_{i,j} - b'_{i,j} \ge 0$, and
    $c''^\pi_{\akrlbv} := c^\pi_{\akrlbv} - c'^\pi_{\akrlbv} \ge 0$
\end{definition}

Consider the assignment to the village system variables, corresponding to $a''_i, b''_{i,j}, c''^\pi_{\akrlbv}$ above.
Note that \emph{since this is difference between two solutions to the weak village system,
it is a solution to the homogeneous weak village system}.
We now show:

\begin{claim} \label{clm:lastconstraint} This assignment also satisfies  the last constraint from Definition~\ref{def:village_system},  the
{\bf Absorption of phantom vertices with peasant 1-type constraint}: for every $1 \le i \le \numouterperiods$,
$x_i > 0$. 
\end{claim}

\begin{proof}
We distinguish two cases.
\begin{itemize}
    \item
        Suppose that $a''_t = 0$ for some $1 \le t \le \numouterperiods$.
        We have $\agbv \in \cL(\frb + a'_t\frp_t; \frP \setminus \set{\frp_t})$.
        It is possible that $\frb + a'_t\frp_t$ and $\frP \setminus \set{\frp_t}$ is not strongly concretifying. If this
        is the case, by Lemma~\ref{lem:strongly_concretifying},
        there exists $\frb' \in \cL(\frb + a'_t\frp_t; \frP \setminus \set{\frp_t})$
        and $\frP' \subsetneq \frP  \setminus \set{\frp_t}$
        such that
        $\frb'$ and $\frP'$ are strongly concretifying,
        $\abs{\frb'} \le \left(2M\abs{\frb + a'_t\frp_t}\right)^{2^{\numouterperiods}}$, and
        $\agbv \in \cL(\frb';\frP')$,
        where $M$ depends on $\frP$, as obtained in Lemma~\ref{lem:strongly_concretifying}.
        We then return to Step~1 with $(\frb';\frP')$. 
        Since the cardinality of $\frP'$ is strictly less than the cardinality of $\frP$,
        this process terminates after at most $\abs{\frP}$ iterations, and after that we must be in the other case.
        We abuse notation by continuing to call the resulting $\gbv$  $\frb$, and to call the resulting set of $\gbv$s $\frP$.

        Note that $a'_t$ is bounded by an exponential function 
        (given by Lemma~\ref{lem:minimal} in Appendix~\ref{app:solutionsizelemma}).
        Thus, there is a triple-exponential blow-up  in the size of $\frb'$.
        Since we perform this iteration $\abs{\frP}$ times,
        this step can yield non-elementary blow-up.
        
    \item
        Otherwise, $a''_i > 0$ for every $1\le i\le\numouterperiods$.
        Since every $a''_i > 0$ for every $1\le i\le\numouterperiods$,
        it follows that $\set{x_i \gets a''_i}$ is a solution of the  {\bf Absorption of phantom vertices with peasant 1-type constraint}.
        Therefore, 
        $\set{x_i \gets a''_i, y_{i, j} \gets b''_{i, j}, z^\pi_{\akrlbv} \gets c''^\pi_{\akrlbv}}$
        is a solution to the homogeneous village system $\villagesystem_{\frb,\frP}^\homo$.
\end{itemize}
\end{proof}

So far, we have shown that
$\set{x_i \gets a'_i, y_{i, j} \gets b'_{i, j}, z^\pi_{\akrlbv} \gets c'^\pi_{\akrlbv}}$
is a solution to the weak village system $\villagesystem_{\frb,\frP}^-$,
and
$\set{x_i \gets a''_i, y_{i, j} \gets b''_{i, j}, z^\pi_{\akrlbv} \gets c''^\pi_{\akrlbv}}$
is a solution to the homogeneous village system $\villagesystem_{\frb,\frP}^\homo$.

We define the \emph{target linear representation} and show that its size if bounded.
\begin{definition}[$\frb'$, $\frP'$, and target linear representation] \label{def:targetrep}
    Let $\frb' = \frb + \sum_{i \in \intsinterval{\numouterperiods}} a'_i \frp_i$ and $\frP' = \frP$.
    We refer the pair $\frb'$ and $\frP'$ as the \emph{target linear representation}.
\end{definition}

\begin{claim}
    The size of $\frb'$ is bounded by a computable function.
\end{claim}

\begin{claimproof}
    We bound the size of $\frb'$ as follows.
    The coefficients in these constraints come from the entries of $\frb$ and $\frP$, as well as from the entries in the representation of the semilinear set $\bhset{\pi}$ for each 1-type $\pi$.
    The value $a_i'$ is bounded exponentially in the maximum absolute value of these coefficients,
    by the bound on linear systems given in Lemma~\ref{lem:minimal} in Appendix~\ref{app:solutionsizelemma}.
    Recall that we iteratively construct a system of linear equations admitting a solution in which $a_i''>0$ for every $i$.
    Each iteration causes a triple-exponential blow-up in the coefficients of the system.
    Thus, although the construction is non-elementary, it is effective, and hence the size of $\frb'$ is bounded by a computable function.
\end{claimproof}

We conclude this step by showing that the following \emph{target assignment} is a solution of the homogeneous village system of the target linear representation $\frb'$ and $\frP'$.

\begin{definition}[Target assignment] \label{def:targetassign}
    In the remainder of the proof,
    we refer to $\set{x_i \gets a''_i, y_{i, j} \gets b''_{i, j}, z^\pi_{\akrlbv} \gets c''^\pi_{\akrlbv}}$ as the \emph{target assignment}.
\end{definition}

\begin{claim}\label{claim:target_assign_is_sol}
    The target assignment is a solution to the homogeneous village system $\villagesystem_{\frb',\frP'}^\homo$.
\end{claim}

\begin{claimproof}
    Since $(\frb, \frP)$ and $(\frb', \frP')$ induce the same set of king 1-types and peasant 1-types.
    The homogeneous village systems $\villagesystem_{\frb,\frP}^\homo$ and $\villagesystem_{\frb',\frP'}^\homo$ coincide.
    Therefore
    $\set{x_i \gets a''_i, y_{i, j} \gets b''_{i, j}, z^\pi_{\akrlbv} \gets c''^\pi_{\akrlbv}}$
    is a solution to the homogeneous village system $\villagesystem_{\frb',\frP'}^\homo$.
\end{claimproof}

\begin{claim}
    $\set{x_i \gets 0, y_{i, j} \gets b'_{i, j}, z^\pi_{\akrlbv} \gets c'^\pi_{\akrlbv}}$
    is a solution to the weak village system $\villagesystem_{\frb',\frP'}^-$.
\end{claim}

\begin{claimproof}
    Note that in the weak village system $\villagesystem_{\frb,\frP}^-$,
    the variables $x_i$ occur only through expressions of the form
    $\frb(\pi, x, j) + \sum_{i \in \intsinterval{\numouterperiods}} a'_i \frp_i(\pi, x, j)$.
    This implies that\\
    $\set{x_i \gets 0, y_{i, j} \gets b'_{i, j}, z^\pi_{\akrlbv} \gets c'^\pi_{\akrlbv}}$
    is a solution to the weak village system $\villagesystem_{\frb',\frP'}^-$.
\end{claimproof}

\myparagraph{Step 3: finding a kernel for the target assignment}
The remainder of the proof of Theorem \ref{thm:sattosolve} will involve showing that
there exists a bounded-size kernel for the target linear representation given by $\frb'$ and $\frP'$
as well as showing that the target assignment
$\set{x_i \gets a''_i, y_{i, j} \gets b''_{i, j}, z^\pi_{\akrlbv} \gets c''^\pi_{\akrlbv}}$
is kernel-friendly. 

We obtain the desired kernel by transforming the D-structure $\cH$ constructed in Step~1,
guided by the small solution obtained in Step~2.
The transformation proceeds in three phases. 

\myparagraph{Step 3A: modifying and decomposing to obtain a bounded kernel}
First, we obtain $\cH'$ by partitioning the vertices of peasant 1-type and splitting the vertices of king 1-type, as follows.
\begin{definition}[$\cH'$, Kernel and village vertices, Kernel part, Village part] \label{def:villageandkernelpartvertices}
    \leavevmode
    \begin{compactitem}
        \item
            We first partition the vertices of peasant 1-type.
            For every peasant 1-type $\pi$ and every well-behaved $\krlbv$ $\akrlbv$,
            we choose $c'^\pi_{\akrlbv}$ many $(\pi,\akrlbv)$ split-vertices and designate them as \emph{$(\pi,\akrlbv)$ kernel-vertices}.
            The remaining $(\pi,\akrlbv)$ split-vertices are designated as \emph{$(\pi,\akrlbv)$ village-vertices}.
            Similarly, for every $1 \leq j \leq \numbhperiods{\pi}$ such that $\bhperiod{\pi}{j}$ is non-king-reaching,
            we choose $\frb'(\pi,\periodtext,j)$ many $(\pi,j)$ split-vertices and designate them as \emph{$(\pi,j)$ kernel-vertices}.
            The remaining $(\pi,j)$ split-vertices are designated as \emph{$(\pi,j)$ village-vertices}.

            We leave the realization witnesses of all vertices of peasant 1-type unchanged; that is,
            $\rep_{\cH'}(v)=\rep_{\cH}(v)$ for every such vertex $v$.

        \item
            We next split each vertex of king 1-type into a king and a shadow.
            For every $1 \le i \le k$, we add a phantom vertex $v^\shadowtext_i$ to $\cH'$ and assign it 1-type $\pibar_i$.
            We redirect every edge between $v^\prekingtext_i$ and a $(\pi,\akrlbv)$ village-vertex so that it is instead incident to $v^\shadowtext_i$,
            while keeping its other endpoint unchanged.
            We then rename $v^\prekingtext_i$ as $v^\kingtext_i$.

            We define the realization witnesses of $v^\kingtext_i$ and $v^\shadowtext_i$ as follows.
            For every $\pi$ and $j$, we let
            $\rep_{\cH'}(v^\kingtext_i)(\pi,\basetext,j) = \rep_{\cH}(v^\prekingtext_i)(\pi,\basetext,j)$ and
            $\rep_{\cH'}(v^\kingtext_i)(\pi,\periodtext,j) = b'_{i,j}$.
            For the shadow $v^\shadowtext_i$, we let
            $\rep_{\cH'}(v^\shadowtext_i)(\pi,\basetext,j) = 0$ and \\
            $\rep_{\cH'}(v^\shadowtext_i)(\pi,\periodtext,j) = b''_{i,j}$.
    \end{compactitem}

    We call the $v^\kingtext_i$, $(\pi, \akrlbv)$ kernel-vertices, and $(\pi, j)$ kernel-vertices the \emph{kernel part of $\cH'$},
    and $v^\shadowtext_i$, $(\pi, \akrlbv)$ village-vertices, and $(\pi, j)$ village-vertices the \emph{village part of $\cH'$}.
\end{definition}

We need to generalize the notion of a D-structure realizing a $\gbv$ (Definition \ref{def:realizingdstructure}) to restrict the behaviors
to a subset.
\begin{definition}[Subset of a D-structure realizing a $\gbv$]\label{def:subset_dstruct_realize_gbv}
    Let $\cG$ be a D-structure and $V$ be a subset of the vertices in $\cG$.
    We say that $V$ of $\cG$ \emph{realizes} a $\gbv$ $\agbv$ if 
    there exists a mapping $\rep$ mapping concrete vertices of $\cG$ to valid $\lbv$s and phantom vertices to
    homogeneously valid $\lbv$s such that
    \begin{inparaenum}
        \item for every element $u$, $u$ realizes $\rep(u)$ (Def. \ref{def:lbv}) and
        \item for every $\pi \in \onetypes$, 
        the slice of $\agbv$ on $\pi$ is precisely the sum of $\rep(u)$ over all elements $u$ of 1-type $\pi$ in $V$.
    \end{inparaenum}
    The mapping $\rep$ is again called a \emph{witness of the realization}.
\end{definition}

We argue for some properties of $\cH'$.
\begin{claim}[$\cH'$ properties] \label{claim:hp}
The following  properties hold for $\cH'$.
    \leavevmode
    \begin{compactenum}
        \item
            $\cH'$ also contains no multi-edges or self-loops.
            There are no edges between any two vertices
            $v^\shadowtext_i$ and $v^\shadowtext_j$.
            There are no edges between $v^\kingtext_i$ and any
            $v^\shadowtext_j$, $(\pi,\akrlbv)$ village-vertex, or $(\pi,j)$ village-vertex;
            There are no edges between $v^\shadowtext_i$ and any
            $v^\kingtext_j$, $(\pi,\akrlbv)$ kernel-vertex, or $(\pi,j)$ kernel-vertex.

            However, there may still be edges between the kernel part and the village part.

        \item
            $\rep_{\cH'}$ satisfies the condition of a realization witness; that is, for every vertex $v$ of $\cH'$,
            $\bhg{\cH'}{v} = \eval{\rep_{\cH'}(v)}$.

        \item
            For $1 \le i \le k$,
            $\bhg{\cH'}{v^\shadowtext_i} = \sum\limits_{j \in \intsinterval{\numbhperiods{\pibar_i}}} b''_{i, j} \bhperiod{\pibar_i}{j}$.
where $b''_{i,j}$ are defined in Definition \ref{def:a''ib''ij}.
        \item
            For $1 \le i \le k$,
            $\bhg{\cH'}{v^\kingtext_i} = \bfbbar_i + \sum\limits_{j \in \intsinterval{\numbhperiods{\pibar_i}}} b'_{i, j} \bhperiod{\pibar_i}{j}$.
where $b'_{i,j}$ are as in Definition \ref{def:a'_ib'_ietc}.
        \item
            The kernel part of $\cH'$ realizes the $\gbv$ $\frb'$, in the sense of Definition \ref{def:subset_dstruct_realize_gbv}, with $\rep_{\cH'}$ as a witness.
            %
        \item
            The village part of $\cH'$ realizes the $\gbv$ $\sum_{i \in \intsinterval{\numouterperiods}} a'_i \frp_i$, with $\rep_{\cH'}$ as a witness.
    \end{compactenum}
\end{claim}

\begin{claimproof}(of Claim \ref{claim:hp})
    \begin{compactenum}
        \item
            Since $\cH$ contains neither multi-edges nor self-loops, and the splitting procedure introduces neither,
            $\cH'$ also contains no multi-edges or self-loops.
            The remaining properties are also straightforward from our construction.

        \item
            Let $v$ be a vertex of peasant 1-type. Our construction does not change the behavior of $v$, and hence
            $\bhg{\cH'}{v} = \bhg{\cH}{v}$.
            Moreover, we keep its realization witness unchanged, so $\rep_{\cH'}(v) = \rep_{\cH}(v)$.
            Since $\rep_{\cH}$ is a realization witness of $\cH$, we have $\bhg{\cH}{v} = \eval{\rep_{\cH}(v)}$.
            Combining these equalities yields $\bhg{\cH'}{v} = \eval{\rep_{\cH'}(v)}$.

        \item
            For $1 \le i \le k$,
            note that the $v^\shadowtext_i$ are connected with $c''^\pi_{\akrlbv}$ many $(\pi, \akrlbv)$ village-vertices,
            we have
            \begin{equation*}
                \bhg{\cH'}{v^\shadowtext_i}(\eta, \pi)
                \ =\ 
                \sum\limits_{\akrlbv \in \cS^{\pi}_{\frb', \frP'}}
                \kappa_{\akrlbv, i, \eta} \cdot c''^\pi_{\akrlbv},
            \end{equation*}
            where $\kappa_{\akrlbv, i, \eta}$ is $1$ if $\twotypetoking{i}(\akrlbv) = \dual{\eta}$, otherwise, it is $0$.
            By the constraint {\bf shadow--peasant matching} in $\villagesystem_{\frb', \frP'}^\homo$, which is satisfied by the target assignment, we have:
            \begin{equation*}
                \sum\limits_{j \in \intsinterval{\numbhperiods{\pibar_i}}} b''_{i, j} \bhperiod{\pibar_i}{j}(\eta, \pi)
                \ =\ 
                \sum\limits_{\akrlbv \in \cS^{\pi}_{\frb, \frP}}
                \kappa_{\akrlbv, i, \eta} \cdot c''^\pi_{\akrlbv},
            \end{equation*}
            Therefore, $\bhg{\cH'}{v^\shadowtext_i} = \sum\limits_{j \in \intsinterval{\numbhperiods{\pibar_i}}} b''_{i, j} \bhperiod{\pibar_i}{j}$.

        \item
            Recall that the edges of the $v^\shadowtext_i$ in $\cH'$ are redirected from $v^\prekingtext_i$ in $\cH$,
            which implies that
            \begin{equation*}
                \bhg{\cH}{v^\prekingtext_i}
                \ =\ 
                \bhg{\cH'}{v^\shadowtext_i} + \bhg{\cH'}{v^\kingtext_i}.
            \end{equation*}
            By Claim~\ref{claim:h}, it holds that
            \begin{equation*}
                \bhg{\cH}{v^\prekingtext_i}
                \ =\ 
                \bfbbar_i + \sum\limits_{j \in \intsinterval{\numbhperiods{\pibar_i}}} b_{i, j} \cdot \bhperiod{\pibar_i}{j}
            \end{equation*}
            Thus, we have
            \begin{equation*}
                \begin{aligned}
                    \bhg{\cH'}{v^\kingtext_i}
                    \ =\ &
                    \bhg{\cH}{v^\prekingtext_i} - \bhg{\cH'}{v^\shadowtext_i} \\
                    \ =\ &
                    \bfbbar_i
                    + \sum\limits_{j \in \intsinterval{\numbhperiods{\pibar_i}}} b_{i, j} \bhperiod{\pibar_i}{j}
                    - \sum\limits_{j \in \intsinterval{\numbhperiods{\pibar_i}}} b''_{i, j} \bhperiod{\pibar_i}{j} \\
                    \ =\ &
                    \bfbbar_i
                    + \sum\limits_{j \in \intsinterval{\numbhperiods{\pibar_i}}} b'_{i, j} \bhperiod{\pibar_i}{j},
                \end{aligned}
            \end{equation*}
            where the final equality follows from Definition~\ref{def:a''ib''ij}, $b'_{i, j} = b_{i, j} - b''_{i, j}$.
            
        \item
            Recall that, by Definition~\ref{def:struct_realize_gbv},
            the kernel part of $\cH'$ realizes $\frb'$ with witness $\rep_{\cH'}$
            if, for every 1-type $\pi$, the sum of $\rep_{\cH'}(v)$ over all vertices in the kernel part of 1-type $\pi$
            is exactly $\frb'(\pi,-)$.
            We verify this by computing the sum and comparing it with $\frb'$ entrywise, considering each case separately.
            For every 1-type $\pi$, let $U_\pi$ denote the set of vertices of 1-type $\pi$ in the kernel part of $\cH'$.
            \begin{compactitem}
                \item \itemtitle{King 1-type $\pi$, base components}
                    Let $\pi$ be a king 1-type.
                    We first consider the base components.
                    We have
                    \begin{equation*}
                        \begin{aligned}
                            \sum_{v \in U_\pi} \rep_{\cH'}(v)(\basetext, j)
                            \ =\ &
                            \sum_{i \in I_\pi} \rep_{\cH'}(v^\kingtext_i)(\basetext, j) \\
                            \ =\ &
                            \sum_{i \in I_\pi} \rep_{\cH}(v^\prekingtext_i)(\basetext, j) \\
                            \ =\ &
                            \frb'(\pi, \basetext, j).
                        \end{aligned}
                    \end{equation*}

                \item \itemtitle{King 1-type $\pi$, period components}
                    We next consider the period components.
                    By the definition of $\rep_{\cH'}$, we have
                    \begin{equation*}
                        \begin{aligned}
                            \sum_{v \in U_\pi} \rep_{\cH'}(v)(\periodtext, j)
                            \ =\ &
                            \sum_{i \in I_\pi} \rep_{\cH'}(v^\kingtext_i)(\periodtext, j) \\
                            \ =\ &
                            \sum_{i \in I_\pi} b_{i, j}.
                        \end{aligned}
                    \end{equation*}
                    Since
                    $\set{x_i \gets 0, y_{i, j} \gets b'_{i, j}, z^\pi_{\akrlbv} \gets c'^\pi_{\akrlbv}}$
                    is a solution to the weak village system $\villagesystem_{\frb',\frP'}^-$,
                    the constraint {\bf counting shadows} yields
                    \begin{equation*}
                        \begin{aligned}
                            \sum_{i \in I_\pi} b_{i, j}
                            \ =\ &
                            \frb'(\pi, \periodtext, j)
                            +
                            \sum\limits_{i \in \intsinterval{\numouterperiods}} 0 \cdot \frp_i(\pi, \periodtext, j) \\
                            \ =\ &
                            \frb'(\pi, \periodtext, j).
                        \end{aligned}
                    \end{equation*}
                    Therefore, $\sum_{v \in U_\pi} \rep_{\cH'}(v)(\periodtext,j) = \frb'(\pi,\periodtext,j)$.

                \item \itemtitle{Peasant 1-type $\pi$, base components}
                    Let $\pi$ be a peasant 1-type.
                    Recall that, for every $(\pi,\krlbv)$ kernel-vertex $v$ (Definition \ref{def:villageandkernelpartvertices}),
                    $\rep_{\cH'}(v)(\periodtext,j)=0$ whenever $\bhperiod{\pi}{j}$ is not king-reaching.
                    Moreover, for every $(\pi,j)$ kernel-vertex $v$,
                    $\rep_{\cH'}(v)(\periodtext,j)=1$,
                    while all other entries of $\rep_{\cH'}(v)$ are $0$.

                    Recall that the assignment
                    $\set{x_i \gets 0, y_{i, j} \gets b'_{i, j}, z^\pi_{\akrlbv} \gets c'^\pi_{\akrlbv}}$
                    is a solution to the weak village system $\villagesystem_{\frb',\frP'}^-$.
                    We use the constraints of this system to compute the contribution of vertices in the kernel part.

                    We first consider the base components.
                    By the definition of the $(\pi,\krlbv)$ kernel-vertices,
                    \begin{equation*}
                        \sum_{v \in U_\pi} \rep_{\cH'}(v)(\basetext, j)
                        \ =\ 
                        \sum_{\krlbv \in \cS^\pi_{\frb', \frP'}} c'^\pi_\krlbv \cdot \krlbv(\basetext, j)
                    \end{equation*}
                    Since the above assignment is a solution to $\villagesystem_{\frb',\frP'}^-$,
                    the constraint {\bf counting peasant bases} yields
                    \begin{equation*}
                        \begin{aligned}
                            \sum_{\krlbv \in \cS^\pi_{\frb', \frP'}} c'^\pi_\krlbv \cdot \krlbv(\basetext, j)
                            \ =\ &
                            \frb'(\pi, \basetext, j) + \sum_{i \in \intsinterval{\numouterperiods}} 0 \cdot \frp_i(\pi, \basetext, j) \\
                            \ =\ &
                            \frb'(\pi, \basetext, j)
                        \end{aligned}
                    \end{equation*}
                    Therefore, $\sum_{v \in U_\pi} \rep_{\cH'}(v)(\basetext,j) = \frb'(\pi,\basetext,j)$.

                \item \itemtitle{Peasant 1-type $\pi$, period components}
                    We next consider the period components.
                    Let $1 \le j \le \numbhperiods{\pi}$ and suppose first that $\bhperiod{\pi}{j}$ is king-reaching.
                    In this case, the $(\pi,j)$ kernel-vertices make no contribution,
                    and hence
                    \begin{equation*}
                        \sum_{v \in U_\pi} \rep_{\cH'}(v)(\periodtext, j)
                        \ =\ 
                        \sum_{\krlbv \in \cS^\pi_{\frb', \frP'}} c'^\pi_\krlbv \cdot \krlbv(\periodtext, j)
                    \end{equation*}
                    Since the above assignment satisfies the constraint {\bf counting peasant periods}, we obtain
                    \begin{equation*}
                        \begin{aligned}
                            \sum_{\krlbv \in \cS^\pi_{\frb', \frP'}} c'^\pi_\krlbv \cdot \krlbv(\periodtext, j)
                            \ =\ &
                            \frb'(\pi, \periodtext, j) + \sum_{i \in \intsinterval{\numouterperiods}} 0 \cdot\frp_i(\pi, \periodtext, j) \\
                            \ =\ &
                            \frb'(\pi, \periodtext, j)
                        \end{aligned}
                    \end{equation*}
                    Therefore, $\sum_{v \in U_\pi} \rep_{\cH'}(v)(\periodtext,j) = \frb'(\pi,\periodtext,j)$.

                    Finally, suppose that $\bhperiod{\pi}{j}$ is not king-reaching.
                    In this case, the $(\pi,\krlbv)$ kernel-vertices make no contribution.
                    By construction, there are exactly $\frb'(\pi,\periodtext,j)$ many $(\pi,j)$ kernel-vertices,
                    each contributing $1$ to this entry.
                    Hence,
                    \begin{equation*}
                        \sum_{v \in U_\pi} \rep_{\cH'}(v)(\periodtext,j)
                        \ =\
                        \sum_{\text{$v$ is a $(\pi,j)$ kernel-vertex}} 1
                        \ =\
                        \frb'(\pi,\periodtext,j).
                    \end{equation*}
            \end{compactitem}
         
        \item
            Recall that, by Definition~\ref{def:struct_realize_gbv},
            $\cH'$ realizes a $\gbv$ $\agbv$ with witness $\rep_{\cH'}$
            if, for every 1-type $\pi$, the sum of $\rep_{\cH'}(v)$ over all vertices in $\cH'$ is exactly $\frb'(\pi,-)$.

            Because the kernel part and the village part form a partition of the vertices of $\cH'$,
            the $\gbv$ realized by $\cH'$ with witness $\rep_{\cH'}$ is the sum of the $\gbv$s realized by its kernel and village parts with the same witness.

            Since $\cH'$ realizes $\agbv$ with realization witness $\rep_{\cH'}$,
            and its kernel part realizes $\frb'$ with the same witness (according to Definition \ref{def:subset_dstruct_realize_gbv}),
            it follows that the village part of $\cH'$ realizes
            $\agbv - \frb' = \sum_{i \in \intsinterval{\numouterperiods}} a'_i \frp_i$.
    \end{compactenum}
\end{claimproof}

\myparagraph{Step 3B: removing edges crossing from the kernel part to the village part}
After Step 3a, there may still be some edges between the kernel part and the village part in $\cH'$.
This is undesirable, as our goal is to split the model into two D-structures, one of which serves as the desired kernel.
Note that these edges do not involve any $v^\kingtext_i$ or $v^\shadowtext_i$.

In the next stage, we will obtain another D-structure $\cH''$ from $\cH'$ by eliminating all these undesired edges.
This will be achieved by repeatedly applying the swap operation defined in Lemma~\ref{lem:operation}.
This process may introduce new multi-edges and self-loops for vertices with peasant 1-type.
However, this is acceptable, as multi-edges and self-loops involving such vertices are permitted in kernels (Definition \ref{def:kernalkingnoblephantomnoble}).

\begin{definition}[$\cH''$]
    We start by setting $\cH''$ to $\cH'$ and $\rep_{\cH''}$ to $\rep_{\cH'}$.
    Suppose we have an edge $e$ from $u$ to $v$ where $u$ is in the kernel part with 1-type $\pi_1$ and $v$ is in the village part with 1-type $\pi_2$ and $\twotypeof(e)=\eta$.
    Let $V_1$ be the set of vertices with 1-type $\pi_1$ in the kernel part
    and $V_2$ be the set of vertices with 1-type $\pi_2$ in the kernel part.
    By Claim~\ref{claim:hp}, the kernel part of $\cH'$ satisfies the $\gbv$ $\frb'$.
    Recall that, by Definition~\ref{def:targetrep}, $\frb' \in \cL(\frb;\frP)$,
    which implies that $\frb'$ is a solution of $\match$,
    since we assume that $\cL(\frb;\frP) \subseteq \sol{\overapprox{\anfgptwosentence \land \agcc}(\frx)}$.
    Therefore, we have
    \begin{equation*}
        \sum_{v \in V_1} \bhg{\cH''}{v}(\eta, \pi_2)
        \ =\ 
        \sum_{v \in V_2} \bhg{\cH''}{v}(\dual{\eta}, \pi_1).
    \end{equation*}

    Note that the edge $e$ contributes only to the left-hand side of the equality.
    On the other hand, every edge from $V_1$ to $V_2$ with 2-type $\eta$ contributes one to each side,
    as does every edge from $V_2$ to $V_1$ with $\dual{\eta}$.
    Therefore, 
    one of the following holds:
    (i) there exists another edge $e'$ from $u'$ to $v'$ where $u'$ is in $\cK$, $v'$ is in $\cV$
    and $\twotypeof(e')=\dual{\eta}$; or
    (ii) there exists another edge $e'$ from $v'$ to $u'$ where $u'$ is in the kernel, $v'$ is in $\cV$
    and $\twotypeof(e')=\eta$.
    Otherwise -- if there is no such edge --
    then the left-hand side would be strictly larger than the right-hand side,
    contradicting the equality above and hence the fact that $\frb'$ is a solution of $\match$.
    Then, we can perform swapping between these two edges $e$ and $e'$ to remove $e$ and $e'$ and add one new edge
    from $u$ to $u'$ and another new edge from $v$ to $v'$.
    When $u=u'$, this creates a self-loop in the kernel, but this is allowed in the kernel.
    The 2-types of the new edges are defined as in Lemma~\ref{lem:operation}.
\end{definition}

We argue for some properties of $\cH''$.
\begin{claim}[$\cH''$ properties] \label{claim:hpp}
    \leavevmode
    \begin{compactitem}
        \item
            The following structural properties hold for $\cH''$.
            \begin{compactitem}
                \item
                    There are no edges between any two shadow vertices $v^\shadowtext_i$ and $v^\shadowtext_j$.

                \item
                    There are no multi-edges incident to any $v^\kingtext_i$.

                \item
                    There are no edges between the kernel part and the village part.
            \end{compactitem}

        \item
            The kernel part of $\cH''$ realizes the $\gbv$ $\frb''$, with $\rep_{\cH''}$ as a witness.

        \item
            The village part of $\cH'$ realizes the $\gbv$ $\sum_{i \in \intsinterval{\numouterperiods}} a'_i \frp_i$, with $\rep_{\cH''}$ as a witness.
    \end{compactitem}
\end{claim}
\begin{claimproof}
    \leavevmode
    \begin{compactitem}
        \item
            Note that the swapping procedure does not involve any $v^\kingtext_i$ or $v^\shadowtext_i$.
            The structural properties for $\cH''$ follow from the construction and properties for $\cH'$ in Claim~\ref{claim:hp}.
        
        \item
            By Lemma~\ref{lem:operation}, the swapping operation does not change the behavior of any vertex.
            Since $\rep_{\cH'}$ is a realization witness of $\cH'$, for every vertex $v$ of $\cH''$ we have
            \begin{equation*}
                \bhg{\cH''}{v}
                \ =\ 
                \bhg{\cH'}{v}
                \ =\ 
                \eval{\rep_{\cH'}(v)}
                \ =\ 
                \eval{\rep_{\cH''}(v)}
            \end{equation*}

            In Claim~\ref{claim:hp},
            we showed that the kernel part of $\cH'$ realizes $\frb'$ with witness $\rep_{\cH'}$.
            Since the swapping operation changes neither the set of vertices in the kernel part nor their assigned $\lbv$s under the realization witness,
            the kernel part of $\cH''$ also realizes $\frb'$ with witness $\rep_{\cH''}$.

        \item
            The property follows by an argument similar to the previous one.
    \end{compactitem}
\end{claimproof}

\myparagraph{Step 3C: showing that the kernel part is the desired kernel}
In Claim~\ref{claim:target_assign_is_sol},
we showed that the target assignment is a solution to the homogeneous village system $\villagesystem_{\frb',\frP'}^\homo$.
Our remaining goals are to construct a bounded size kernel of the target linear representation $\frb'$ and $\frP'$ and to show that the target assignment is friendly with this kernel.

Let $\cK$ be the D-structure corresponding to the kernel part of $\cH''$,
and let $\rep_\cK$ be the witness obtained by restricting $\rep_{\cH''}$ to the kernel part of $\cH''$.
We argue that $\cK$ is a kernel of $\frb$ and $\frP$,
where the $v^\kingtext_i$ serve as kings, the $(\pi,\krlbv)$ kernel-vertices serve as nobles, and the $(\pi,i)$ kernel-vertices serve as phantom nobles.

\begin{claim} \label{claim:k}
    $\cK$ is a kernel of $\frb'$ and $\frP'$.
\end{claim}

\begin{claimproof}
    We verify that $\cK$ satisfies the conditions for a kernel in Definition \ref{def:kernalkingnoblephantomnoble}.
    By Claim~\ref{claim:hpp},
    the kernel part of $\cH''$ realizes $\frb'$ by $\rep_{\cH''}$,
    which implies that $\cK$ realizes $\frb'$ by $\rep_\cK$.
    \begin{compactitem}
        \item[\emph{K1.}]
            Since $\cK$ is just the kernel part of $\cH''$ and
            there are no phantom vertices with king 1-types in the kernel part of $\cH''$ by Claim~\ref{claim:hpp},
            the same holds for $\cK$.

        \item[\emph{K2.}]
            Since $\cK$ is just the kernel part of $\cH''$ and no multi-edge of the kernel part of $\cH''$ involves a $v^\kingtext_i$ by Claim~\ref{claim:hpp}, the same holds for $\cK$.
            Moreover, there is no edge in $\cH''$ connecting a $v^\kingtext_i$ to a phantom vertex by Claim~\ref{claim:hpp}, and hence the same holds for $\cK$.
    \end{compactitem}
\end{claimproof}

It remains to show that the size of $\cK$ is bounded by a computable function,
and the target assignment (Definition \ref{def:targetrep}) is $\cK$-friendly.

\begin{claim}
    The size of $\cK$ is bounded by a computable function.
\end{claim}
\begin{claimproof}
    By Claim~\ref{claim:k},
    $\cK$ realizes the $\gbv$ $\frb'$.
    Thus, the size of $\cK$ is bounded by $\frb'$.
    Moreover, the size of $\frb'$ is bounded by a computable function, as shown in Step~2.
    Hence, the size of $\cK$ is also bounded by a computable function.
\end{claimproof}

\begin{claim}
    The target assignment from Definition \ref{def:targetassign}, 
    $\set{x_i \gets a''_i, y_{i, j} \gets b''_{i, j}, z^\pi_{\akrlbv} \gets c''^\pi_{\akrlbv}}$,
    is $\cK$-friendly (Definition \ref{def:village_friendly}).
\end{claim}

\begin{claimproof}
    \leavevmode
    \begin{compactitem}
        \item \itemtitle{Absorption of phantom nobles in $\cK$ constraint}
            Note that, by definition, every peasant 1-type $\pi$ is incorporated basely by some $\frp_i$ for $1 \le i \le \numouterperiods$.
            The assignment constructed in Step~2 guarantees that $a''_i > 0$ for every $1 \le i \le \numouterperiods$.
            Thus for every phantom noble in $\cK$ realizing the peasant 1-type $\pi$,
            there exists $1 \le i \le \numouterperiods$ such that
            $a_i > 0$ and $\pi$ is incorporated basely by $\frp_i$.
            Therefore the assignment satisfies the constraint. 

        \item \itemtitle{Absorption of multi-edges in $\cK$ constraint} 
            The only multi-edges that may arise in the kernel part of $\cH''$ are those created by swapping an edge between the kernel part and the village part of $\cH'$.
            In the construction of $\cH''$, we showed
           in Claim \ref{claim:hpp}
            that for such an edge of 2-type $\eta$,
            there exists a corresponding edge from the village part to the kernel part of 2-type $\eta$.
            Since the village part realizes $\gbv$
            $\sum_{i \in \intsinterval{\numouterperiods}} a''_i \frp_i$ (Definition \ref{def:subset_dstruct_realize_gbv}),
            there exists some $1 \le i \le \numouterperiods$ such that $\frp_i$ incorporates $\eta$.
            Hence, the desired property holds.

        \item \itemtitle{Absorption of self-loops in $\cK$ constraint} 
            The only self-loops in the kernel part of $\cH''$ that are present after the swapping procedure emerge 
            from swapping with an edge between vertices in the kernel part of $\cH''$ and vertices in the village part of $\cH''$.
            Thus the property holds by an argument similar to the previous property.
    \end{compactitem}
\end{claimproof}

\end{proof}

%% file: fig/intuition.tex
\begin{tikzpicture}[scale=.8]
    \begin{scope}[shift={(0, 0)}]
        \draw[bdr] (0.5, -2.5) rectangle (1.5, 2.5);

        \node[newvtx, red] (k1k1) at (1.0, 2.0) {\scriptsize 1};
        \node[newvtx, red] (k1k2) at (1.0, 0.0) {\scriptsize 2};
        \node[newvtx, red] (k1k3) at (1.0,-2.0) {\scriptsize 3};

        \draw[edge]                (k1k1) to (k1k2);
        \draw[edge, bend right=40] (k1k1) to (k1k3);
    \end{scope}

    \begin{scope}[shift={(3, 3)}]
        \draw[bdr] (-0.5,-0.5) rectangle (2.5, 0.5);
        \node[anchor=south] at (2.0, 0.5) {\footnotesize $(\pi_1, \akrlbv_{1, 1})$};

        \node[vtx, blue] (k1n1) at (0.0, 0.0) {};
        \node[vtx, blue] (k1n2) at (1.0, 0.0) {};

        \draw[edge, bend right]  (k1k1) to (k1n1);
        \draw[edge, bend right]  (k1k1) to (k1n2);
        \draw[edge, bend right]  (k1k2) to (k1n1);
        \draw[edge, bend right]  (k1k2) to (k1n2);
    \end{scope}

    \begin{scope}[shift={(3, 0)}]
        \draw[bdr] (-0.5,-0.5) rectangle (2.5, 0.5);
        \node[anchor=south] at (2.0, 0.5) {\footnotesize $(\pi_1, \akrlbv_{1, 2})$};

        \node[vtx, blue] (k2n1) at (0.0, 0.0) {};
        \node[vtx, blue] (k2n2) at (1.0, 0.0) {};
        \node[vtx, blue] (k2n3) at (2.0, 0.0) {};

        \draw[edge]  (k1k3) to (k2n1);
        \draw[edge]  (k1k3) to (k2n2);
        \draw[edge]  (k1k3) to (k2n3);

        \draw[edge, bend right]  (k1k2) to (k2n1);
        \draw[edge, bend right]  (k1k2) to (k2n2);
        \draw[edge, bend right]  (k1k2) to (k2n3);
    \end{scope}

    \node[] at (3.75, -1.25) {\textbf{$\vdots$}};

    \begin{scope}[shift={(3, -3)}]
        \draw[bdr] (-0.5,-0.5) rectangle (2.5, 0.5);
        \node[anchor=south] at (2.0, 0.5) {\footnotesize $(\pi_i, \akrlbv_{i, j})$};

        \node[vtx, blue] (k3n1) at (0.0, 0.0) {};

        \draw[edge] (k1k3) to (k3n1);
    \end{scope}

    \node[] at (3.75, -4) {\textbf{$\vdots$}};

    \begin{scope}[shift={(8, 3)}]
        \draw[bdr] (-0.5,-0.5) rectangle (2.5, 0.5);
        \node[anchor=south] at (2.0, 0.5) {\footnotesize $(\pi_1, 1)$};

        \node[vtx, dashed, blue] (s1n1) at (0.0, 0.0) {};
        \node[vtx, dashed, blue] (s1n2) at (1.0, 0.0) {};
        \node[vtx, dashed, blue] (s1n3) at (2.0, 0.0) {};

        \draw[edge] (s1n2) to (s1n3);
        \draw[edge, bend left] (s1n1) to (k1n1);
    \end{scope}

    \begin{scope}[shift={(8, 0)}]
        \draw[bdr] (-0.5,-0.5) rectangle (2.5, 0.5);
        \node[anchor=south] at (2.0, 0.5) {\footnotesize $(\pi_1, 2)$};

        \node[vtx, dashed, blue] (s2n1) at (0.0, 0.0) {};

        \draw[edge] (s2n1) to (k1n2);
    \end{scope}

    \node[] at (8.75, -1.25) {\textbf{$\vdots$}};

    \begin{scope}[shift={(8, -3)}]
        \draw[bdr] (-0.5,-0.5) rectangle (2.5, 0.5);
        \node[anchor=south] at (2.0, 0.5) {\footnotesize $(\pi_i, j)$};

        \node[vtx, dashed, blue] (s3n1) at (0.0, 0.0) {};
        \node[vtx, dashed, blue] (s3n2) at (1.0, 0.0) {};
        \node[vtx, dashed, blue] (s3n3) at (2.0, 0.0) {};

        \draw[edge] (k2n1) to (s3n1);
        \draw[edge] (k2n2) to (s3n3);
        \draw[edge] (k2n3) to (s3n2);
    \end{scope}

    \node[] at (8.75, -4) {\textbf{$\vdots$}};
\end{tikzpicture}

%% file: 4-gcc.tex

\subsection{Putting it all together to prove Theorem \ref{thm:ourlogic}}

We now complete the proof of Theorem~\ref{thm:ourlogic}.
First, as outlined in Section~\ref{sec:outline},
the satisfiability of a $\ourlogic$ sentence can be effectively 
reduced to the satisfiability of a conjunction of 
a $\gptwo$ sentence in normal form and a $\gcc$. 

Next, 
Lemmas~\ref{lem:pseudo_model_tight} 
and~\ref{lem:sol_semilinear_strong}
show that it suffices to consider the realizability problem for 
a strongly concretifying linear subset $\cL(\frb; \frP)$ of $\sol{\overapprox{\anfgptwosentence \land \agcc}}$.
Moreover, there exists an algorithm to decide whether a village system is kernel-friendly.

\begin{lemma} \label{lem:village_algo}
    Let $\cK$ be a kernel of $\frb$ and $\frP$.
    There is an algorithm that decides whether the homogeneous village system $\villagesystem_{\frb, \frP}^\homo$ has a $\cK$-friendly solution.
\end{lemma}

\begin{proof}
    Observe that
    the conditions of $\cK$-friendly
    can each be expressed as systems of integer linear inequalities over the variables of the village system.

    Therefore, the conjunction of the homogeneous village system and the three extra ``friendly conditions'' is again an existential Presburger formula.
    Consequently, deciding whether the homogeneous village system has a $\cK$-friendly solution reduces to deciding the satisfiability of an existential Presburger formula,
    which can be done effectively.
\end{proof}

We are now ready to prove Theorem~\ref{thm:ourlogic}.

\begin{proof}[Proof of Theorem~\ref{thm:ourlogic}]
    Let $\anfgptwosentence \land \agcc$ be a conjunction of a $\GPtwo$ sentence and a $\GCC$.
    By Lemma~\ref{lem:sol_semilinear_strong}, we can  construct a finite family of linear sets of $\gbv$s\\
        $\set{
            \cL(\frb_1;\frP_1),
            \ldots,
            \cL(\frb_m;\frP_m)
        }$
    such that for every $1 \le i \le m$,
    $\frb_i$ and $\frP_i$ are strongly concretifying and
    $\cL(\frb_i; \frP_i) \subseteq \sol{\overapprox{\anfgptwosentence \land \agcc}(\frx)}$.
    Moreover, every $\gbv$ realized by a finite model of $\anfgptwosentence \land \agcc$ belongs to one of these linear sets.

    For each $1 \le i \le m$,
    we enumerate all $\frb' \in \cL(\frb_i;\frP_i)$ of size at most $f(\frb_i,\frP_i)$;
    all subsets $\frP' \subseteq \frP_i$ such that $\frb'$ and $\frP'$ are strongly concretifying;
    and all kernels $\cK$ that realize $\frb'$.
    Since the bound $f(\frb_i,\frP_i)$ is computable, all of these can be enumerated effectively.

    For every such kernel $\cK$,
    we construct the corresponding homogeneous village system
    $\villagesystem_{\frb', \frP'}^\homo$.
    By Lemma~\ref{lem:village_algo}, we can decide whether this system has a $\cK$-friendly solution.
    The procedure returns \textsc{SAT} if a kernel has a homogeneous 
    village system with a $\cK$-friendly solution, and returns \textsc{UNSAT} otherwise.

    We verify the correctness of this procedure.
    Suppose first that the procedure returns \textsc{SAT}.
    Then there exists $1 \le i \le m$, $\frb' \in \cL(\frb_i;\frP_i)$, $\frP'\subseteq\frP_i$, and a kernel $\cK$ of $\frb'$ and $\frP'$ realizing $\frb'$,
    such that the homogeneous village system
    $\villagesystem_{\frb', \frP'}^\homo$ has a $\cK$-friendly solution.
    By Theorem~\ref{thm:solvetovillage},
    there exists a village compatible with $\cK$;
    by Theorem~\ref{thm:pump},
    there is an asymmetric self-loop free structure $\cG$ realizing a $\gbv$ in $\cL(\frb';\frP')$.
    Since
    $\cL(\frb'; \frP') \subseteq \cL(\frb_i; \frP_i) \subseteq \sol{\overapprox{\anfgptwosentence \land \agcc}(\frx)}$,
    by Lemma~\ref{lem:pseudo_model_tight}, the structure $\cG$ is a model of $\anfgptwosentence \land \agcc$.
    Hence $\anfgptwosentence \land \agcc$ is satisfiable.

    Conversely, suppose that $\anfgptwosentence \land \agcc$ has a finite model $\cG$.
    Let $\agbv$ be a $\gbv$ realized by $\cG$.
    By Lemma~\ref{lem:sol_semilinear_strong}, there exists
    $1 \le i \le m$ such that $\agbv \in \cL(\frb_i;\frP_i)$.
    Since $\agbv$ is realizable,
    Theorem~\ref{thm:sattosolve} yields $\frb' \in \cL(\frb_i;\frP_i)$ within the prescribed size bound,
    $\frP' \subseteq \frP_i$,
    and a kernel $\cK$ for $\frb'$ and $\frP'$ realizing $\frb'$,
    such that the corresponding homogeneous village system has a $\cK$-friendly solution.
    This kernel is considered by the enumeration, so the procedure returns \textsc{SAT}.

    The procedure performs only finitely many effective computations.
    Therefore, the finite satisfiability problem of $\ourlogic$ is decidable.
\end{proof}

%% file: 7-conclusion.tex

\section{Discussion} \label{sec:conc}

We have shown decidability of a logic with both ``local counting quantifiers'' and ``global counting quantifiers'', the first
such result we are aware of. The result has immediate applications to verifying graph neural networks with global readout.

There are a number of  extensions of our techniques. First, although we follow the literature in referring to ``Presburger quantifiers'', the quantifiers we consider do not support modulus-class restrictions (e.g. the number must be divisible by $3$ etc), which are allowed in Presburger arithmetic. It is routine to add support for these in both local and global quantifiers.
The static analysis problem we focus on is satisfiability. Decidability for other problems like equivalence follows immediately. But in fact we can support more sophisticated statements like ``is there a model where the fraction of elements satisfying $\phi(x)$
is at least $\frac{1}{4}$'': this is again easy, since we can encode the property in another GCC.
A natural subclass of our logic is where local aggregation is ``outgoing only''. Here the proofs simplify substantially;
see the appendix.

We leave open the exact complexity of the satisfiability problem, noting only that our procedure gives a tower bound.
We also note that our logic is incomparable to $\ctwo$, since we do not allow  new one-variable formulas
to be formed via unguarded counting (e.g. constraints on the number of non-neighbors). Ideally one would like a logic that subsumes
$\ourlogic$ and also two-variable logic with unbounded almost-periodic counting~\cite{tonymichael}.

\forarxiv{We have also applied the techniques we developed to
provide a new result on the \emph{spectrum} for $\gptwo$ sentences. 
The spectrum of a logical sentence is the set of cardinalities of finite models of the sentence.
It has been shown that the spectrum of $\ctwo$ is a semilinear set: definable in Presburger arithmetic~\cite{ianrevisited,spectrumeryktony}.
We can show that the same holds for $\gptwo$. This result is surprising, given the negative result on $1$-type spectra from \cite{uscsl} mentioned above. To keep this manuscript manageable, we leave the details for a second preprint.}

Finally we note that our results deal with satisfiability over structures which (by definition) are free of self-loops.
We believe it is easy to reduce the general case to this one, but have not verified this.

%% file: appendix.tex
\input{app-lemmalinearalgebra}
\input{app-normalform}

\input{app-opspreserve}
\input{app-almostcsl}
\input{app-constructDstruct}
\input{app-outgoingonly}

%% file: app-lemmalinearalgebra.tex
\section{Bounds on solution size} \label{app:solutionsizelemma}

The following lemma on the solutions of linear systems is used throughout the body.

\begin{lemma}[\cite{Pottier}, also see Proposition 4 in~\cite{semilinear}]\label{lem:minimal}
    For every  integer linear system $\cQ(\bfx): \bfA \bfx = \bfc$,
    where $\bfA \in \bbZ^{m \times n}$ and $\bfc \in \bbZ^{m}$,
    and for every non-negative solution $\bfv$ of $\cQ(\bfx)$,
    there exists another non-negative solution $\bfu$ such that each entry of $\bfu$ is bounded by
    $\left((n+1) \cdot M + M + 1\right)^{m}$,
    where $M$ is the maximum absolute value of the entries of $\bfA$ and $\bfc$,
    and $\bfu \le \bfv$ componentwise.
\end{lemma}

%% file: app-normalform.tex

\section{Proof of Lemma~\ref{lem:normal}: Normal form for $\ourlogic$}

\normalform*

\begin{proof}

We will need to define an inductive translation that handles formulas that may have free variables; they will translate to a conjunction of a formula with the same free variables
and a $\gcc$. We recall that a global Presburger quantifier is an open formula with one implicit free variable, so it will translate to a formula with one free variable as well.

    Suppose that the input $\ourlogic$ formula is over the vocabulary $\vocab$.
    We normalize the formula in the following three steps.

    \myparagraph{Step 1: Convert global Presburger quantifiers into $\gcc$s}

    We recall from Definition \ref{def:gcc} that a $\gcc$ is a Boolean combination of atomic $\gcc$s, and thus a $\gcc$ can include both $\geq$ and $\leq$ inequalities with constants. Our first part will thus be a straightforward normalization that converts global cardinality constraints into a $\gcc$.
    Formally a global Presburger quantifier has an implicit free variable, while a $\gcc$ defines a property of a structure: the $\gcc$ that the generate will be equivalent to the global Presburger quantifier in the sense that the $\gcc$ holds in a structure if and only if the global Presburger quantifier holds for all elements of the structure.
    
    We simultaneously define four mappings by structural induction on $\varphi$:
    $\transgptwo$, which inputs a $r$-ary $\ourlogic$ formula $\anourlogicsentence$ and produces a $r$-ary $\gptwo$ formula; 
    $\transgptwoex$, which inputs a $\ourlogic$ formula  $\anourlogicsentence$ and produces a $\gptwo$ sentence;
    $\transgcc$, which inputs a $\ourlogic$ formula $\anourlogicsentence$ and produces a $\gcc$;
    $\vocabext$, which inputs a $\ourlogic$ formula  $\anourlogicsentence$ and produce an extended vocabulary;
    and another function $\transstructp{\anourlogicsentence}{\alogicstructure}$, which takes as input a $\ourlogic$ formula $\anourlogicsentence$ and 
    a $\vocab$-structure $\alogicstructure$ and produces a $\vocabextp{\anourlogicsentence}$-structure expanding $\alogicstructure$.

    \begin{compactitem}
        \item \itemtitle{Atomic formulas}
            For atomic formulas $\alpha \in \set{U(x), R(x, y)}$, define
            \begin{equation*}
                \transgptwop{\alpha}       \ =\ \alpha,
                \quad                     
                \transgptwoexp{\alpha}     \ =\ \top,
                \quad                     
                \transgccp{\alpha}         \ =\ \top,
                \quad                     
                \vocabextp{\alpha}         \ =\ \vocab,
                \quad
                \transstructp{\alpha}{\alogicstructure} \ =\ \alogicstructure.
            \end{equation*}

        \item \itemtitle{Negation}
            For $\varphi := \neg \psi$, define
            \begin{equation*}
                \begin{aligned}
                    &
                    \transgptwop{\varphi}       \ =\ \neg\transgptwop{\psi},
                    \quad
                    \transgptwoexp{\varphi}     \ =\ \transgptwoexp{\psi},
                    \quad
                    \transgccp{\varphi}         \ =\ \transgccp{\psi}, \\
                    &
                    \vocabextp{\varphi}         \ =\ \vocabextp{\psi},
                    \quad
                    \transstructp{\varphi}{\alogicstructure} \ =\ \transstructp{\psi}{\alogicstructure}.
                \end{aligned}
            \end{equation*}

        \item \itemtitle{Conjunction}
            For $\varphi := \psi_1 \land \psi_2$, define
            \begin{equation*}
                \begin{aligned}
                    &
                    \transgptwop{\varphi} \ =\ \transgptwop{\psi_1}\ \land\ \transgptwop{\psi_2},
                    \quad
                    \transgptwoexp{\varphi}\ =\ \transgptwoexp{\psi_1}\ \land\ \transgptwoexp{\psi_2}, \\
                    &
                    \transgccp{\varphi}   \ =\ \transgccp{\psi_1}\ \land\ \transgccp{\psi_2},
                    \quad
                    \vocabextp{\varphi}   \ =\ \vocabextp{\psi_1}\ \cup\ \vocabextp{\psi_2}.
                \end{aligned}
            \end{equation*}
            The structure $\transstructp{\varphi}{\alogicstructure}$
            is the common expansion of
            $\transstructp{\psi_1}{\alogicstructure}$ and
            $\transstructp{\psi_2}{\alogicstructure}$.

        \item \itemtitle{Ordinary quantification}
            For $\varphi := Q x\ \psi(x)$, where $Q \in \set{\forall, \exists}$, define
            \begin{equation*}
                \begin{aligned}
                    &
                    \transgptwop{\varphi}       \ =\ Q x\ \transgptwop{\psi(x)},
                    \quad
                    \transgptwoexp{\varphi}     \ =\ \transgptwoexp{\psi},
                    \quad
                    \transgccp{\varphi}         \ =\ \transgccp{\psi}, \\
                    &
                    \vocabextp{\varphi}         \ =\ \vocabextp{\psi},
                    \quad
                    \transstructp{\varphi}{\alogicstructure} \ =\ \transstructp{\psi}{\alogicstructure}.
                \end{aligned}
            \end{equation*}

        \item \itemtitle{Guarded quantification}
            For guarded quantification, define
            \begin{equation*}
                \begin{aligned}
                    \transgptwop{\forall x\ \epsilon(x, y) \to \psi(x, y)}
                    \ =\ &
                    \forall x\ \epsilon(x, y) \to \transgptwop{\psi(x, y)}, \\
                    \transgptwop{\exists x\ \epsilon(x, y) \land \psi(x, y)}
                    \ =\ &
                    \exists x\ \epsilon(x, y) \land \transgptwop{\psi(x, y)}.
                \end{aligned}
            \end{equation*}
            For both forms, $\transgptwoex$, $\transgcc$, and $\vocabext$ are inherited from $\psi$,
            and the structure translation is $\transstruct{\psi}$.

        \item \itemtitle{Local Presburger quantifier}
            For $P(x) := \sum_{t \in \intsinterval{k}} \lambda_t \cdot \presby{\epsilon_t(x, y) \land \psi_t(x,y)} \circledast \delta $, 
            \begin{equation*}
                \begin{aligned}
                    \transgptwop{P(x)}
                    \ =\ &
                    \sum_{t \in \intsinterval{k}}
                    \lambda_t \cdot
                    \presby{
                        \epsilon_t(x, y) \land \transgptwop{\psi_t(x, y)}
                    } \circledast \delta, \\
                    \transgptwoexp{P(x)}
                    \ =\ &
                    \bigwedge_{t \in \intsinterval{k}} \transgptwoexp{\psi_t(x, y)}, \\
                    \transgccp{P(x)}
                    \ =\ &
                    \bigwedge_{t \in \intsinterval{k}} \transgccp{\psi_t(x, y)}, \\
                    \vocabextp{P(x)}
                    \ =\ &
                    \bigcup_{t \in \intsinterval{k}} \vocabextp{\psi_t(x, y)}.
                \end{aligned}
            \end{equation*}
            The structure $\transstructp{P(x)}{\alogicstructure}$ is
            the common expansion of $\transstructp{\psi_t(x, y)}{\alogicstructure}$ for all $t \in \intsinterval{k}$.

        \item \itemtitle{Global Presburger quantifier}
            Finally, for $P(x) := \sum_{t\in\intsinterval{k}} \lambda_t\cdot\presby{\psi_t(y)} \circledast \delta$,
            letting $H$ be a fresh unary predicate symbol associated with this occurrence of $P(x)$, and $U_1, \ldots, U_k$ be fresh unary predicates,
            define 
            \begin{equation*}
                \begin{aligned}
                    \transgptwop{P(x)}
                    \ =\ &
                    H(x), \\ 
                    \transgptwoexp{P(x)}
                    \ =\ &
                    \bigwedge_{t \in \intsinterval{{k}}} \forall x \ U_t(x) \leftrightarrow \transgptwop{\psi_t(x)}
                    \land
                    \bigwedge_{t \in \intsinterval{k}} \transgptwop{\psi_t(y)}\wedge \forall x ~ T(x)), \\
                    \transgccp{P(x)}
                    \ =\ &
                    \left(
                        P^\star
                        \to
                        \left(
                            \termg{H(x)} = \termg{T(x)}
                        \right)
                    \right)
                    \land
                    \left(
                        \neg P^\star
                        \to
                        \left(
                            \termg{H(x)} = 0
                        \right)
                    \right)
                    \land
                    \bigwedge_{t \in \intsinterval{k}} \transgccp{\psi_t(y)}, \\
                    \vocabextp{P(x)} 
                    \ =\ &
                    \set{H, U_1, \ldots, U_k, T} \cup
                    \bigcup_{t \in \intsinterval{k}}
                    \vocabextp{\psi_t(y)},
                \end{aligned}
            \end{equation*}
            where $P^\star$ is the $\gcc$ defined by 
            \begin{equation*}
                P^\star
                \ :=\ 
                \sum_{t \in \intsinterval{k}} \lambda_t \cdot \termg{U_t(x)}
                \circledast \delta.
            \end{equation*}
            Above $T$ is a fresh unary predicate. For structures that satisfy $\alogicstructure \models \forall x ~ T(x)$, $\abs{T^\alogicstructure}$ is just the cardinality of the domain; we do this
            because we do not allow a term $\termg{\top}$ in a $\gcc$.

            Let $\cB$ be the common expansion of $\transstructp{\psi_t(y)}{\alogicstructure}$ for all $t \in \intsinterval{k}$.
            We define $\transstructp{P(x)}{\alogicstructure}$ as the expansion of $\cB$ interpreting $T$ by the whole domain, and interpreting $H$ and $U_t$ by
            \begin{equation*}
                \begin{aligned}
                    H^{\transstructp{P(x)}{\alogicstructure}}
                    \ =\ &
                    \setc{v \in A}{\alogicstructure, v/x \models P(x)} \\
                    U_t^{\transstructp{P(x)}{\alogicstructure}}
                    \ =\ &
                    \setc{v \in A}{\alogicstructure, v/x \models \psi_t(x)} \\
                \end{aligned}
            \end{equation*}
            Note that since $P(x)$ is a global Presburger quantifier,
            if there exists an element $v \in A$ satisfies $\alogicstructure, v/x \models P(x)$,
            then it holds for every $u \in A$ such that $\alogicstructure, u/x \models P(x)$.
            Therefore, either $H^{\transstructp{P(x)}{\alogicstructure}} = A$ or $H^{\transstructp{P(x)}{\alogicstructure}} = \emptyset$.
    \end{compactitem}

    The intuition behind the translation is that the truth value of a global Presburger quantifier is independent of the choice of elements in the model.
    Thus, it is either satisfied by all elements or by none.
    We encode this truth value using a fresh predicate symbol,
    and the constraints generated by $\transgptwoex$ and $\transgcc$ ensure that this predicate has the intended interpretation in the extended structure.

We conclude Step 1 by arguing for the correctness of the translations.

    \begin{claim} \label{claim:ext_struct}
        For every $\ourlogic$ formula $\varphi(x,y)$,
        for every $\vocab$-structure $\alogicstructure$,
        the following properties hold.
        \begin{compactitem}
            \item
                For every pair of elements $v, u \in A$,
                \begin{equation*}
                    \alogicstructure, v/x, u/y \models \varphi
                    \quad\text{if and only if}\quad
                    \transstructp{\varphi}{\alogicstructure}, v/x, u/y \models \transgptwop{\varphi}.
                \end{equation*}
            \item
                $\transstructp{\varphi}{\alogicstructure} \models \transgptwoexp{\varphi}$.
            
            \item
                $\transstructp{\varphi}{\alogicstructure} \models \transgccp{\varphi}$.
        \end{compactitem}
   The analogous invariant holds also for formulas with one or zero free variables.
    \end{claim}

    \begin{claimproof}
        We prove the claim by structural induction on the formula.
        The base case of atomic formulas is trivial.
        For the inductive step, we consider only conjunction and global Presburger quantification.
        The remaining cases -- negation, ordinary quantification, guarded quantification, and local Presburger quantification -- follow by arguments similar to that for conjunction.
        \begin{compactitem}
            \item \itemtitle{Conjunction}
                Suppose that $\varphi = \psi_1 \land \psi_2$.

                \begin{compactitem}
                    \item
                        Suppose
                        $\alogicstructure, v/x, u/y \models \varphi$,
                        then it holds that
                        \begin{equation*}
                            \alogicstructure, v/x, u/y \models \psi_1
                            \quad\text{and}\quad
                            \alogicstructure, v/x, u/y \models \psi_2.
                        \end{equation*}
                        By the induction hypothesis,
                        it holds that
                        \begin{equation*}
                            \transstructp{\psi_1}{\alogicstructure}, v/x, u/y \models \transgptwop{\psi_1}
                            \quad\text{and}\quad
                            \transstructp{\psi_2}{\alogicstructure}, v/x, u/y \models \transgptwop{\psi_2}.
                        \end{equation*}
                        Since $\transstructp{\varphi}{\alogicstructure}$
                        is the common expansion of
                        $\transstructp{\psi_1}{\alogicstructure}$ and
                        $\transstructp{\psi_2}{\alogicstructure}$.
                        We have
                        \begin{equation*}
                            \transstructp{\varphi}{\alogicstructure}, v/x, u/y \models \transgptwop{\psi_1}
                            \quad\text{and}\quad
                            \transstructp{\varphi}{\alogicstructure}, v/x, u/y \models \transgptwop{\psi_2}.
                        \end{equation*}
                        Therefore
                        \begin{equation*}
                            \transstructp{\varphi}{\alogicstructure}, v/x, u/y \models \transgptwop{\psi_1} \land \transgptwop{\psi_2}
                        \end{equation*}

                        On the other hand,
                        suppose $\transstructp{\varphi}{\alogicstructure}, v/x, u/y \models \transgptwop{\varphi}$.
                        It holds that
                        \begin{equation*}
                            \transstructp{\varphi}{\alogicstructure}, v/x, u/y \models \transgptwop{\psi_1}
                            \quad\text{and}\quad
                            \transstructp{\varphi}{\alogicstructure}, v/x, u/y \models \transgptwop{\psi_2}.
                        \end{equation*}
                        Since $\transstructp{\varphi}{\alogicstructure}$
                        is the common expansion of
                        $\transstructp{\psi_1}{\alogicstructure}$ and
                        $\transstructp{\psi_2}{\alogicstructure}$,
                        and each $\transgptwop{\psi_i}$ is a formula over $\vocabextp{\psi_i}$ for $i \in \set{1, 2}$,
                        the truth values of both formulas are preserved in the common expansion.
                        We have
                        \begin{equation*}
                            \transstructp{\psi_1}{\alogicstructure}, v/x, u/y \models \transgptwop{\psi_1}
                            \quad\text{and}\quad
                            \transstructp{\psi_2}{\alogicstructure}, v/x, u/y \models \transgptwop{\psi_2}.
                        \end{equation*}
                        By the induction hypothesis,
                        it holds that
                        \begin{equation*}
                            \alogicstructure, v/x, u/y \models \psi_1
                            \quad\text{and}\quad
                            \alogicstructure, v/x, u/y \models \psi_2.
                        \end{equation*}
                        Hence, $\alogicstructure, v/x, u/y \models \psi_1 \land \psi_2$.

                    \item
                        For the second property,
                        by the induction hypothesis,
                        we have
                        \begin{equation*}
                            \transstructp{\psi_1}{\alogicstructure} \models \transgptwoexp{\psi_1}
                            \quad\text{and}\quad
                            \transstructp{\psi_2}{\alogicstructure} \models \transgptwoexp{\psi_2}.
                        \end{equation*}
                        Again by the fact that
                        $\transstructp{\varphi}{\alogicstructure}$
                        is the common expansion of
                        $\transstructp{\psi_1}{\alogicstructure}$ and
                        $\transstructp{\psi_2}{\alogicstructure}$,
                        we have
                        \begin{equation*}
                            \transstructp{\varphi}{\alogicstructure} \models \transgptwoexp{\psi_1}
                            \quad\text{and}\quad
                            \transstructp{\varphi}{\alogicstructure} \models \transgptwoexp{\psi_2}.
                        \end{equation*}
                        Hence
                        $\transstructp{\varphi}{\alogicstructure} \models \transgptwoexp{\psi_1} \land \transgptwoexp{\psi_2}$.

                    \item
                        The last property follows by an argument similar to that for the previous property.
                \end{compactitem}

            \item \itemtitle{Global Presburger quantifier}
                Suppose that $\varphi(x) = P(x) = \left(\sum_{t\in\intsinterval{k}} \lambda_t\cdot\presby{\psi_t(y)} \circledast \delta \right)$.
                Let $H, U_1, \ldots, U_k$ be the fresh unary predicate symbol associated with this occurrence of $P(x)$ in the translation.
                \begin{compactitem}
                    \item
                        By the definition of $\transstructp{P(x)}{A}$,
                        \begin{equation*}
                            v \in H^{\transstructp{P(x)}{A}}
                            \quad\text{if and only if}\quad
                            \alogicstructure, v/x \models P(x).
                        \end{equation*}
                        Recall that $\transgptwop{P(x)} = H(x)$.
                        Thus we have
                        \begin{equation*}
                            \alogicstructure, v/x \models P(x)
                            \quad\text{if and only if}\quad
                            \transstructp{P(x)}{A}, v/x \models \transgptwop{P(x)}.
                        \end{equation*}

                    \item
                        For the second property, we focus on the first conjunct.
                        For every $1 \le t \le k$ and $v \in A$,
                        by the definition of $\transstructp{P(x)}{A}$,
                        \begin{equation*}
                            v \in U_t^{\transstructp{P(x)}{A}}
                            \quad\text{if and only if}\quad
                            \alogicstructure, v/x \models \psi_t(x).
                        \end{equation*}
                        By the induction hypothesis,
                        \begin{equation*}
                            \alogicstructure, v/x \models \psi_t(x).
                            \quad\text{if and only if}\quad
                            \transstructp{\psi_t(x)}{A}, v/x \models \transgptwop{\psi_t(x)}.
                        \end{equation*}
                        Since $\transstructp{P(x)}{A}$ is the expansion of $\transstructp{\psi_t(x)}{A}$,
                        we have
                        \begin{equation*}
                            \transstructp{\psi_t(x)}{A}, v/x \models \transgptwop{\psi_t(x)}.
                            \quad\text{if and only if}\quad
                            \transstructp{P(x)}{A}, v/x \models \transgptwop{\psi_t(x)}.
                        \end{equation*}
                        Therefore
                        \begin{equation*}
                            \transstructp{P(x)}{A}, v/x \models U_t(x) \leftrightarrow \transgptwop{\psi_t(x)}.
                        \end{equation*}

                    \item
                        For every $1 \le t \le k$,
                        for every $u \in A$,
                        by the definition of $\transstructp{P(x)}{\alogicstructure}$,
                        \begin{equation*}
                            u \in U_t^{\transstructp{P(x)}{\alogicstructure}}
                            \quad\text{if and only if}\quad
                            \alogicstructure, v/x \models \psi_t(x),
                        \end{equation*}
                        which implies that
                        \begin{equation*}
                            \sum_{t\in\intsinterval{k}} \lambda_t \cdot \abs{\setc{u \in A}{\transstructp{P(x)}{\alogicstructure}, u/x \models U_t(x)}}
                            \ =\ 
                            \sum_{t\in\intsinterval{k}} \lambda_t \cdot \abs{\setc{u \in A}{\alogicstructure, u/y \models \psi_t(x)}}.
                        \end{equation*}
                        Recall that we defined the expanded structure so that the predicate $T$ held of every element, thus the count of $T$ is the same as the number of elements in th domain. Therefore for every $v \in A$,
                        \begin{equation*}
                            \transstructp{P(x)}{\alogicstructure} \models P^\star
                            \quad\text{if and only if}\quad
                            \alogicstructure, v/x \models P(x).
                        \end{equation*}

                        Recall that by the definition of $\transstructp{P(x)}{A}$,
                        \begin{equation*}
                            v \in H^{\transstructp{P(x)}{A}}
                            \quad\text{if and only if}\quad
                            \alogicstructure, v/x \models P(x).
                        \end{equation*}
                        It follows that
                        if $\transstructp{\varphi}{\alogicstructure} \models P^\star$, 
                        then $H^{\transstructp{\varphi}{\alogicstructure}} = A$;
                        otherwise, $H^{\transstructp{\varphi}{\alogicstructure}} = \emptyset$.
                        Hence, $\transstructp{P(x)}{\alogicstructure}$ satisfies the first two conjuncts of $\transgccp{P(x)}$.

                \end{compactitem}
        \end{compactitem}
    \end{claimproof}

    \begin{claim}\label{claim:res_struct}
        For every $\ourlogic$ formula $\varphi$,
        for every $\vocabextp{\varphi}$-structure $\alogicstructure$, 
        let $\cB$ be the $\vocab$-structure obtained by restricting $\alogicstructure$ to the vocabulary $\vocab$.
        It holds that
        if $\alogicstructure \models \transgptwoexp{\varphi} \land \transgccp{\varphi}$, then for every $v, u \in A$ with $v \neq u$, it holds that
        \begin{equation*}
            \alogicstructure, v/x, u/y \models \transgptwop{\varphi}
            \quad\text{if and only if}\quad
            \cB, v/x, u/y \models \varphi.
        \end{equation*}

        The analogous property holds for formulas with zero or one free variable.
    \end{claim}

    \begin{claimproof}
        We prove the claim by structural induction on $\varphi$.
        The base case of atomic formulas and the inductive cases for Boolean connectives, ordinary quantification, guarded quantification, and local Presburger quantification are straightforward.
        We therefore consider only the case of global Presburger quantification, where $\varphi = P(x) = \left(\sum_{t\in\intsinterval{k}} \lambda_t\cdot\presby{\psi_t(y)} \circledast \delta \right)$.
        Let $H, U_1, \ldots, U_k$ be the fresh unary predicate symbols associated with this occurrence of $P(x)$ in the translation.

        Firstly,
        since $\alogicstructure \models \transgptwoexp{P(x)}$,
        for $1 \le t \le k$ and $v \in A$,
        it holds that
        \begin{equation*}
            \alogicstructure, v/x \models U_t(x)
            \quad\text{if and only if}\quad
            \alogicstructure, v/x \models \transgptwop{\psi_t(x)}.
        \end{equation*}
        Moreover, $\alogicstructure \models \transgptwoexp{P(x)} \land \transgccp{P(x)}$
        ensures that the conditions required to apply the induction hypothesis to each $\psi_t$,
        $\alogicstructure \models \transgptwoexp{\psi_t(x)} \land \transgccp{\psi_t(x)}$, are satisfied.
        Hence, for every $1 \le t \le k$ and $v \in A$,
        \begin{equation*}
            \alogicstructure, v/x \models U_t(x)
            \quad\text{if and only if}\quad
            \cB, v/x \models \psi_t(x).
        \end{equation*}

        Consequently, the corresponding counting terms are equal
        \begin{equation*}
            \sum_{t\in\intsinterval{k}} \lambda_t \cdot \abs{\setc{u \in A}{\alogicstructure, u/x \models U_t(x)}}
            \ =\ 
            \sum_{t\in\intsinterval{k}} \lambda_t \cdot \abs{\setc{u \in A}{\cB, u/y \models \psi_t(x)}}.
        \end{equation*}
        Since $P^\star$ and $P(x)$ use the same comparison operator $\circledast$ and constant $\delta$, it follows that, for every $v \in A$,
        \begin{equation}
            \alogicstructure \models P^\star
            \quad\text{if and only if}\quad
            \cB, v/x \models P(x).
        \end{equation}

        Finally,
        since $\alogicstructure \models \transgccp{P(x)}$,
        the constraints ensure that 
        if $\alogicstructure \models P^\star$, then $H^\alogicstructure = A$.
        Otherwise if $\alogicstructure \not\models P^\star$, then $H^\alogicstructure = \emptyset$.
        Combining this with the equivalence above, we obtain, for every $v \in A$,
        \begin{equation}
            \alogicstructure, v/x \models H(x)
            \quad\text{if and only if}\quad
            \cB, v/x \models P(x).
        \end{equation}
    \end{claimproof}

    To conclude this step, Claims~\ref{claim:ext_struct} and~\ref{claim:res_struct} imply that,
    for every $\ourlogic$ formula $\varphi$ and every $k \in \bbN$,
    $\varphi$ has a model with $k$ elements if and only if $\transgptwop{\varphi} \land \transgptwoexp{\varphi} \land \transgccp{\varphi}$ has a model with $k$ elements.

    \myparagraph{Step 2: Convert local Presburger quantifiers to use only the comparison operator $\ge$}

    In Step~1, we showed that every $\ourlogic$ formula $\varphi$ can be translated into an equisatisfiable $\gptwo \land \gcc$ formula while preserving the cardinalities of models.
    It remains to normalize the $\gptwo$ part. The first issue is to ensure that we only use ``lower bound Presburger quantifiers”, as required by our normal form.

    We use the following equivalences between local Presburger quantifiers.
    \begin{equation*}
        \begin{aligned}
            \sum_{t \in \intsinterval{k}} \lambda_t \cdot \presby{-} = \delta
            \ \equiv\ &
            \left(\sum_{t \in \intsinterval{k}} \lambda_t \cdot \presby{-} \ge \delta\right)
            \land
            \left(\sum_{t \in \intsinterval{k}} \lambda_t \cdot \presby{-} \le \delta\right) \\
            \sum_{t \in \intsinterval{k}} \lambda_t \cdot \presby{-} \neq \delta
            \ \equiv\ &
            \left(\sum_{t \in \intsinterval{k}} -\lambda_t \cdot \presby{-} \ge -(\delta - 1)\right) \lor
            \left(\sum_{t \in \intsinterval{k}}  \lambda_t \cdot \presby{-} \ge \delta + 1\right) \\
            \sum_{t \in \intsinterval{k}} \lambda_t \cdot \presby{-} > \delta
            \ \equiv\ &
            \sum_{t \in \intsinterval{k}} \lambda_t \cdot \presby{-} \ge \delta + 1 \\
            \sum_{t \in \intsinterval{k}} \lambda_t \cdot \presby{-} < \delta
            \ \equiv\ &
            \sum_{t \in \intsinterval{k}} -\lambda_t \cdot \presby{-} \ge -(\delta - 1) \\
            \sum_{t \in \intsinterval{k}} \lambda_t \cdot \presby{-} \le \delta
            \ \equiv\ &
            \sum_{t \in \intsinterval{k}} -\lambda_t \cdot \presby{-} \ge -\delta
        \end{aligned}
    \end{equation*}
    Therefore, we can further transform the $\gptwo$ formula so that all local Presburger quantifiers use only the comparison operator $\ge$.

    \myparagraph{Step 3: Normalize the $\gptwo$ component}

    Finally, we normalize the $\gptwo$ component using the procedure of \cite[Lemma 3.3]{twovarpres}.
    In that result, there is no explicit statement about the use of $\leq$ or $\geq$ in Presburger quantifiers. But
    inspection of the procedure shows that it preserves the comparison operators in local Presburger quantifiers. Thus the resulting formula continues to use only $\ge$.
\end{proof}

%% file: app-opspreserve.tex

\section{Proof of Lemma \ref{lem:operation}: operations preserve realized $\gbv$s}

\opspreserve*

Let $\cG$ be a D-structure realizing a $\gbv$ $\agbv$ with witness $\rep$.
By definition of realization,
the following property holds.
\begin{compactitem}
    \item
        For every vertex $v$, $v$ realizes $\rep(v)$.
        If $v$ is a concrete vertex, then $\rep(v)$ is a valid $\lbv$;
        if $v$ is a phantom vertex, then $\rep(v)$ is a homogeneously valid $\lbv$.
    \item
        For every 1-type $\pi$, $x \in \set{\basetext, \periodtext}$, and $1 \le i \le \numbhxs{\pi}$,
        \begin{equation*}
            \agbv(\pi, x, i) = \sum_{u \in V_\pi} \rep(u)(x, i),
        \end{equation*}
        where $V_\pi$ is the set of vertices with 1-type $\pi$.
\end{compactitem}

We prove that the resulting D-structure also realizes $\agbv$, by constructing the corresponding witness.
\begin{compactitem}
    \item \itemtitle{Merge}
        Let $\cG'$ be obtained from $\cG$ by merging a phantom vertex $u_1$ into a concrete or phantom vertex $u_2$, producing a new vertex $w$,
        where $\onetypeof(u_1) = \onetypeof(u_2)$ and there is no edge between $u_1$ and $u_2$.

        Define $\rep'$ by
        \begin{equation*}
            \rep'(v)
            \ :=\
            \begin{cases*}
                \rep(v),             & \text{if $v \notin \{u_1,u_2\}$,} \\
                \rep(u_1)+\rep(u_2), & \text{if $v=w$.}
            \end{cases*}
        \end{equation*}
        It is immediate that, for every 1-type $\pi$, the sum of the witnesses of type $\pi$ is preserved.
        Thus, it suffices to show that $w$ realizes $\rep'(w)$.

        Observe that the behavior of $w$ in $\cG'$ is exactly the sum of the behaviors of $u_1$ and $u_2$ in $\cG$.
        Since $u_1$ realizes $\rep(u_1)$ and $u_2$ realizes $\rep(u_2)$, we have
        \begin{equation*}
            \begin{aligned}
                \bh{w}
                \ =\ &\bh{u_1} + \bh{u_2}
                \ =\  \eval{\rep(u_1)} + \eval{\rep(u_2)} 
                \ =\  \eval{\rep(w)}.
            \end{aligned}
        \end{equation*}
        Hence, $w$ realizes $\rep'(w)$.

    \item \itemtitle{Swap}
        Let $\cG'$ be obtained from $\cG$ by swapping a pair of edges satisfying the conditions described in the statement of the lemma.
        It is clear that the behavior of every vertex remains unchanged.
        Hence, the witness showing that $\cG$ realizes $\agbv$ is also a witness showing that $\cG'$ realizes $\agbv$.
\end{compactitem}

%% file: app-almostcsl.tex

\section{Proof of Proposition \ref{prop:conformantsemilinear}: the set of conformant behaviors is hybrid linear} \label{app:almostcsl}
Recall from the body that $\univpartphi(x)$  is the formula obtained from a normal-form sentence $\anfgptwosentence$ by discarding 
the outer quantifier $\forall x$:

\[
\quantifierfreepartphi(x)\ \land\ 
\bigwedge_{i \in \intsinterval{m}} \forall y\ \left(R_i(x,y) \land x \neq y\right) \to \alpha_i(x, y)
\ \land\ 
\bigwedge_{i \in \intsinterval{n}} U_i(x) \to P_i(x)
\]
where $\quantifierfreepartphi$ is a quantifier-free formula, each $R_i(x,y)$ is a binary predicate,
each $\alpha_i(x,y)$ is a quantifier-free formula, each $U_i(x)$ is a unary predicate
and each $P_i(x)$ is a local Presburger quantifier using only the comparison operator $\ge$:
\[
\sum_{t \in \intsinterval{s_i}} \lambda_{i,k} \cdot \presby{(x \neq y) \land R_{i,t}(x, y)}  \ \ge\ \delta_i \] 
with each $R_{i,t}(x,y)$ a binary atom, 

Recall from Definition  \ref{def:conformanttypeasymmetrictype} that a 1-type $\pi$ is conformant (w.r.t. $\anfgptwosentence$), 
if $\pi(x) \wedge \neg R(x,x) \models \quantifierfreepartphi(x)$.
A $2$-type $\eta$ is conformant with 1-types $\pi_1$ and $\pi_2$ (w.r.t. $\anfgptwosentence$),
if $\pi_1(x) \land \eta(x,y) \land \pi_2(y) \models (R_i(x,y)\land x\neq y) \to \alpha_i(x,y)$
and $\pi_2(x) \land \dual{\eta}(x,y) \land \pi_1(y) \models (R_i(x,y)\land x\neq y) \to \alpha_i(x,y)$
for every $i \in \intsinterval{m}$. 

Lastly, recall from Definition \ref{def:conf_behavior} that a (local) behavior is conformant with $\univpartphi$ and $\pi$
if there exists a structure $\cG$ and an element $v\in \cG$  having that behavior, with 1-type $\pi$, such that every one type and non-null two-type is conformant.

We now recall the statement of Proposition \ref{prop:conformantsemilinear}:

\almostcsl*

The proof is quite straightforward. Consider the following linear constraints, a variation of those introduced in prior work (e.g. \cite{uscsl}).

\begin{definition}[Characteristic system]
\label{def:characteristic-system}
    For a $\GPtwo$ sentence $\varphi$ in normal form and conformant 1-type $\pi$,
    the characteristic system of $\varphi$ and $\pi$, denoted $\cC^\varphi_\pi(\bfx)$,
    is the integer linear system defined as follows:
    \begin{compactitem}
        \item For every $\pi'$ a conformant one-type and $\eta \in \twotypespi{\pi}{\pi'}$,
        there is a non-negative variables $x_{\eta, \pi'}$ in $\cC^\varphi_\pi(\bfx)$.
        \item (Counting)
        For $1 \le i \le n$, if $U_i(x) \in \pi$
        then the following constraint is in $\cC^\varphi_\pi(\bfx)$:
        \begin{equation*}
            \left(
                \sum_{t \in \intsinterval{s_i}}
                \lambda_{i, t} \cdot
                \sum_{\pi' \in \onetypes}
                \sum_{\eta \in \twotypespi{\pi}{\pi'}, R_{i,t}(x,y) \in \eta} 
                x_{\eta, \pi'}
            \right)
            \ \ge\ \delta_i.
        \end{equation*}

    \end{compactitem}
\end{definition}

\begin{claim} \label{clm:charsysworks}
For any conformant $1$-type $\pi$, the solutions of the characteristic system are exactly the conformant behaviors.
\end{claim}
\begin{proof}
In one direction, we argue that a conformant behavior satisfies the system.  By definition, this is the behavior of an element $v$ in a structure with only conformant types.
Recall also that we only consider structures that are loopless; thus the $x \neq y$ in a local Presburger quantifier is irrelevant for such structures. Thus a local Presburger constrains
the neighbors of an element so that its behavior satisfies the inequality constraint above. Satisfaction of the other constraints is straightforward.

On the other hand,
suppose we have a solution with $x_{\eta, \pi'} \gets c_{\eta, \pi'}$. We create a tree-shaped structure with a root element $v$, labeled
so as to satisfy $\pi$. 
For each two-type $\eta$ and $1$-type $\pi'$ such that both
$\pi'$ and $\pi, \eta, \pi'$ are conformant, we link $v$ via edges consistent with $\eta$ to
$c_{\eta, \pi'}$ distinct vertices, 
each labeled to satisfy $\pi'$.

We claim that the root node in the resulting structure  witnesses that the behavior is conformant.
That is, we claim that the root node satisfies $\univpartphi$. The fact that it satisfies $\quantifierfreepartphi$ follows from the fact that
$\pi$ is conformant.
We now consider a constraint $(R_i(x,y) \land x \neq y) \to \alpha_i(x, y)$. This is only relevant for edges between the root node and one of the
children created below, and these are all arranged to satisfy a conformant $2$-type, which guarantees that the implication holds.
Finally we consider a constraint of the form $U_i(x) \to P_i(x)$. We only need concern ourselves with $U_i$ contained in $\pi$, since otherwise
the implication holds vacuously. But then the counting constraint applies to guarantee $P_i$.
\end{proof}

To complete the proof of Proposition \ref{prop:conformantsemilinear}, note that the characteristic system contains no disjunctions and can therefore be expressed in the form $Ax \ge c$.
Hence, its solution set is hybrid linear~\cite[Proposition 4]{semilinear}. 

%% file: app-constructDstruct.tex

\section{Proof of the rewiring lemma, Lemma~\ref{lem:rewiring}, and Lemma~\ref{lem:match} (constructing a decomposed structure from a $\gbv$ satisfying $\match$)} \label{app:dstructurefrommatch}


\rewiring*

\begin{proof}
 
    Let $\cG_0$ be $\cG$.
    We define the D-structure $\cG_{t + 1}$ extended $\cG_{t}$ by adding edges according to the following procedure.
    For every vertex $v$, we write $\bhg{t}{v}$ for the behavior of $v$ in $\cG_t$.
    The D-structure $\cG_{t+1}$ is obtained as follows:
    \begin{enumerate}
        \item
            Pick a vertex $v_1 \in V_1$ such that
            $\bhg{t}{v_1}(\eta, \pi_2) < f_1(v_1)$.

        \item
            Pick a vertex $v_2 \in V_2$ such that
            $\bhg{t}{v_2}(\dual{\eta},\pi_1) < f_2(v_2)$.

        \item
            Add an edge from $v_1$ to $v_2$ of 2-type $\eta$.
    \end{enumerate}

    Let $T := \sum_{v \in V_1} f_1(v)$.
    We argue that the following invariants are preserved for each $\cG_t$ for $t \le T$.
    \begin{enumerate}
        \item
            For every vertex $v \in V_1$,
            $\bhg{t}{v}(\eta,\pi_2) \le f_1(v)$.

            For the base case, by assumption, we have
            $\bhg{0}{v}(\eta,\pi_2) = 0 \le f_1(v)$.

            For the induction step, we consider three cases.
            \begin{compactitem}
                \item
                    Suppose that $v$ is picked in the $t^{th}$ iteration.
                    Then $\bhg{t}{v}(\eta,\pi_2) < f_1(v)$.
                    Since we add only one edge from $v$ with 2-type $\eta$ to a vertex with 1-type $\pi_2$,
                    $\bhg{t+1}{v}(\eta,\pi_2) = \bhg{t}{v}(\eta,\pi_2)+1$.
                    Combining these two inequalities, we obtain
                    $\bhg{t+1}{v}(\eta,\pi_2) \le f_1(v)$.

                \item
                    If $v$ is not picked and the added edge is not incident to $v$,
                    then clearly the behavior of $v$ remains unchanged, so the invariant continues to hold.

                \item
            
                    Otherwise, if $v$ is not picked but the added edge is incident to $v$,
                    this can only occur when $\pi_1 = \pi_2$,
                    which implies that $V_1 = V_2$.
                    \emph{Here we use that $\eta$ is asymmetric, 
                    so $\eta$ and $\dual{\eta}$ are distinct}.
                    Thus, the added edge does not affect the $(\eta, \pi_2)$ entry of the behavior of $v$,
                    and hence $\bhg{t+1}{v}(\eta,\pi_2) = \bhg{t}{v}(\eta,\pi_2)$.
            \end{compactitem}

            The analogous bounds hold for every vertex in $v \in V_2$,
            $\bhg{t}{v}(\dual{\eta},\pi_1) \le f_2(v)$.

        \item
            We have
            \begin{equation*}
                \sum_{v \in V_1} \bhg{t}{v}(\eta,\pi_2)
                \ =\ 
                \sum_{v \in V_2} \bhg{t}{v}(\dual{\eta},\pi_1)
                \ =\ 
                t.
            \end{equation*}

            We prove this by induction on $t$.
            The base case is immediate, since all the relevant behavior values are zero by assumption.

            For the induction step, suppose that an edge from $v_1$ to $v_2$ of
            2-type $\eta$ is added. Then
            $\bhg{t+1}{v_1}(\eta,\pi_2) = \bhg{t}{v_1}(\eta,\pi_2) + 1$ and
            $\bhg{t+1}{v_2}(\dual{\eta},\pi_1) = \bhg{t}{v_2}(\dual{\eta},\pi_1) + 1$,
            while all other relevant behavior values remain unchanged.
            Thus, both sides of the equality increase by $1$, and the equality is preserved.
    \end{enumerate}

    Note that the second invariant implies that suitable vertices $v_1$ and $v_2$ can always be chosen at every iteration $t < T$.
    Otherwise, we would have
    \begin{equation*}
        \sum_{v \in V_1} \bhg{t}{v}(\eta,\pi_2)
        \ =\ 
        \sum_{v \in V_1} f_1(v)
        \ =\ 
        T,
    \end{equation*}
    which contradicts the second invariant.

    Therefore, after $T$ iterations, every vertex $v \in V_1$ satisfies
    $\bh{v}_T(\eta,\pi_2) = f_1(v)$,
    and every $v \in V_2$ satisfies
    $\bh{v}_T(\dual{\eta},\pi_1) = f_2(v)$.

    Moreover, the procedure affects only the
    $(\eta,\pi_2)$-entries of the behaviors of vertices in $V_1$ and the
    $(\dual{\eta},\pi_1)$-entries of the behaviors of vertices in $V_2$.
    Therefore, $\cG_T$ is the desired D-structure $\cH$.

    Note that if $\cG$ is an asymmetric D-structure,
    since the construction only add edges with asymmetric 2-type $\eta$,
    the resulting D-structure $\cH$ is also asymmetric.
\end{proof}

\constructDstruct*

\begin{proof}
    Let $\agbv$ be a $\gbv$ that satisfies $\match$.
    That is, for every 1-type $\pi_1$, $\pi_2$ and 2-type $\eta$,

    \begin{equation*}
        \begin{aligned}
            &\sum\limits_{i \in \intsinterval{\numbhbases{\pi_1}}}   
            \agbv(\pi_1, \basetext, i)   \cdot \bhbase{\pi_1}{i}(\eta, \pi_2) +
            \sum\limits_{i \in \intsinterval{\numbhperiods{\pi_1}}} 
            \agbv(\pi_1, \periodtext, i) \cdot \bhperiod{\pi_1}{i}(\eta, \pi_2) \\
            \ =\ 
            &\sum\limits_{i \in \intsinterval{\numbhbases{\pi_2}}}   
            \agbv(\pi_2, \basetext, i)   \cdot \bhbase{\pi_2}{i}(\dual{\eta}, \pi_1) +
            \sum\limits_{i \in \intsinterval{\numbhperiods{\pi_2}}} 
            \agbv(\pi_2, \periodtext, i) \cdot \bhperiod{\pi_2}{i}(\dual{\eta}, \pi_1).
        \end{aligned}
    \end{equation*}
    Our goal is to construct a D-structure $\cG$ such that
    for every 1-type $\pi$,
    \begin{itemize}
        \item
            for every $1 \le i \le \numbhbases{\pi}$,
            there are exactly $\gbv(\pi, \basetext, i)$ concrete vertices with 1-type $\pi$ and behavior $\bhbase{\pi}{i}$;
        \item
            for every $1 \le i \le \numbhperiods{\pi}$,
            there are exactly $\gbv(\pi, \periodtext, i)$ phantom vertices with 1-type $\pi$ and behavior $\bhperiod{\pi}{i}$.
    \end{itemize}
    It is straightforward that such a D-structure realizes $\agbv$.
    Indeed, we can construct a realization witness $\rep$ as follows.
    For every concrete vertex $v$ with 1-type $\pi$ and behavior $\bhbase{\pi}{i}$, set
    $\rep(v)(\basetext, i) = 1$, and set all other entries to $0$.
    For every phantom vertex $v$ with 1-type $\pi$ and behavior $\bhperiod{\pi}{i}$, set
    $\rep(v)(\periodtext, i) = 1$, and set all other entries to $0$.

    We now construct a D-structure satisfying the required properties. We start with the following D-structure $\cG_0$.
    For every 1-type $\pi$,
    \begin{itemize}
        \item
            for every $1 \le i \le \numbhbases{\pi}$,
            there are $\gbv(\pi, \basetext, i)$ concrete vertices with 1-type $\pi$,
            called $(\pi, i)$ concrete vertices. We define the \emph{intended behavior} of those vertices to be $\bhbase{\pi}{i}$;
        \item
            for every $1 \le i \le \numbhperiods{\pi}$,
            there are $\gbv(\pi, \periodtext, i)$ phantom vertices with 1-type $\pi$,
            called $(\pi, i)$ phantom vertices. We define the \emph{intended behavior} of those vertices to be $\bhperiod{\pi}{i}$.
    \end{itemize}
    There are no edges in $\cG_0$.

    We fix an enumeration of all triples $(\pi_1,\pi_2,\eta)$,
    where $\pi_1, \pi_2$ are 1-types and $\eta$ is an asymmetric 2-type,
    and let $T$ denote the length of this enumeration. We filter the enumeration so that it does not contain $(\pi_1, \pi_2, \eta)$ and
    $(\pi_2, \pi_1, \dual{\eta})$.
    We construct a sequence of D-structures $\cG_0, \cG_1, \ldots, \cG_T$ as follows.

    For each $1 \le t \le T$,
    let $(\pi_1, \pi_2, \eta)$ be the $t^{th}$ triple in the enumeration.
    We let $f_1$ be a function from vertices with 1-type $\pi_1$ to non-negative integers, and 
    $f_2$ be a function from vertices with 1-type $\pi_2$ to non-negative integers
    defined as follows.
    \begin{equation*}
        \begin{aligned}
            f_1(v)
            \ =\ &
            \begin{cases}
                \bhbase{\pi_1}{i}(\eta, \pi_2), \text{$v$ is a $(\pi_1, i)$ concrete vertex} \\
                \bhperiod{\pi_1}{i}(\eta, \pi_2), \text{$v$ is a $(\pi_1, i)$ phantom vertex}
            \end{cases} \\
            f_2(v)
            \ =\ &
            \begin{cases}
                \bhbase{\pi_2}{i}(\dual{\eta}, \pi_1), \text{$v$ is a $(\pi_2, i)$ concrete vertex} \\
                \bhperiod{\pi_2}{i}(\dual{\eta}, \pi_1), \text{$v$ is a $(\pi_2, i)$ phantom vertex}.
            \end{cases}
        \end{aligned}
    \end{equation*}
    By condition $\match$, the sum of the $(\eta, \pi_2)$-entries of intended behaviors of the $(\pi_1, i)$ concrete and phantom vertices
    equals
    the sum of the intended $(\dual{\eta}, \pi_1)$-behavior entries of the $(\pi_2, i)$ concrete and phantom vertices,
    which implies that the sum of $f_1(v)$ over vertices with 1-type $\pi_1$
    is equal to
    the sum of $f_2(v)$ over vertices with 1-type $\pi_2$.
    We obtain $\cG_t$ from $\cG_{t-1}$ by applying the rewiring lemma, Lemma~\ref{lem:rewiring}.

    We write $\bhg{t}{v}$ for the behavior of a vertex $v$ in $\cG_t$.
    By the conclusion of Lemma~\ref{lem:rewiring}, the following invariant holds for every $1 \le t \le T$:

    \medskip 
    
    for every triple $(\pi_1, \pi_2, \eta)$ occurring among the first $t$ triples in the enumeration,
    every $(\pi_1, i)$ concrete vertex $v$ satisfies $\bhg{t}{v}(\eta, \pi_2) = \bhbase{\pi_1}{i}(\eta, \pi_2)$,
    and every $(\pi_1, j)$ phantom vertex $v$ satisfies $\bhg{t}{v}(\eta, \pi_2) = \bhperiod{\pi_1}{j}(\eta, \pi_2)$.

    \medskip
    
    Notice that each application of the lemma only adds edges of a specific $2$-type and specific $1$-types: all other triples are unaffected.

    Since the enumeration contains all triples of 1-types and 2-types, every vertex has its intended behavior in $\cG_T$.
    Hence, $\cG_T$ is the desired D-structure
\end{proof}

%% file: app-outgoingonly.tex

\section{Outgoing-only aggregation} \label{app:outgoing-only}

We recall from the body that we stated that our decision procedure
simplifies dramatically for a ``modal'' fragment. Indeed, we do not require the machinery of D-structures which is the main technical innovation in this paper. 
Instead, we can make use of a ``homogeneous model property'' that has been established in other contexts \cite{daumantas,nexptimemodalglobalcounting}. We now formalize this.
\begin{definition}[Outgoing-only full logic ($\optwo$), unmixed version $\ounmixedgptwo$]
We let $\optwo$ denote the variant of $\ourlogic$ 
formed by closing unary atoms 
under the usual connectives, with the use of quantifiers restricted as follows.
\begin{compactitem}
\item 
If $\varphi(x)$ is a formula,
then so are $\forall x\ \varphi(x)$ and $\exists x\ \varphi(x)$.
\item 
If $\varphi(y)$ is a formula,
then so are $\forall y\ R(x, y) \to \varphi(y)$ and $\exists y\ R(x, y) \land \varphi(y)$,
for some binary relation $R$.

\item\itemtitle{Mixing local and global Presburger formulas}
If $\set{\varphi_t(x, y)}_{t \in \intsinterval{k}}$ is a set of formulas,
$\set{R_t(x, y)}_{t \in \intsinterval{k}}$ is a set of atoms. $\set{\lambda_t}_{t \in \intsinterval{k}}$ and $\set{\lambda'_t}_{t \in \intsinterval{l}}$
are integers, and
$\set{\varphi'_t(y)}_{t \in \intsinterval{\ell}}$ is another set of formulas,
then 
    $\sum\limits_{t \in \intsinterval{k}} \lambda_t \cdot \presby{R_t(x,y)\land\varphi_t(x, y)}
    + \sum\limits_{t\in \intsinterval{\ell}} \lambda'_t \cdot \presby{\varphi'_t(y)}
    \ \circledast\ \delta$
is also a formula.
\end{compactitem}
We let $\ounmixedgptwo$ be the variant where we do not allow the mixing of local and global counts.
\end{definition}

Note that the standard modal formula $\forall y\ R(x, y) \to \varphi(y)$ and $\exists y\ R(x, y) \land \varphi(y)$
can be expressed in terms of local Presburger quantifiers:
$\presby{R(x,y)\land\neg \varphi(y)}=0$ and $\presby{R(x,y)\land\varphi(y)}\ge 1$.

We can now formalize our main claim:
\begin{theorem}\label{thm:outgoing_main}
The satisfiability problem for $\optwo$ and $\ounmixedgptwo$ is complete for $\nexp$. 
In addition, the spectrum of such sentences is semilinear.
\end{theorem}

To prove the upper bound and the semilinearity claim, we start with a normalization result. Unlike Lemma \ref{lem:normal}, this is a straightforward
Scott normal form normalization, with no additional restrictions put on the Presburger quantifiers.

\begin{lemma}[Normal form]\label{lem:outgoingnormal}
From an $\optwo$ sentence $\aoptwosentence$ we can create, in exponential time, a normal form sentence, which is 
a disjunction of \emph{normal form disjunct} sentences over an extended vocabulary where 
\begin{compactitem}
\item each normal form disjunct has the form:
    \[
        \anfoptwosentence
        \ :=\ 
        \forall x~ \quantifierfreepartphi(x)
        \ \land\ 
        \bigwedge_{i \in \intsinterval{m}} \forall x \forall y\ \left(R_i(x,y) \land x \neq y\right) \to \alpha_i(x, y)
        \ \land\ 
        \bigwedge_{i \in \intsinterval{n}} \forall x\ U_i(x) \to P_i(x)
    \]
    Here each $P_i(x)$ is a Presburger quantifier in the ``modal form'':
    \[
        \sum\limits_{t \in \intsinterval{k}} \lambda_t \cdot \presby{(x \neq y) \land R_t(x,y)}
        + \sum\limits_{t' \in \intsinterval{\ell}} \lambda_t' \cdot \presby{U_t(y)}
        \ \circledast\ \delta
    \]
    each $U_t$ is a unary relation, $R_t(x,y)$ is a binary relation,
    and $\quantifierfreepartphi$ and each $\alpha_i$ is quantifier-free.
    We also require that each $\alpha_i(x, y)$ does not contain atoms of the form $R(y, x)$.

    \item $\aoptwosentence$ has a model of size $k$ if and only if $\anfoptwosentence$ has a model of size $k$.

    \item  the size of each disjunct is polynomial in the size of $\aoptwosentence$. 
\end{compactitem}
\end{lemma}
\begin{proof}
    The normalization procedure presented here is based on the
    the corresponding procedure for $\gptwo$ in \cite[Lemma 3.3]{twovarpres}, which in turn
    follows the usual normalization  procedure for two-variable logic.
   Each $k$-ary subformula of $\aoptwosentence$
    is associated with a fresh $k$-ary predicate,
    which encodes the truth value of the corresponding subformula.

    We describe the key new step, for  mixed local and global Presburger quantifiers.
    Suppose that
    \[
        P(x) :=
        \sum\limits_{t \in \intsinterval{k}}
        \lambda_t \cdot \presby{R_t(x,y)\land\varphi_t(x,y)}
        +
        \sum\limits_{t\in \intsinterval{\ell}}
        \lambda'_t \cdot \presby{\varphi'_t(y)}
        \ \circledast\ \delta.
    \]
    Since we consider only self-loop-free models,
    $R_t(x,x)$ never holds for any $t \in \intsinterval{k}$.
    Hence, $P(x)$ is equivalent to
    \[
        \sum\limits_{t \in \intsinterval{k}}
        \lambda_t \cdot
        \presby{(x \neq y) \land R_t(x,y) \land \varphi_t(x,y)}
        +
        \sum\limits_{t\in \intsinterval{\ell}}
        \lambda'_t \cdot \presby{\varphi'_t(y)}
        \ \circledast\ \delta.
    \]

    We introduce two fresh unary predicates $U, \tilde{U}$ associated with $P(x)$,
    together with fresh binary predicates $R'_1,\ldots,R'_k$
    and fresh unary predicates $U'_1,\ldots,U'_\ell$.
    We then replace every occurrence of $P(x)$ in $\aoptwosentence$
    with $U(x)$ and conjoin the resulting formula with
    the following constraints:
    \begin{equation*}
        \begin{aligned}
            & \forall x\ 
            \left(U(x) \land \neg \tilde{U}(x)\right) \lor
            \left(\neg U(x) \land \tilde{U}(x)\right) \\
            &\forall x\,
            \left(
                U(x) \to
                \left(
                    \sum\limits_{t \in \intsinterval{k}}
                    \lambda_t \cdot
                    \presby{(x \neq y) \land R'_t(x,y)}
                    +
                    \sum\limits_{t\in \intsinterval{\ell}}
                    \lambda'_t \cdot \presby{U'_t(y)}
                    \ \circledast\ \delta
                \right)
            \right),\\
            &\forall x\,
            \left(
                \tilde{U}(x) \to
                \left(
                    \sum\limits_{t \in \intsinterval{k}}
                    \lambda_t \cdot
                    \presby{(x \neq y) \land R'_t(x,y)}
                    +
                    \sum\limits_{t\in \intsinterval{\ell}}
                    \lambda'_t \cdot \presby{U'_t(y)}
                    \ \tilde{\circledast}\ \delta
                \right)
            \right),\\
            &\forall x\forall y\,
                R'_t(x,y) \to
                (R_t(x,y) \land \varphi_t(x, y))
            \qquad \text{for each } t \in \intsinterval{k},\\
            &\forall x\forall y\ R_t(x, y) \to [\varphi_t(x, y) \to R'_t(x, y)]
            \qquad \text{for each } t \in \intsinterval{k},\\
            &\forall x\,
            \left(
                U'_t(x) \leftrightarrow \varphi'_t(x)
            \right)
            \qquad \text{for each } t \in \intsinterval{\ell},
        \end{aligned}
    \end{equation*}
    where $\tilde{\circledast}$ is defined as follows:
    if $\circledast$ is $=$, then $\tilde{\circledast}$ is $\neq$;
    if $\circledast$ is $\neq$, then $\tilde{\circledast}$ is $=$;
    if $\circledast$ is $\ge$, then $\tilde{\circledast}$ is $<$.

    Note that the conditions guarantee that $U$ and $\tilde{U}$ partition the domain, and thus we have implications in both directions.

    The remainder of the inductive steps are as in \cite[Lemma 3.3]{twovarpres}. These steps will be integrated into a bottom-up inductive procedure. We omit the details, but observe that the normalization procedure
    in \cite[Lemma 3.3]{twovarpres}
    does not introduce binary atoms which reverse the argument order from the original atoms.
    In particular, the fresh binary predicates introduced above
    preserve the direction of the original binary relations.
    Consequently, the normalization procedure preserves
    the outgoing-only restriction.

The inductive invariant will be that every structure and valuation satisfying the original formula expands to a model of the revised formula, and every
structure and valuation satisfying the revised formula restricts to a model of the original formula.

\end{proof}

We now show that the satisfiability of sentences
in normal form as defined in Lemma~\ref{lem:outgoingnormal} is $\nexp$-complete.
We fix $\anfoptwosentence$ a normal form disjunct.
We start with the main observation that simplifies the outgoing-only case. 
We can always restrict our attention to models where the $1$-type determines the behavior, 
and thus we can describe a model with a much simpler set of linear constraints.

\begin{definition}[Outgoing $2$-type, Outgoing behavior, $\twotypescmpi{\pi}{\pi'}$ ] \label{def:outgoingtwotypebehavior}
A 2-type $\eta$ is \emph{outgoing} if it consists only of binary predicates $R(x,y)$ and $\neg R(x,y)$, and contains at least one positive atom.
We define the \emph{outgoing behavior} of an element as the behavior whose 
indices $(\eta,\pi)$ are restricted to $\eta$ being outgoing 2-types.

We say that an outgoing 2-type $\eta$ is \emph{conformant with 1-types $\pi_1$ and $\pi_2$} (w.r.t. $\anfoptwosentence$),
if for every $i \in \intsinterval{m}$,
\begin{equation*}
    \begin{aligned}
        \pi_1(x) \land \eta(x,y) \land \pi_2(y) &\models (R_i(x,y)\land x\neq y) \to \alpha_i(x,y).
    \end{aligned}
\end{equation*}

Following Definition \ref{def:oneandtwotype},
we denote by $\onetypes$ the $1$-types that are conformant with the input sentence $\anfoptwosentence$,
We now let $\twotypescmpi{\pi}{\pi'}$  denote the outgoing $2$-types that are conformant with $\pi$, $\pi'$, and $\anfoptwosentence$.

\end{definition}

\begin{definition}[Behavior homogeneity]
We say a structure $\cG$ is \emph{behavior homogeneous} if  for all $v,v'\in \cG$:
\[
\onetypeof_\cG(v)=\onetypeof_\cG(v')\implies \obh_\cG(v)=\obh_\cG(v')
\]    
\end{definition}

We will show that if there is a satisfying structure, there is a behavior homogeneous one. We first give a sufficient condition for a transformation to preserve
normal form disjuncts.

\begin{claim} \label{clm:behavtransfer}       
    Let $\cG$ and $\cG'$ be structures.
    Suppose that there exists a bijection $f$ from $V$ to $V'$ which preserves 1-types of vertices.
    Suppose also that for every $v' \in G'$, there exists $v \in G$ such that $v$ and $v'$ have the same 1-type and the same outgoing behavior,
    Then for every normal form disjunct $\anfoptwosentence$ 
    if $\cG \models \anfoptwosentence$,
    then $\cG' \models \anfoptwosentence$.
\end{claim}
\begin{proof}
    Suppose that 
    \[
        \anfoptwosentence
        \ :=\ 
        \forall x~ \quantifierfreepartphi(x)
        \ \land\ 
        \bigwedge_{i \in \intsinterval{m}} \forall x \forall y\ \left(R_i(x,y) \land x \neq y\right) \to \alpha_i(x, y)
        \ \land\ 
        \bigwedge_{i \in \intsinterval{n}} \forall x\ U_i(x) \to P_i(x).
    \]
    We show that $\cG' \models \anfoptwosentence$ by checking every component in $\anfoptwosentence$.
    \begin{itemize}
        \item 
            For every $v \in G'$,
            there exists $u \in G$,
            such that $v$ and $u$ have the same 1-type $\pi$.
            Since $\cG \models \anfoptwosentence$,
            it holds that 
            $\cG \models \quantifierfreepartphi(u)$,
            which implies that
            $\pi(x) \models \quantifierfreepartphi(x)$.
            Therefore,
            $\cG' \models \quantifierfreepartphi(v)$.

        \item
            For every edge $e$ with outgoing 2-type $\eta$ from $v_1$ with 1-type $\pi_1$ to $v_2$ with 1-type $\pi_2$ in $\cG'$,
            there exists $u_1$ with 1-type $\pi_1$ in $\cG$ such that $v_1$ and $u_1$ have the same outgoing behavior.
            Therefore, there exists an edge $e'$ with outgoing 2-type $\eta$ from $u_1$ to some $u_2$ with 1-type $\pi_2$ in $\cG$.
            Since $\cG \models \anfoptwosentence$,
            this implies that $\cG \models (R(u_1, u_2) \land (u_1 \neq u_2)) \to \alpha_i(u_1, u_2)$.
            Therefore
            $\cG' \models (R(v_1, v_2) \land (v_1 \neq v_2)) \to \alpha_i(v_1, v_2)$.

        \item
            We first observe that
            for every 1-type $\pi$,
            since $f$ is a bijection that preserves 1-types,
            the number of elements in $\cG$ with 1-type $\pi$ is the same as the number in $\cG'$.

            Next, note that for every vertex $v$ in $\cG'$, the evaluation of $P_i(v)$ depends only on the outgoing behavior of $v$ and the count of elements with specific 1-types in $\cG'$. A similar property holds for $u$ in $\cG$.
            By assumption, there exists an element $u$ in $\cG$ such that $v$ and $u$ have the same outgoing behavior.
            Since $\cG \models \anfoptwosentence$,
            if $U_i \in \pi$, then $P_i(u)$ holds in $\cG$,
            which implies that $P_i(v)$ holds in $\cG'$.
    \end{itemize}


\end{proof}

The following key lemma gives the crucial property of outgoing-only languages: 
for the purposes of satisfiability, we can always deal with a behavior homogeneous structure. 
Again we emphasize that a similar homogeneity property was observed  in earlier work, either concerning a different logic \cite{daumantas}, or for a different purpose \cite[Lemma 18]{daumantasgnn}. And similar bounds were established -- again for a  different logic -- in \cite{nexptimemodalglobalcounting}.

\begin{lemma}\label{homogeneous_structure_existence}
If there is a structure $\cG$ that satisfies a normal form disjunct $\varphi$, then
there exists a behavior homogeneous structure $\cG'$ with the same cardinality which also satisfies $\varphi$.
\end{lemma}
\begin{proof}
We construct $\cG'$ as follows. The set of elements is the same as $\cG$. 
For any $1$-type $\pi$, let $V_\pi$ be the set of elements with that 1-type in $\cG$. 
If $V_\pi$ is empty, then there is nothing to do. 
Otherwise, take an arbitrary $v^*\in V_\pi$, and for every other $v\in V_\pi$, 
ensure that the outgoing relations are exactly those of $v^*$ in $\cG$. 
In the case that there is an $E$ relation from $v^*$ to $v$ in $\cG$, 
we do \emph{not} add a self-loop $E(v,v)$, but add $E(v,v^*)$ instead.

Clearly the resulting structure has the same elements, is behavior homogeneous, 
and does not introduce new outgoing behaviors; 
in particular, note that making a change to the outgoing relations on one element 
cannot impact the outgoing behavior vector of another.
The identity is a bijection satisfying the hypothesis of Claim \ref{clm:behavtransfer}, so the conclusion follows.
\end{proof}

We are now ready to reduce satisfiability to solving a linear system.

\begin{definition}[Linear system for an $\optwo$ normal form disjunct]\label{def:global_cardinalities_system}
For any $\optwo$ normal form disjunct  $\varphi$, for any nonempty set $S$ of conformant $1$-types, we define the integer linear system $\outonlyintegersystemsat^\varphi_S$ as follows:
\begin{itemize}
\item 
There is a non-negative integer variable $x_\pi$ for every $\pi \in S$.
Informally, this stands for the number of realizers of the $1$-type in $S$. 
We denote by $\bld x$ the vector of these variables.

\item 
There is a non-negative integer variable $y_{\pi,\eta,\pi'}$ for every $\pi,\pi'\in S$ and outgoing $\eta\in\twotypescmpi{\pi}{\pi'}$. 
Intuitively, this stands for the $\eta$-degree of a $\pi$ element for edges linking to $\pi'$ elements.
We denote by $\bld y$ the vector of these variables.

\item (Presburger constraint)
For every 1-type $\pi \in S$,
for $1\le i\le n$,
if $\pi(x) \models U_i(x)$,
we include the following constraint, 
corresponding to a conjunct $U_i(x)\to P_i(x)$ where the Presburger quantifier $P_i(x)$ is in the normal form  
$(\sum\limits_t \lambda_{i,t}\cdot \presby{U_{i,t}(y)}) +  
(\sum\limits_t \lambda'_{i,t} \cdot \presby{R_{i,t}(x,y)\land (x \neq y)}) 
\circledast_i \delta_i$:
\begin{align*} 
& \sum\limits_{t} 
\left( \lambda_{i,t}\cdot\sum\limits_{\pi' \in S, \pi'(x)\models U_{i,t}(x)} x_{\pi'}\right)
 + 
\left(\sum\limits_{t} \lambda'_{i,t}\cdot \sum\limits_{\pi' \in S, \eta \in \twotypescmpi{\pi}{\pi'}, \eta(x, y) \models R_{i, t}(x, y)}
y_{\pi, \eta, \pi'}\right) 
\ \circledast_i\  \delta_i
\end{align*}
\item (Capacity) 
For each $\pi_1, \pi_2 \in S$,
if $\pi_1 \neq \pi_2$, then $\sum\limits_{\eta \in \twotypescmpi{\pi_1}{\pi_2}} y_{\pi_1, \eta, \pi_2} \leq x_{\pi_2}$.
To cover the case $\pi_1 = \pi_2$,
 we have
$\sum\limits_{\eta \in \twotypescmpi{\pi_1}{\pi_1}} y_{\pi_1, \eta, \pi_1} \leq x_{\pi_1} - 1$.
\item (Sanity)
 For each $\pi \in S$ $x_{\pi}>0$
\end{itemize} 
\end{definition}

The following result shows that this system captures the cardinalities of finite models, which implies the semilinearity of the spectrum:

\begin{theorem}  \label{thm:systmecapturesolutions}
$\anfoptwosentence$ has a satisfying model of size $k > 0$ exactly when there is some non-empty $S$ such that $\outonlyintegersystemsat^{\anfoptwosentence}_S$ has a 
solution where $\Sigma_{\pi \in S} ~ x_\pi=k$.
\end{theorem}
\begin{proof}
Suppose $\anfoptwosentence$ has a satisfying model $\cG$ of size $ > 0$. Let $S$ be the set of $1$-types realized in $\cG$.
Since $\cG$ is non-empty, $S$ is non-empty.
By Lemma~\ref{homogeneous_structure_existence} and Claim \ref{clm:behavtransfer}, we may assume that $\cG$ is a behavior homogeneous structure.
For each $\pi \in S$, let $V_{\pi}$ be the set of elements with $1$-type $\pi$.
Since each $V_\pi$ is nonempty, by realizability of $\pi$, we pick an arbitrary element $a_{\pi} \in V_{\pi}$.
For each $\pi,\pi'\in S$, $\eta \in \twotypescmpi{\pi}{\pi'}$,
let $V_{\pi,\eta,\pi'}$ be the set $\{b \in V_{\pi'} \mid \twotypeof(a_\pi,b)=\eta \}$ where $\twotypeof(a_\pi,b)$ denotes the outgoing-only two-type of the pair.
Since $\cG$ is behavior homogeneous, 
the cardinality of $V_{\pi,\eta,\pi'}$ does not depend on the choice of $a_{\pi}$.

We claim that $\{x_\pi\gets |V_{\pi}|,y_{\pi,\eta,\pi'}\gets |V_{\pi,\eta,\pi'}|\}$ is a solution of 
$\outonlyintegersystemsat^{\anfoptwosentence}$.
First, each Presburger constraint in $\outonlyintegersystemsat^{\anfoptwosentence}_S$ is satisfied
because $\cG\models \forall x ~ U_i(x)\to P_i(x)$, and thus for each element satisfying a $\pi$ considered in the constraint for $U_i$, $P_i$ must
hold of that element.
Second, to see that the Capacity constraint is satisfied, note that $\bigcup_{\eta}V_{\pi,\eta,\pi'}\subseteq V_{\pi'}$.
Thus in the case $\pi \neq \pi'$ we have the desired bound. When $\pi=\pi'$ we note that there cannot be a self-loop, by definition of a structure, and
thus the upper bound follows for this case.
The Sanity constraint is easily seen to hold by definition.

For the converse direction,
suppose for some non-empty set $S$,
$\outonlyintegersystemsat^{\anfoptwosentence}_S$ has a solution 
$\{x_\pi\gets N_{\pi},y_{\pi,\eta,\pi'}\gets N_{\pi,\eta,\pi'}\}$.
We show how to construct a model $\cG$ that satisfies $\anfoptwosentence$.
For each $\pi\in S$, we have a set $V_{\pi}$ of $N_\pi$ elements
and these sets are pairwise disjoint. By the sanity constraint, all $V_{\pi}$ are nonempty, and we will let their union be the domain of our structure.
We set the 1-type of all the elements in $V_{\pi}$ to be $\pi$.
Then, for each element $a\in V_{\pi}$, we do the following.
\begin{compactitem}
\item
For each $\pi'\in S$ with $\pi \neq \pi'$ and $\eta\in\twotypescmpi{\pi}{\pi'}$,
we choose from the set $V_{\pi'}$ a set $V_{\eta,\pi'}$ of $N_{\pi,\eta,\pi'}$ elements, while in the case $\pi=\pi'$, we choose from the set $V_{\pi'}-\{a\}$; we
do this in such a way
that the sets for different $\pi'$ and $\eta$ are pairwise disjoint.

Such sets exist due to the Capacity constraint. 
\item 
Then for each $b\in V_{\eta,\pi'}$,
we set the 2-type of $(a,b)$ to be $\eta$. The exception for the case $\pi=\pi'$ above guarantees that we are not creating self-loops, as required for a structure.
\end{compactitem}
The 2-types of the remaining pairs $(a,b)$ for $a\neq b$ are set to be the null type.
We claim that the constructed model $\cG$ satisfies $\anfoptwosentence$.
Indeed, $\cG$ satisfies $\forall x~ \quantifierfreepartphi(x)$,
since the system $\outonlyintegersystemsat^{\anfoptwosentence}$ has only variable $x_{\pi}$ where $\pi\in S$.
Moreover, when we add an edge with outgoing 2-type $\eta$ from a vertex with 1-type $\pi_1$ to a vertex with 1-type $\pi_2$, $\eta$ is conformant with $\pi_1$ and $\pi_2$.
Therefore, $\pi_1(x) \land \eta(x, y) \land \pi_2(y) \models \left(R_i(x,y) \land x \neq y\right) \to \alpha_i(x, y)$ for every $1 \le i \le m$.
Finally, $\cG$ satisfies $\forall x~U_i(x)\to P_i(x)$ because
$\{x_\pi\gets N_{\pi},y_{\pi,\eta,\pi'}\gets N_{\pi,\eta,\pi'}\}$ satisfies
the Presburger constraint in the system $\outonlyintegersystemsat^{\anfoptwosentence}$.
This completes the proof.
\end{proof}

 Theorem \ref{thm:systmecapturesolutions} immediately implies the desired upper bound:
\begin{theorem}\label{thm:outgoing_upper}
    The satisfiability problem for $\optwo$ is in $\nexp$.
\end{theorem}

\begin{proof}
    We can guess the disjunct first, then guess $S$,  and there are exponentially many choices for each. The system $\outonlyintegersystemsat^\varphi_S$ contains exponentially many variables, and exponentially many constraints. Since solving such constraints is in $\np$ \cite{linearandintegerprogramming}, we can find a solution in $\nexp$. 
\end{proof}

\bigskip For the matching lower bound, 
we note that $\optwo$ subsumes the logic $\oMC$ which is shown to be $\nexp$-complete in~\cite{nexptimemodalglobalcounting}.
For completeness, we present its definition here.

\begin{definition}[$\oMC$]
    The syntax of \emph{outgoing-only two-variable modal logic with global counting} ($\oMC$), over a vocabulary $\tau$ with a single binary predicate $E$ and unary predicates $U_1,\dots,U_n$ , is built up according to the following rules:
    \begin{itemize}
        \item For every unary $U_i\in \tau$, the atom $U_i(x)$ is a formula.
        
        \item If $\varphi$ is a formula, then so is $\neg \varphi$.

        \item If $\varphi$ and $\psi$ are formulas, then so is $\varphi \wedge \psi$.
        
        \item If $\varphi(x)$ is a formula, then so is $\exists y\ [E(x,y) \wedge \varphi(y)]$ 

        \item If $\varphi(x)$ is a formula, then so is $\exists^{<k} x\ \varphi(x)$
    \end{itemize}
    The semantics of the Boolean connectives and existential quantifier is as usual. For the counting quantifier,
        $\cG\models \exists^{<k}x\ \varphi(x) 
        \ \iff \ 
        \left|\set{
            v\in \cG : \cG\models \varphi(v)
        }\right|  <k$
\end{definition}

\bigskip

It was shown in \cite{nexptimemodalglobalcounting} that finite satisfiability for $\oMC$ is $\nexp$-complete,
using a reduction from the exponential tiling problem.

\begin{theorem}\label{thm:omc}
    The finite satisfiability problem of $\oMC$ is $\nexp$-complete.
\end{theorem}

We conclude the section with the proof of the main result:

\begin{proof}[Proof (of \autoref{thm:outgoing_main})]
    By Lemma \ref{lem:outgoingnormal} we may consider a normal form sentence. The upper bound was shown in \autoref{thm:outgoing_upper}. 
    
    For the matching lower bound, observe that $\oMC$ is contained in $\ounmixedgptwo$. 
    In particular, the counting quantifier $\exists^{<k} x\ [\varphi (x)]$ in $\oMC$ may be expressed via:
    $\presbx{\varphi(x)}<k$.

\end{proof}